\documentclass[letterpaper,twocolumn,10pt]{article}
\usepackage{usenix2019,epsfig,endnotes}

\newcommand{\fun}[1]{\hbox{\it #1\/}}

\newcommand{\synbracket}[1]{[\![{#1}]\!]}

\newcommand{\DOM}[1]{{\it Domain\/}(#1)}

\def\elem{\hbox{\raise.13ex\hbox{$\scriptstyle\in$}}}

\newcommand{\INTEGER}{{\it Int}}

\newcommand{\IFHEAD}[1]{#1 \rightarrow}
\newcommand{\IFELSE}{[\!]~}

\newcommand{\CASEHEAD}[1]{\mbox{{\bf cases} ${#1}$ {\bf of} }}
\newcommand{\CASE}[2]{{#1} \:\rightarrow\: {#2}}

\newcommand{\CASEELSE}{\mbox{{\bf else} }}
\newcommand{\CASEEND}{{\bf end}}

\newcommand{\LETHEAD}[2]{\mbox{{\bf let} }{#1} = {#2}\mbox{ {\bf in} }}

\newcommand{\LEFTREASONBRACKET}{\langle}
\newcommand{\RIGHTREASONBRACKET}{\rangle}

\newcommand{\BYPHRASE}{by~}

\newenvironment{calculation}{
\samepage
\begin{tabbing}
n \= n \= n \= n \= mm \= n \= mm \= n \= mm \= n \= mm \= n \= mm \= \kill
}{
\end{tabbing}
}

\newcommand{\LASTFORMULA}[1]{\> \ensuremath{#1}}

\newcommand{\FORMULA}[1]{\LASTFORMULA{#1} \\}

\newcommand{\REASON}[2]{%
    \ensuremath{#1}\> \>
       \mbox{\ensuremath{\LEFTREASONBRACKET}{\BYPHRASE}#2\ensuremath{\RIGHTREASONBRACKET}} \\}

\usepackage{graphicx}
\usepackage[many]{tcolorbox}
\usepackage{mathtools,amssymb,latexsym,amsfonts,stmaryrd}
\usepackage{amsthm}
\usepackage{multirow}
\usepackage{colortbl}
\usepackage{tikz}
\usepackage{mathrsfs}
\usepackage{mathpartir}
\usepackage{url}
\usepackage{syntax}
\usepackage{framed}
\usepackage{float}
\usepackage{algorithm}
\usepackage{algpseudocode}
\usepackage{listings}
\usepackage{moresize}
\usepackage{wrapfig}
 
\usepackage{seqsplit}
\usepackage{alltt}
\usepackage{xspace}
\usepackage[caption=false]{subfig}
\usepackage{caption}
\usepackage{enumitem}
\usepackage{authblk}

\DeclareMathAlphabet{\mathcal}{OMS}{cmsy}{m}{n}

\usepackage{amsmath, listings, amsthm, amssymb, proof, xspace}

\usepackage{thmtools}

\declaretheoremstyle[spaceabove=\topsep,notefont=\normalfont\itshape]{mystyle}

\usepackage{ifpdf}
\ifpdf
\else
\fi

\usetikzlibrary{arrows,patterns, decorations.pathreplacing}
\newtheorem{definition}{Definition}
\newtheorem{claim}{Claim}
\newtheorem{theorem}{Theorem}[section]
\newtheorem{lemma}[theorem]{Lemma}

\newcommand{\RNum}[1]{\uppercase\expandafter{\romannumeral #1\relax}}
\newcommand*{\rom}[1]{\expandafter\@slowromancap\romannumeral #1@}

\newcommand{\F}{Fig.}
\newcommand{\T}{Table}
\renewcommand{\S}{Sec.}

\newcommand{\ignore}[1]{}

\lstdefinestyle{base}{
  moredelim=**[is][\color{red}]{@}{@},
  escapeinside={<@}{@>}
}

\newcommand{\rs}{\gg}

\newcommand{\cachea}{CacheAudit\xspace}
\newcommand{\cachediff}{CacheD\xspace}
\newcommand{\caches}{CacheS\xspace}

\newcommand{\ai}{abstract interpretation}

\newcommand{\Ai}{Abstract interpretation}

\newcommand{\SAI}{Secret-Augmented Abstract Interpretation}

\newcommand{\sas}{\textbf{SAS}}
\newcommand{\SAS}{\textbf{SAS}}

\newcommand{\powerset}[1]{\mathscr{P}(#1)}
\newcommand\scalemath[2]{\scalebox{#1}{\mbox{\ensuremath{\displaystyle #2}}}}

\newcommand{\INFOP}{\ensuremath{\hat{\Diamond}}}
\newcommand{\INFSTORE}{\ensuremath{\hat{\sigma}}}
\newcommand{\INFMEM}{\ensuremath{\hat{m}}}
\newcommand{\MEANINGFUN}[1]{{\cal #1}}
\newcommand{\Comm}{\SYNTAXCATEGORY{Command}}
\newcommand{\ABSV}{\ensuremath{\mathbf{AV}}}

\usepackage{amssymb}
\usepackage{pifont}

\newcommand\DejaVuttfamily{%
  \fontfamily{DejaVuSansMono-TLF}\selectfont
}

\lstdefinestyle{base}{
  moredelim=**[is][\color{red}]{@}{@},
  escapeinside={<@}{@>}
}

\definecolor{ForestGreen}{rgb}{0.13, 0.55, 0.13}
\definecolor{bittersweet}{rgb}{1.0, 0.44, 0.37}
\colorlet{myPurple}{blue!40!red}

\lstdefinelanguage
   [x64]{Assembler}     
   [x86masm]{Assembler} 
   {morekeywords={CDQE,CQO,CMPSQ,CMPXCHG16B,JRCXZ,LODSQ,MOVSXD, %
                  POPFQ,PUSHFQ,SCASQ,STOSQ,IRETQ,RDTSCP,SWAPGS, %
                  rax,rdx,rcx,rbx,rsi,rdi,rsp,rbp, %
                  r8,r8d,r8w,r8b,r9,r9d,r9w,r9b}} 

\usepackage{xcolor}
\colorlet{myPurple}{blue!40!red}
\definecolor{ForestGreen}{rgb}{0.13, 0.55, 0.13}
\newcommand{\code}[1]{\textcolor{myPurple}{\texttt{#1}}}
\newcommand{\comm}[1]{\textcolor{blue}{\texttt{#1}}}

\newcommand{\revise}[2]{{\color{red}{\ifx&#1&\else- #1\fi}} {\color{ForestGreen}{\ifx&#2&\else+ #2\fi}}}%
\renewcommand{\revise}[2]{#2}%

\newcommand{\Header}{\mathbb{U}}
\newcommand{\Stack}{\mathbb{S}}
\newcommand{\IM}{\mathbb{IM}}

\newcommand{\COLLCP}{\ensuremath{\textit{collc}_{t}}}
\newcommand{\COLLAP}{\ensuremath{\textit{colla}_{\hat{t}}}}

 \newcommand{\V}{\textit{Val}}

\newcommand{\MEM}{m}

\newcommand{\figref}[1]{Fig.~\ref{#1}}
\newcommand{\tabref}[1]{Table~\ref{#1}}
\newcommand{\lemref}[1]{Lemma~\ref{#1}}
\newcommand{\defref}[1]{Def.~\ref{#1}}
\newcommand{\thmref}[1]{Theorem~\ref{#1}}

\newcommand{\nonterm}[1]{\textrm{\textit{#1}}}

\newcommand{\DOMAIN}[1]{\fun{#1}}
\newcommand{\Store}{\DOMAIN{Store}}

\newcommand{\ME}[2]{\MEANINGFUN{E}\synbracket{ #1 }{(#2)}}

\newcommand{\LET}{\mbox{\bf{let}}}
\newcommand{\LETIN}{\mbox{\bf{in}}}
\renewcommand{\IFHEAD}[1]{\mbox{\bf{if}}~#1}

\newcommand{\IFTHEN}{\mbox{\bf{then}}~}
\renewcommand{\IFELSE}{\mbox{\bf{else}}~}

\newcommand{\MO}[1]{\MEANINGFUN{MO}\synbracket{ #1 }}

\newcommand{\MP}[2]{\MEANINGFUN{MP}\synbracket{#1}_{\theta}{(#2)}}

\newcommand{\MS}[2]{\MEANINGFUN{MS}\synbracket{#1}{(#2)}}

\newcommand{\SYNTAXCATEGORY}[1]{\nonterm{#1}}

\newcommand{\ROOT}[1]{\ensuremath{root(#1)}}

\newcommand{\Indexed}[2]{\ensuremath{#1 \!\!\downharpoonright\!\! #2}}

\newcommand{\APPROX}{\ensuremath{\approx}}
\newcommand{\TRANSIT}{\ensuremath{\rightsquigarrow}}

\newcommand{\Trace}{\textit{Trace}}

\newcommand{\Expr}{\SYNTAXCATEGORY{E}}

\newcommand{\INFPC}{\ensuremath{\hat{pc}}}
\newcommand{\INFME}[2]{\MEANINGFUN{\hat{E}}\synbracket{#1}{(#2)}}
\newcommand{\INFMS}[2]{\MEANINGFUN{\hat{MS}}\synbracket{#1}{(#2)}}
\newcommand{\INFMP}[2]{\MEANINGFUN{\hat{MP}}\synbracket{#1}{(#2)}}

\newcommand{\secref}[1]{Section~\ref{#1}}

\newcommand{\INFABS}{\ensuremath{\alpha}}

\newcommand{\INFCONCRETE}{\ensuremath{\gamma}}

\newcommand{\CODOM}{\fun{codom}}

\usepackage[compact]{titlesec}

  \titlespacing*{\subsection}
     {0pt}{0.10\baselineskip}{0.1\baselineskip}
  \titlespacing*{\section}
    {0pt}{0.10\baselineskip}{0.1\baselineskip}
   \titlespacing*{\subsubsection}
    {0pt}{0.1\baselineskip}{0.1\baselineskip}

\let\OLDthebibliography\thebibliography
\renewcommand\thebibliography[1]{
  \OLDthebibliography{#1}
  \setlength{\parskip}{0pt}
  \setlength{\itemsep}{1.0pt plus 0.15ex}
}


\begin{document}

\title{\large \bf Identifying Cache-Based Side Channels through \SAI}

\date{}

\author[1]{\rm Shuai Wang\thanks{Most of this work is done while Shuai Wang was working at PSU.}}
\author[2]{\rm Yuyan Bao}
\author[2]{\rm Xiao Liu}
\author[2]{\rm Pei Wang}
\author[2]{\rm Danfeng Zhang}
\author[2]{Dinghao Wu}
\affil[1]{The Hong Kong University of Science and Technology}
\affil[2]{The Pennsylvania State University} 
\affil[ ]{\rm \textit{shuaiw@cse.ust.hk, \{yxb88,  xvl5190, pxw172\}@ist.psu.edu, zhang@cse.psu.edu, dwu@ist.psu.edu}}
\maketitle

\thispagestyle{empty}

\subsection*{Abstract}
Cache-based side channels enable a dedicated attacker to reveal program secrets
by measuring the cache access patterns. Practical attacks have been shown
against real-world crypto algorithm implementations such as RSA, AES, and
ElGamal. By far, identifying information leaks due to cache-based side channels,
either in a static or dynamic manner, remains a challenge: the existing
approaches fail to offer high precision, full coverage, and good scalability
simultaneously, thus impeding their practical use in real-world scenarios.

In this paper, we propose a novel static analysis method on 
binaries to detect cache-based side channels. We use
\ai\ to reason on program states with respect to abstract values at each
program point. To make such \ai\ scalable to real-world cryptosystems while
offering high precision and full coverage, we propose a novel
abstract domain called the Secret-Augmented Symbolic domain (\sas). \sas\ tracks
program secrets and dependencies on them for precision, while it tracks only
coarse-grained public information for scalability.

We have implemented the proposed technique into a practical tool named
\caches\ and evaluated it on the implementations of widely-used cryptographic
algorithms in real-world crypto libraries, including Libgcrypt, OpenSSL, and
mbedTLS. \caches\ successfully confirmed a total of 154 information leaks
reported by previous research and 54 leaks that were previously unknown. We have
reported our findings to the developers. And they confirmed that many of those unknown
information leaks do lead to potential side channels.

\section{Introduction}
\label{sec:introduction}
Cache-based timing channels enable attackers to reveal secret program 
information, such as private keys, by measuring the runtime cache
behavior of the victim program. Practical attacks have been executed with 
different attack scenarios, such as
time-based~\cite{daniel2005cache,Kocher1996},
access-based~\cite{Gullasch:2011,Osvik+:2006,Percival:2005}, and
trace-based~\cite{aciicmez2006trace}, each of which exploits a victim program
through either coarse-grained or fine-grained monitoring of cache behavior.
Additionally, previous research has successfully launched attacks on commonly used
cryptographic algorithm implementations, for example,
AES~\cite{Gullasch:2011,Osvik+:2006, Tromer10,daniel2005cache},
RSA~\cite{Brumley:Boneh:2005,Kocher1996, Aciicmez07, Percival:2005, Yarom14},
and ElGamal~\cite{Zhang12}.

Pinpointing cache-based side channels from production cryptosystems remains a
challenge. Existing research employs either static or dynamic methods to detect
underlying
issues~\cite{wang2017cached,cacheaudit,Doychev16,gorka2017side,jan2018microwalk,bortzman2018casym,weiser2018data}.
However, the methods are limited to low detection coverage, low precision, and
poor scalability, which impede their usage in analyzing real-world cryptosystems
in the wild.

Abstract interpretation is a well-established framework that can be tuned to
balance precision and scalability for static analysis.
It models program execution within one or several
carefully-designed \textit{abstract domains}, which abstract program concrete semantics 
by tracking certain program states of interest in a concise representation. 
Usually, the elements in an abstract
domain form a complete lattice of finite height, and the operations of the program concrete semantics are
mapped to the abstract transfer functions over the abstract domain. A
well-designed \ai\ framework can correctly approximate program execution and usually
yields a terminating analysis within a finite step of computations.
Nevertheless, the art is to carefully design an abstraction domain that fits the
problem under consideration, while over-approximating others to bound the analysis
to a controllable size; this enables the analysis of non-trivial cases.

We propose a novel abstract domain named the Secret-Augmented Symbolic domain (\sas),
which is specifically designed to perform \ai\ on \textit{large-scale}
secret-aware software, such as real-world cryptosystems. 
\sas\ is designed to perform fine-grained tracking of 
program secrets (e.g., private keys) and dependencies on them, 
while coarsely approximating
non-secret information to speed up the convergence of the analysis.

We implement the proposed technique as a practical tool named \caches, which
models program execution within the \sas\ and pinpoints cache-based side
channels with constraint solving techniques. Like many bug finding
techniques~\cite{machiry2017checker,xie2005scalable,livshits2003tracking},
\caches\ is soundy~\cite{livshits2015soundness}; the implementation is unsound
for speeding up analysis and optimizing memory usage, due to its lightweight but
unsound treatment of memory. However, in contrast to previous studies that
analyze only small-size programs, single procedure or single execution trace
\cite{cacheaudit,Doychev16,wang2017cached,bortzman2018casym,weiser2018data},
\caches\ is scalable enough to deliver whole program static analysis of
real-world cryptosystems without sacrificing much accuracy. 
We have evaluated \caches\ on multiple popular crypto libraries. Although most
libraries have been checked by many previous tools, \caches\ is able to detect
54 unknown information leakage sites from the implementations of RSA/ElGamal
algorithms in three real-world cryptosystems: Libgcrypt (ver. 1.6.3), OpenSSL
(ver. 1.0.2k and 1.0.2f), and mbedTLS (ver. 2.5.1). We show that \caches\ has
good scalability as it largely outperforms previous research regarding coverage;
it is able to complete context-sensitive interprocedural analysis of over 295 K
lines of instructions within 0.5 CPU hour.
In summary, we make the following contributions:
\begin{itemize}[noitemsep,topsep=0pt]
\item We propose a novel \ai-based analysis to pinpoint information leakage
  sites that may lead to cache-based side channels. We propose a novel abstract
  domain named \sas, which performs fine-grained tracking of program secrets and
  dependencies, while over-approximating non-secret values to enable precise
  reasoning in a scalable way.
\item Enabled by the ``symbolic'' representation of abstract values in \sas, we
  facilitate information leak checking in this research with constraint solving techniques.
  Compared with previous \ai-based methods, which only reason on
  the information leakage upper-bound, our technique adequately
  simplifies the process of debugging and fixing side channels.
\item We implement the proposed technique into a practical tool named
  \caches\ and apply it to detect cache-based side channels in real-world
  cryptosystems. From five popular crypto library implementations,
  \caches\ successfully identified 208 information leakage sites (with only one
  false positive), among which 54 are unknown to previous research, to the best
  of our knowledge.
\end{itemize}

\section{Background}
\label{sec:background}
\newsavebox{\figa}
\begin{lrbox}{\figa}
\begin{minipage}[b]{.38\linewidth}
\begin{lstlisting}[basicstyle=\fontsize{6.0}{7.5}\ttfamily]
foo:
  <@\code{mov}@>   @eax@, @ebx@
  <@\code{add}@>   @eax@, 0x1
  <@\code{load}@>  ecx, esi
  <@\code{add}@>   ecx, 0x12
  <@\code{mov}@>   edx, edi
  <@\code{add}@>   @eax@, ecx
\end{lstlisting}
\end{minipage}
\end{lrbox}

\newsavebox{\figb}
\begin{lrbox}{\figb}
\begin{minipage}[b]{0.9\linewidth}
\begin{lstlisting}[basicstyle=\fontsize{6.0}{7.5}\ttfamily]
{<@ebx = $\{\mathtt{k_{1}}\}$@>}
{<@ebx = $\{\mathtt{k_{1}}\}$@>, <@eax = $\{\mathtt{k_{1}}\}$@>} 
{<@ebx = $\{\mathtt{k_{1}}\}$@>, <@eax = $\{\mathtt{k_{1}+1}\}$@>} 
{<@ebx = $\{\mathtt{k_{1}}\}$@>, <@eax = $\{\mathtt{k_{1}+1}\}$@>, <@ecx = $\{\mathtt{m_{1}}\}$@>} 
{<@ebx = $\{\mathtt{k_{1}}\}$@>, <@eax = $\{\mathtt{k_{1}+1}\}$@>, <@ecx = $\{\mathtt{m_{1}+12}\}$@>} 
{<@ebx = $\{\mathtt{k_{1}}\}$@>, <@eax = $\{\mathtt{k_{1}+1}\}$@>, <@ecx = $\{\mathtt{m_{1}+12}\}$@>, <@edx = $\{\mathtt{edi_{0}}\}$@>} 
{<@ebx = $\{\mathtt{k_{1}}\}$@>, <@eax = $\{\mathtt{k_{1}+ m_{1} +13}\}$@>, <@ecx = $\{\mathtt{m_{1}+12}\}$@>, <@edx = $\{\mathtt{edi_{0}}\}$@>} 
\end{lstlisting}
\end{minipage}
\end{lrbox}

\newsavebox{\figc}
\begin{lrbox}{\figc}
\begin{minipage}[b]{0.8\linewidth}
\begin{lstlisting}[basicstyle=\fontsize{6.0}{7.5}\ttfamily]
{<@ebx = $\{s_{1}\}$@>}
{<@ebx = $\{s_{1}\}$@>, <@eax = $\{s_{1}\}$@>} 
{<@ebx = $\{s_{1}\}$@>, <@eax = $\{s_{1}+1\}$@>} 
{<@ebx = $\{s_{1}\}$@>, <@eax = $\{s_{1}+1\}$@>, <@ecx = $\{p\}$@>} 
{<@ebx = $\{s_{1}\}$@>, <@eax = $\{s_{1}+1\}$@>, <@ecx = $\{p\}$@>} 
{<@ebx = $\{s_{1}\}$@>, <@eax = $\{s_{1}+1\}$@>, <@ecx = $\{p\}$@>, <@edx = $\{p\}$@>} 
{<@ebx = $\{s_{1}\}$@>, <@eax = $\{\top\}$@>, <@ecx = $\{p\}$@>, <@edx = $\{p\}$@>} 
\end{lstlisting}
\end{minipage}
\end{lrbox}

\begin{figure*}[!t]
\subfloat[\scriptsize Sample Code.]{\label{subfig:sample-code}\usebox{\figa}}%
\subfloat[\scriptsize Modeling program states with logic formulas $l \in L$.]{\label{subfig:psl}\usebox{\figb}}\hspace{5pt}%
\subfloat[\scriptsize Modeling program states with \sas.]{\label{subfig:pss}\usebox{\figc}}

\caption{Execute assembly code with different program representations. Program secrets and
all the affected registers are marked as \textcolor{red}{red}
in \F~\ref{subfig:sample-code}. Program states at line 1 of \F~\ref{subfig:psl}
and \F~\ref{subfig:pss} represent the initial state. Here $\mathtt{k_{1}}$ is a
symbol exhibiting one piece of program secrets (e.g., the first element in a key
array), and $\mathtt{m_{1}}$ is a free symbol representing non-secret content of
unknown memory cells. $\mathtt{edi_{0}}$ is a symbol representing the initial
value of register \texttt{edi}. Symbol $s_{1}$, $p$, and $\top$ defined in \SAS\
stand for one piece of secret, entire non-secret information and all the program
information, respectively (see \S~\ref{sec:sas}).}
\label{fig:motivation-modeling}
\end{figure*}

\noindent \textbf{Abstract Interpretation.}~Abstract interpretation is a
well-established framework to perform sound approximation of program
semantics~\cite{cousot77abstract}. Considering that program concrete semantics
forms a value domain $\mathbf{C}$, \ai\ maps $\mathbf{C}$ to an abstract (and
usually more concise) representation, namely, an abstract domain $\mathbf{A}$.
The design of the abstraction is usually based on certain program properties of
interest, and (possibly infinite) sets of concrete program states are usually
represented by one abstract state in $\mathbf{A}$. To ensure termination,
abstract states could form a lattice with a finite height, and computations of
program concrete semantics are mapped into operators over the abstract elements
in $\mathbf{A}$.

The abstract function ($\alpha$) and concretization function ($\gamma$) need to
be defined jointly with an abstract domain $\mathbf{A}$. Function $\alpha$ lifts
the elements in $\mathbf{C}$ to their corresponding abstract elements in
$\mathbf{A}$, while $\gamma$ casts an abstract value to a set of 
values in $\mathbf{C}$. To establish the correctness of an abstract
interpretation, the abstract domain and the concrete domain need to form a
Galois connection, and operators defined upon elements in an abstract domain are
required to form the local and global soundness notions~\cite{cousot77abstract}.

\noindent \textbf{Cache Structure and Cache-Based Timing Channels.}~A cache is a
fast on-CPU data storage unit with a very limited capacity compared to the main
memory. Caches are usually organized to be set-associative, meaning that the
storage is partitioned into several disjoint sets while each set exclusively
stores data of a particular part of the memory space. Each cache set can be
further divided into smaller storage units of equal size, namely cache lines.
\revise{}{Given the size of each cache line as $2^L$ bytes, usually the upper
  $N-L$ bits of a $N$-bit memory address uniquely locate a cache line where the
  data from that address will be temporally held.}

When the requested data is not found in the cache, the CPU will have to fetch
them from the main memory. This is called a cache miss and causes a significant
delay in execution, compared with fetching data directly from the cache.
Therefore, an attacker may utilize the timing difference to reveal the cache
access pattern and further infer any information on which this pattern may
depend.

\noindent \textbf{Threat Model.}~As mentioned above, some bits of a memory
address can be directly mapped to cache lines being visited, which potentially
enables information leakage via \textit{secret-dependent memory traffic}. In
this research, attackers are assumed to share the same hardware platform with
the victim program, and therefore are able to ``probe'' the shared cache state
and infer cache lines being accessed by the victim. As illustrated in
\F~\ref{fig:threat-model}, \revise{}{our threat model assumes that the attacker
  can observe the address of every memory access, expect for the low bits of
  addresses that distinguish locations in the same cache line.} Overall, by
tracking the secret-dependent cache access of the victim, several bits of
program secrets (w.r.t. entropy) could be leaked to the attacker.

We note that this threat model indeed captures most infamous and practical side
channel attacks~\cite{he2017side}, including prime-and-probe~\cite{Osvik+:2006},
flush-and-reload~\cite{Yarom14}, and prime-and-abort~\cite{disselkoen2017prime},
which are designed to infer the cache line access by measuring the latency of
the victim program or attacker's program at different scales and for different
attack scenarios. Additionally, while this threat model is aligned with many
existing side channel detection
works~\cite{wang2017cached,Doychev16,gorka2017side,jan2018microwalk,bortzman2018casym},
novel techniques proposed in this work enable us to perform scalable static
analysis and reveal much more information leaks of real-world
cryptosystems.\footnote{Consistent with this line of research,
  \caches\ pinpoints information leaks in cryptosystems where cache access
  depends on secrets. Cryptosystem developers can fix the code with information
  provided by \caches. \revise{}{Contrarily, the exploitability of the leaks
    (e.g., reconstruct the entire key by recovering half bits of the RSA private
    key~\cite{boneh1998attack}) is beyond the scope of this work.}} In addition,
while this model is relatively stronger than those based on cache
status~\cite{cacheaudit}, cache status at any point can be determined by
analyzing the accessed cache units in execution.

\underline{\revise{Control-Flow Side Channels.}{}} \revise{Some side channels are caused by
secret-dependent control flows: different secrets detour the execution at
branches, which may lead to variants of executed instructions or execution time.
In general, control-flow side channels are caused by program behavior variants
between different branches, and discovering such differences requires analyzing
two branches simultaneously. Benefitting from the whole-program static analysis
(see \S~\ref{sec:design}), \caches\ is perfectly suitable for analyzing
multi-branch issues. Although this work focuses on subtler leakages via
cache access and does not detect control-flow side channels (since the latter case
has been well studied~\cite{Agat00, Hedin:Sands:2005, Molnar:2005,Barthe:2006,
  coppens2009practical,multirun}), we emphasize that it is straightforward to
pinpoint control-flow side channels by extending \caches.}{}

\begin{figure}[t]
\centering
\includegraphics[width=0.60\linewidth]{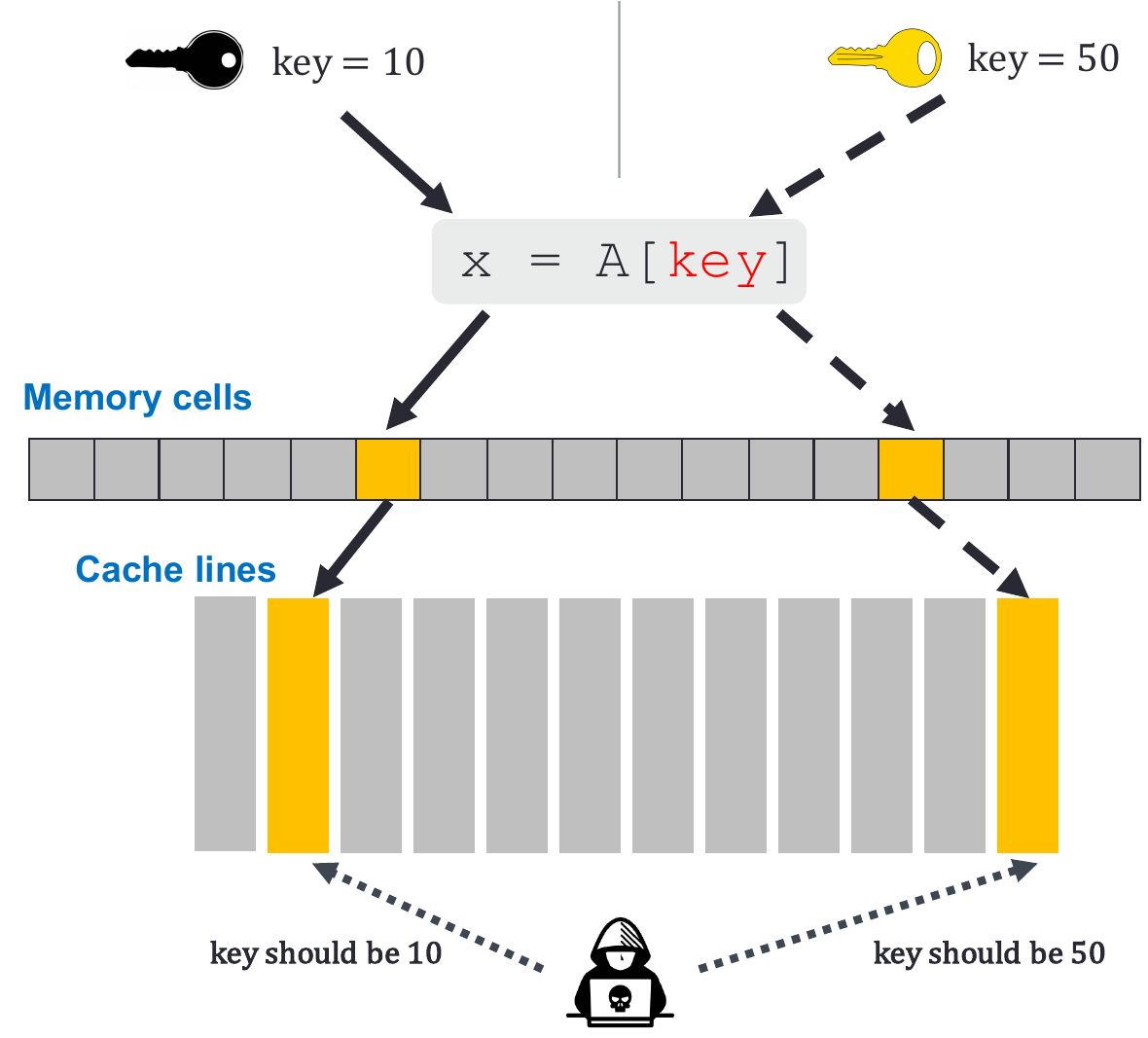}
\caption{The threat model. Different secrets lead to the access of different
  cache lines at one particular program point, which may leak secret information
  to the attackers by indirectly observing cache line access variants. At least
  one bit information (w.r.t. entropy) could be leaked in this example. }
\label{fig:threat-model}
\end{figure}

\section{Motivation}
\label{sec:motivation}
In general, capturing cache-based side channels requires modeling program
secret-dependent semantics (we will discuss the connection between program
semantics and cache access in \S~\ref{sec:check}). In this section we
begin by discussing two baseline approaches to modeling program semantics;
the limitations of both approaches naturally motivate the design of our novel
abstract domain.

\noindent \textbf{Modeling Program Semantics with Logic Formulas.}~An intuitive
way is to represent program concrete semantics with logic formulas (as in a typical
symbolic execution approach~\cite{wang2017cached}), and perform whole-program
static reasoning until a fixed point is reached. The overall workflow exhibits a
typical dataflow analysis procedure, and upon termination, each program point
maintains a program state that maps variables (i.e., registers, CPU flags, and
memory cells) to sets of formulas representing the possible values each variable
may hold regarding any execution paths and inputs. For ease of presentation,
we name the value domain formed by logic formulas $l$ as logic domain $L$.

An example is given in \F~\ref{fig:motivation-modeling}, where we model the
execution of instructions with logic formulas (\F~\ref{subfig:psl}). 
While the overall approach will precisely model program semantics,
some tentative studies indicate its low scalability. Indeed, we implement this
approach and evaluate it with two real-world cases: the AES and RSA
implementations of OpenSSL. We report that both tests are unable to terminate 
(evaluation results are given in \S~\ref{sec:evaluation}). In summary, the
analysis is impeded for the following reasons:

\begin{itemize}[noitemsep,topsep=0pt]
  \item Typically, more and more memory cells would be modeled throughout the
    analysis, and for each variable, its value set (i.e., set of formulas) would 
    also continue to increase. Therefore, the memory usage
    could become significant to even unrealistic for real-world cases.
  \item Program states could be continuously updated within loop iterations. In
    addition, ``loops'' on the call graph (e.g., recursive calls) could exist in
    cryptosystems as well and complicate the analysis.
\end{itemize}

We implement algorithms to detect loop induction variables~\cite{appel2004tiger}
considering both registers and stack memories. Identified induction variables are lifted
into a linear function of symbolic loop iterators; operations on induction
variables are ``merged'' into the linear function, thereby leading to a stable
stage. While the simpler AES case terminated when we re-ran the test, the RSA
case still yielded a ``timeout'' due to the practical challenges mentioned above
(see results in \S~\ref{subsec:evaluation-results}).

\noindent \textbf{Modeling Program Semantics with Free Symbols.}~Another
``baseline'' approach is to model program semantics in a permissive way. That
is, we introduce two free symbols: one for any public information and the other
for secrets. Any secret-related computation outputs the same secret symbol,
while others preserve the same public symbol. Note that this is comparable to
static taint tracking, where each value is either ``tainted'' or ``untainted''.
Despite its simplicity, our tentative study reveals new hurdles as follows:

\begin{itemize}[noitemsep,topsep=0pt]
  \item Memory tracking becomes pointless. Every memory address becomes
    (syntactically) identical because it holds the same public or secret
    symbol. Therefore, a memory store could overturn the entire memory space.
\item Even if memory addresses are tracked in a more precise way, representing
  any secret value and their dependencies coarsely as one free secret symbol yields many
  false positives (since secret-dependent memory
  accesses do not necessarily lead to vulnerable cache accesses; see
  our cache modeling in \S~\ref{sec:check}). Tentative tests of the AES case
  report a false positive rate of 20\% (8 out of 40) due to such modeling. In
  contrast, our novel program modeling yields no false positive when testing
  this case (see \S~\ref{sec:evaluation}).
\end{itemize}

\noindent \textbf{Motivation of Our Approach.}~This paper presents a novel abstract domain
that enables \ai\ of
large-scale cryptosystems in the wild. Our observation is that imprecise tracking of
secrets impedes the accurate modeling of cache behaviors
(cache access modeling is discussed in \S~\ref{sec:check}). 
Nevertheless, tracking too much information, such
as modeling whole-program semantics with logic formulas, could face scalability issues when analyzing
real-world cryptosystems due to various practical challenges.

Our study of real-world cryptosystems actually reveals an interesting and intuitive finding.
That is, program secrets and their dependencies usually exhibit at a very
\textit{small portion} of program points, and even in such secret-carrying
points, most variables maintain \textit{only public information}. It should be noted that in
common scenarios non-secret information is not critical for modeling cache-based
timing channels.
Hence, based on our observation, we promote a novel abstract domain that is particularly
designed to \textit{model the secret-dependent semantics of real-world crypto systems}.
Our abstract domain delivers fine-grained tracking of 
program secrets and their dependencies
with different identifiers for each piece of secret information,
while performing coarse-grained tracking of other public values to effectively enhance scalability.

\section{Secret-Augmented Symbolic Domain}
\label{sec:sas}

This section presents the definition of our abstract domain \sas. We formally
define each component following convention, including the concrete semantics,
the abstract domain, and the abstract transfer functions. We also prove that the
computations specified in \sas\ correctly over-approximate concrete semantics.
Due to space limitations, we highlight only certain necessary components in
this section. Please refer to Appendix~\ref{sec:formalization-proof} for a
complete formulation.

\subsection{Abstract Values}
\label{subsec:formula-sas}
\begin{figure}[!t]
$
\begin{array}{l@{\ }l@{\ }c@{\ }l}
\mathsf{Literal} & n &\in &\mathbb{Z} \\
\mathsf{OP_1} & \oplus &::= &+ \mid - \\
\mathsf{OP_2} & \otimes &::= &\times \mid \div \mid \% \mid \mathrm{AND} \mid \mathrm{OR} \mid \mathrm{XOR} \mid \mathrm{SHIFT} \\
\mathsf{Atom} & t &::=\; &\top \mid p \mid s_{i} \mid n \\
\mathsf{Expression} & exp &::=\; &t \mid t \oplus exp \mid t \otimes exp \\
\mathsf{Formula} & f &::=\; &e \mid exp \mid e \oplus exp \\
\end{array}
$
\caption{Syntax of abstract value.}
\label{fig:syntax}
\end{figure}

We start by defining abstract values $f \in \mathbf{AV}$ (soon we will show that
\sas\ is defined as the powerset of $\mathbf{AV}$). Comparable to ``symbolic
formulas'' in symbolic execution, $f$ combines symbols and constants via
operators. Elementary symbols in each abstract value are defined as follows:

\begin{itemize}[noitemsep,topsep=0pt]
  \item $p$: a unique symbol representing all the program public information.
  \item $s_{i}$: a symbol representing a piece of program
    secrets; for instance, the i-th element of a secret array.
  \item $e$: a unique symbol representing the initial value of the
    x86 stack register \texttt{esp}.
\end{itemize}

While only one free symbol $p$ is used to represent any and all
unknown non-secret information (e.g., initial value $\mathtt{edi_{0}}$ of
register \texttt{edi} in \F~\ref{subfig:psl}), we retain finer-grained
information about program secrets. Multiple $s_{i}$ are generated, and are
mapped to different pieces of program secrets (e.g., a symbol $s_{1}$
representing $\mathtt{k_{1}}$ in \F~\ref{subfig:pss}). Therefore, different
$s_{i}$ symbols are \textit{semantically different}, meaning each of
them stands for different secrets.

\noindent \textbf{Syntax.}~The syntax of a core of abstract values $f \in
\mathbf{AV}$ is defined in \F~\ref{fig:syntax}. $\mathsf{Literal}$ specifies
that concrete data is preserved in $\mathbf{AV}$. $\mathsf{OP_1}$ and
$\mathsf{OP_2}$ explain typical operators in $\mathbf{AV}$.
\revise{}{$\mathsf{Atom}$ includes symbols and literals}, among which $\top$
(top) is the abstraction of any concrete value. $\mathsf{Expression}$ and
$\mathsf{Formula}$ additionally define expressions and formulas. Note that stack
memory expands linearly in the process address space, and stack register
\texttt{esp} at any program point shall hold a value which adds or subtracts an
offset from the initial value of \texttt{esp} (i.e., $e$). In the syntax
definition, stack memory offsets could be a constant or an $exp$.

Since the symbol $\{s_i\}$ represents the secrets, which our
analysis intends to keep track of, the formulas that contain these symbols
usually need to be specially treated.
We denote this infinite set of special formulas by $\mathbf{AV}_{s}$, where
$\mathbf{AV}_{s} = \{f \in \mathbf{AV} \mid \exists s\in \{s_i\}~\text{s.t.}~s~\text{occurs in}~f\}$.

\noindent \textbf{Reduction of Abstract Formulas.}~We now define the 
operator semantics of abstract value $f \in \mathbf{AV}$. For any operator
$\odot\in \{\oplus\}\cup \{\otimes\}$, we define a reduction rule $T_\odot:
\mathbf{AV} \times \mathbf{AV} \rightarrow \mathbf{AV}$ such that $\llbracket
a_1\odot a_2 \rrbracket = T_\odot(\llbracket a_1 \rrbracket, \llbracket a_2 \rrbracket)$ for any $a_1, a_2 \in \mathbf{AV}$, where
$\llbracket \cdot \rrbracket$ denotes the semantics.
We then define $T_\odot(a_1, a_2)$ as follows:
\[
\scalemath{1.0}{
T_\odot(a_{1}, a_{2}) = 
   \begin{cases}
      \top        & \textrm{if}\; a_{1} = \top\;\textrm{or}\;a_2=\top\\
      \top        & \parbox[t]{.8\linewidth}{else if $a_{1} = p\land a_2\in \mathbf{AV}_{s}$ or\\ $a_{2} = p\land a_1 \in \mathbf{AV}_{s}$}\\
      p           & \parbox[t]{.8\linewidth}{\textrm{else if} $a_{1} = p \land a_2 \notin \mathbf{AV}_{s}$ or\\ $a_{2} = p\land a_1\notin \mathbf{AV}_{s}$}\\
      a_{1}\odot a_{2} & \textrm{otherwise}\\
   \end{cases}
}
\]

Essentially, the first three cases perform reasonable over-approximation on
$f \in \mathbf{AV}$ with different degrees of abstraction. The last case would apply if no
other case can be matched; indeed similar to symbolic execution, 
most operations on $f \in \mathbf{AV}$ 
``concatenates'' abstract values via abstract operators following this rule. For
the implementation, we also implement ``constant folding'' rules for
operands of concrete data; such rules help the reduction of
stack increment and decrement operations.

Since abstract interpretation typically needs to process sets of facts, we
extend $T_{\odot}$ so that it can be applied to pairs of subsets of abstract
values $f \in \mathbf{AV}$, where
\begin{align*}
  \forall X, Y \in \powerset{\mathbf{AV}},&\forall \odot \in \{\oplus\}\cup\{\otimes\}, \\
  &T_\odot(X, Y)  = \{ T_\odot(a, b) \mid a \in X, b \in Y\}
\end{align*}

\subsection{Abstract Domain}
\label{subsec:abs-lattice}
Naturally, each element in \sas\ represents the possible values 
that a program variable may hold; therefore each element in \sas\ forms a 
set of abstract values. That is,
\begin{definition}
Let $\mathbf{AV}$ be the set of abstract values. Then
\[
    \textbf{\rm\bf SAS} = \powerset{\mathbf{AV}}
\]
forms a domain whose elements are subsets of all valid abstract values.
\end{definition}

\begin{claim}
\sas\ forms a lattice, with the top element $\top_{\textbf{\rm\bf SAS}}$, bottom
element $\bot_{\textbf{\rm\bf SAS}}$ and a join operator $\sqcup$ defined over
$\textbf{SAS}$.
\end{claim}
For further discussion and definition of \sas, please refer to
Appendix~\ref{sec:formalization-proof}.

\underline{Example.}~~\F~\ref{fig:motivation-modeling} explains typical
computations within \sas. We present a set of abstract values for each register
in \F~\ref{subfig:pss}. While the computations over secret symbol $s_{i}$ are
precisely tracked (line 3 in \F~\ref{subfig:pss}), the computations over $p$
preserve this symbol (line 5 in \F~\ref{subfig:pss}), and the computations between
abstract value $a \in \mathbf{AV}_{s}$ and $p$ lead to $\top$ (line 7 in
\F~\ref{subfig:pss}).

\begin{figure*}
\centering
\includegraphics[width=0.95\linewidth]{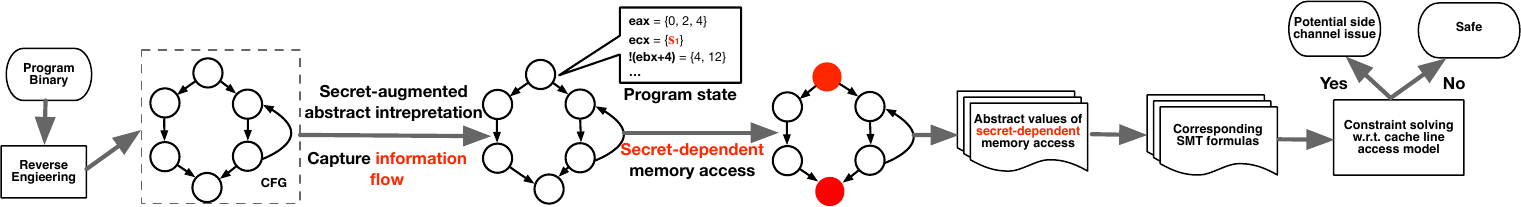}
\caption{The overall workflow of \caches.}
\label{fig:caches}
\end{figure*}

\section{Pinpointing Information Leakage Sites}
\label{sec:check}
Upon the termination of static analysis, we check abstract memory addresses of each
memory load and store instruction. When a secret-dependent address $a \in
\mathbf{AV}_{s}$ is identified, its corresponding memory access instruction is
considered to be ``secret-dependent.'' We then translate each secret-dependent
address $a$ into an SMT formula $f$ for constraint checking (this translation is
discussed in \S~\ref{subsec:design-smt}).

In this research, we adopt a cache model proposed by the existing work to check each
secret-dependent memory access~\cite{wang2017cached}. Given an SMT formula $f$
translated from $a \in \mathbf{AV}_{s}$ that represents a memory address,
\caches\ checks potential cache line access variants by solving the
satisfiability of the following predicate:

\begin{equation} \label{eq:formula}
  f \rs L \neq f[s'_{i} / s_i] \rs L
\end{equation}

As discussed in \S~\ref{sec:background}, assuming the cache has the line size of
$2^L$ bytes, for a memory address of $N$ bits, the upper $N-L$ bits map a memory
access to its corresponding cache line access. In other words, the upper $N-L$
bits decide which cache line the upcoming memory access would visit. Therefore,
for an SMT formula $f$ derived from $a \in \mathbf{AV}_s$, we right shift $f$ by
$L$ bits, and the result $f\rs L$ indicates the cache line being accessed.
Furthermore, by replacing each $s_{i}$ with a fresh secret symbol $s'_{i}$, we
obtain $f[s'_{i} / s_{i}] \rs L$. As a standard setting, the cache line
size is assumed to be 64 ($2^6$) in this work; therefore, we set $L$ as 6.

The constructed constraint checks whether different secrets ($s_{i}$ and
$s_{i}'$) can lead to the access of different cache lines at this memory access.
Recall the threat model shown in \F~\ref{fig:threat-model}, the existence of at
least one satisfiable solution reveals potential side channels at this point.
From an attacker's perspective, by (indirectly) observing the access of
different cache lines, a certain number of secrets could be leaked to
adversaries. \revise{}{In addition, while this constraint assumes that accesses
  to different offsets within cache lines are indistinguishable,
  Constraint~\ref{eq:formula} can be extended to detect related issues. For
  example, information leaks which enable cache bank attacks can be detected by
  changing $L$ from 6 to 2~\cite{yarom2016cachebleed}.}

\section{Design of \caches}
\label{sec:design}
We now present \caches, a tool that uses precise and scalable static analysis to
detect cache-based timing channels in real-world cryptosystems.
\F~\ref{fig:caches} presents the workflow of \caches. Given a binary as the
input, \caches\ first leverages a reverse engineering tool to recover the
assembly code and the control flow structures from the input. The assembly
instructions are further lifted into \textit{platform-independent}
representations before analysis. Technical details on reverse engineering are
discussed in \S~\ref{sec:implementation}.

Given all the recovered program information, we initialize the abstract program
state at each program point. In particular, we update the
initial state of certain program points with one or several ``secret'' symbols to 
represent program secrets (e.g., a sequence of memory cells) 
when the analysis starts. We then perform \ai\ on the whole
program until the fixed point in \SAS\ is reached.

\Ai\ reasons the program execution within \sas\ (\S~\ref{subsec:design-ai}), 
and as mentioned, the proposed abstract domain performs fine-grained 
tracking of program secret-related semantics while maintaining only coarse-grained
public information for scalability. The entire analysis framework forms a
standard worklist algorithm, where each program point maintains its own program
state mapping variables to sets of abstract values (\S~\ref{subsubsec:program-state}).

We define information flow rules to propagate secret information
(\S~\ref{subsec:design-information}) in our context-sensitive and
interprocedural analysis (\S~\ref{subsec:context-sensitivity}). Upon the termination of
analyzing one function, we identify secret-dependent memory
accesses and translate corresponding memory addressing formulas into SMT
formulas (\S~\ref{subsec:design-smt}) and check for side channels
(\S~\ref{sec:check}).

\noindent \textbf{Application Scope.} In this research we design our abstract
domain \sas\ to analyze assembly code: program memory access can be accurately
uncovered by analyzing assembly code, thus supporting a ``down-to-earth''
modeling of cache behavior (see \S~\ref{sec:check}).

To assist the analysis of off-the-shelf cryptosystems and capture information
leaks in the wild, we designed \caches\ to directly process binary executables,
including stripped executables with no debug or relocation information. We rely
on reverse engineering tools to recover program control structures from the
input binary, and further build our analysis framework on top of that (see
\S~\ref{sec:implementation}).

\subsection{Abstract Interpretation}
\label{subsec:design-ai}
In this section, we discuss how the proposed abstract
domain \sas\ is adopted in our tool, and elaborate on several key points to 
deliver a practical and scalable analysis.

\subsubsection{Initialization} 

Before the analysis, we first initialize certain program points with
$\{s_{i}\}$ to represent the initial secret program information; for
the rest their corresponding initial states are naturally defined as
$\{ \}$, or $\{e\}$ for the stack register \texttt{esp}.

Program secrets are maintained in registers or memory cells (e.g.,
on the stack) during execution. Since \caches\ is designed to directly
analyze binary code, we must first recognize the location of program secrets.
We reverse-engineer the input binary and mark
the location of secrets manually. Once the locations of secrets are flagged,
we update the initial value set of corresponding variables (i.e., registers
or memory cells) with a secret symbol $s_{i}$. Additionally, while ``manual
reverse engineering'' is sufficient for studies in this research,
it is always feasible to leverage automatic techniques~\cite{calvet2012aligot}
to search for secrets directly from executables or secret-aware compilers to
track secret locations when source code is available. We leave this to future work.

In addition, since program secrets may be stored in a region of sequential
memory cells (e.g., in an array), we create another identifier named $u$ to
represent the base address of the secret memory region. While $u$ itself is
treated as \textit{public information}, we specify that memory loading from $u$
will obtain program secrets; that is, we introduce one $s_{i}$ for
each memory loading via $u$.

\subsubsection{Program State}
\label{subsubsec:program-state}

\begin{figure}[t]
\centering
\includegraphics[width=0.85\linewidth]{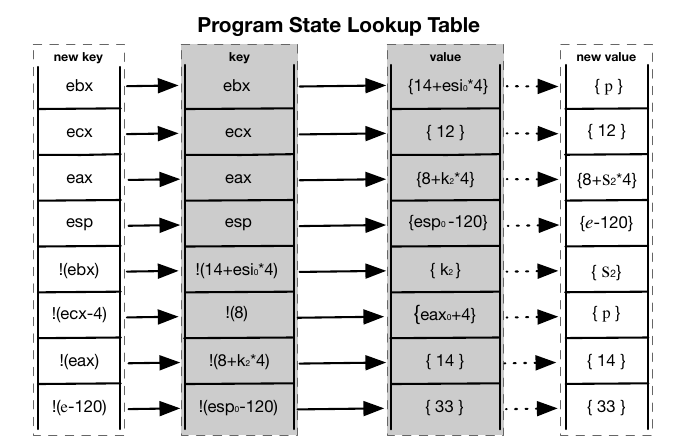}
\caption{A sample program state lookup table. 
  $\mathtt{esp_{0}}$, $\mathtt{eax_{0}}$ and $\mathtt{esi_{0}}$ in the ``key''
  and ``value'' entries are symbols representing the register initial values.
  Symbol $\mathtt{!}$ means pointer dereference, for example $\mathtt{!(eax)}$
  means memory loading from address stored in \texttt{eax}. Lookup tables at
  each program point are the major factor for memory usage, and we optimize the
  design by replacing ``key'' and ``value'' columns with ``new key'' and ``new
  value'' columns, respectively (see \S~\ref{subsubsec:program-state}). Hence,
  shaded boxes are eliminated in \caches.}
\label{fig:program-state} 
\end{figure}

At each program point, \caches\ maintains a lookup table that maps variables
to value sets; each value set $S \in \textbf{SAS}$ consists of abstract values $f \in \mathbf{AV}$
representing possible values of a variable at the current program point.
While the ``lookup table'' is an essential piece of any non-trivial
analysis framework, our study has shown that 
naively-designed program state representations in \caches\ could 
consume significant amounts of computing resources and impede the analysis of non-trivial
programs. Thus, at this step we seek to design a \textit{concise}
and \textit{practical} representation of program states. For the rest of this
section, we first explain a ``baseline'' implementation of the lookup table,
and further discuss two refinements.

\noindent \textbf{The ``Baseline'' Approach.}~A sample lookup table is shown in \F~\ref{fig:program-state} (the ``key'' and
``value'' columns), where each table maps registers and memory addressing
formulas to their corresponding sets for logic formulas $l \in \mathbf{L}$. When
it encounters a memory access instruction, \caches\ computes the memory
addressing formula and searches for its existence in the lookup table. (This
requires some ``equivalence checking''; the details will be explained
in \S~\ref{par:memory-model}). If the search identifies an entry in the
table, \caches\ extracts or updates the content of that entry accordingly.
Consider the example in Listing~\ref{lst:sample}, where we first store
concrete data 14 into memory via address stored in \texttt{eax}, and further
load it out into \texttt{ebx}.

\begin{lstlisting}[caption={Sample instructions.},label={lst:sample},basicstyle=\ttfamily\small,xleftmargin=-.017\textwidth,numbers=none]
                  <@\code{store}@> eax, 14
                  <@\code{load}@>  ebx, eax
\end{lstlisting}

\noindent Knowing that value set of \texttt{eax} is 8+$\mathtt{k_{2}}$*4 (third entry
in \F~\ref{fig:program-state}), the first instruction creates an entry from
address 8+$\mathtt{k_{2}}$*4 to {14} (\F~\ref{fig:program-state} shows program
states after executing the first instruction). Further memory loading would
acquire the value set in \texttt{eax}, and then reset the entry
of \texttt{ebx} with {14} in the state lookup table of the second instruction.

Reading from unknown registers and memory locations would introduce symbols of different
credentials regarding our information flow policy (see \S~\ref{subsec:design-information} for details).

\noindent \textbf{Optimization of Table Values.}~While the precisely
tracked logic formulas $l \in \mathbf{L}$ result in 
notable computing resource usage (\S~\ref{sec:motivation}), 
the proposed abstract domain \SAS\ (\S~\ref{sec:sas}) enables 
succinct representation of abstract values. As shown in \F~\ref{fig:program-state}, the ``value'' column of the
lookup table is now replaced by the ``new value'' column. Consequently, memory
consumption is considerably reduced (details are reported in our evaluation section).

\noindent \textbf{Optimization of Table Keys.}~Since only abstract values are traced in \sas, the ``key'' column can be updated
into a compact representation as well. However, using symbols such as $p$ as the
key will result in an imprecise modeling of memory addresses.

\caches optimizes the ``key'' column in the following way. For most memory
related entries, instead of using abstract memory addressing formulas, memory access
expressions (expressions of registers and constant offsets) are used as keys. For example, the
first instruction in Listing~\ref{lst:sample} uses memory access
expression (i.e., ``eax'') instead of its abstract value $8+s_{2}*4$
for memory lookup. Hence, when analyzing the store instruction, \caches\
creates (or updates) an entry in the lookup table, which uses !(eax) as the key
(here symbol ``!'' means pointer dereference). Likewise, for memory load, 
!(eax) will be used to look up the program state table.
To safely preserve lookup entries via expressions, whenever the value set of a
register is reset in the analysis, entries in the table are
deleted if their keys are memory access expressions via the newly-updated
register. 

Nevertheless, since stack register \texttt{esp} is frequently
manipulated to access stack memory, we preserve abstract addressing formulas via 
$e$ to keep track of stack memory access precisely (e.g., the last entry in
the ``new key'' column of \F~\ref{fig:program-state}).

\subsubsection{Order of Program State}
\label{par:order-program-state}
When multiple program states are possible for a program point, it is important
to define the ``merge'' operation in abstract interpretation. Such an operation
can be defined based on the least upper bound operation $\sqcup$ of \sas\
(recall that \sas\ forms a lattice (\S~\ref{subsec:abs-lattice})).

Given lookup tables $T_{1}$ and $T_{2}$ representing two program states,
$T_{1}\sqcup T_{2}$ is defined as the following table, say $T_3$:
\begin{itemize}[noitemsep,topsep=0pt]
  \item $T_3$'s key set is the union of the key sets of $T_{1}$ and $T_{2}$;
  \item For each key $k$ in $T_3$, $T_3[k]=T_1[k]\sqcup T_2[k]$ (assuming
  $T_1[k]$ or $T_2[k]$ is an empty set if $k$ is not in the table). 
\end{itemize}

Moreover, the least upper bound of program states entails the partial order of any two
program states: $T_1 \sqcup T_2 = T_1 \leftrightarrow T_2 \subseteq T_1$.
\revise{Naturally, $T_{2}$ is treated as less than or equal to $T_{1}$ if the
reverses of these two conditions hold. Otherwise, no order is defined between
these two lookup tables.}{}

\subsubsection{Memory Model}
\label{par:memory-model}
When encountering a memory load and store operation, we must decide which
memory cell is accessed by tracing the memory address. However,
considering \caches\ models program semantics with abstract values, a
memory address can usually contain one or several symbols instead of only
concrete data. Therefore, policies (i.e., a ``memory model'') are usually
required to determine the location of an accessed memory cell given a symbolic
pointer.

When defining the abstract semantics within \sas\
(see Appendix~\ref{sec:formalization-proof}), 
we assume the assistance of a sound points-to 
analysis module as pre-knowledge. Nevertheless, 
finding such a convenient tool for assembly code of large-scale cryptosystems is 
quite difficult in practice. We have tried several popular  ``end-to-end'' binary analysis 
platforms that take an executable as 
the input and perform various reverse engineering campaigns including points-to analysis; 
nevertheless, so far we cannot find a practical and robust solution to our scenario.

Therefore, we aim to implement a rigorous memory model by solving 
the \textit{equality constraints} of two abstract formulas. However,
tentative tests show that such a memory model may
lose considerable precision in terms of reasoning symbolic pointers and may also not be scalable enough.
On the other hand, since keys in the memory lookup table are
formulas of $e$ (for stack pointers; recall that $e$ represents the initial value of \texttt{esp}) or memory access expressions
(for other pointers), the current implementation of \caches\ rigorously reasons on
the equality constraints if abstract values are composed of 
$e$ and concrete offsets, which is indeed often the case in analyzing assembly code.
For the rest (e.g., $e$ and symbolic offsets), 
we reason on the syntactical equivalence of memory access expressions. 
This design tradeoff may incorrectly deem equivalent symbolic pointers inequivalent (due to 
the symbolic
``alias'' issue) but not vice versa. Experiments show that this memory model is
efficient enough to handle real-world cryptosystems while being
promisingly accurate.

\subsection{Information Flow}
\label{subsec:design-information}
Considering that information leaks detected in this research are derived from
secret-dependent memory accesses, \caches\ keeps track of the secret program
information flow throughout the analysis. In this section we elaborate on cases
where the secret information can be propagated.

\noindent \textbf{Variable-Level Information Flow.}~The explicit information
flow is modeled in a straightforward way. Since variables (i.e., registers,
memory cells, and CPU flags) are modeled as abstract formulas, high credential
information (exhibited as abstract value $f \in \mathbf{AV}_{s}$) would
naturally ``flow'' among variables during the computations. 
Moreover, reading from unknown variables (those with empty value sets) generates
a symbol $p$ as a proper over-approximation.

\noindent \textbf{Information Flow via Memory Loading.}~By knowing the underlying memory
layout, it could be feasible to infer table lookup indexes by observing the
memory load outputs, hence leaking table indexes of secrets to attackers. It should
be noted that such cases are not rare in real-world cryptosystems,
where many precomputed data structures are deployed in the memory to speed up
computations. Thus, we define policies to capture information flow
through memory loading. To do so, for a load operation, whenever the value
sets of its base address or memory offset include formula $f \in \mathbf{AV}_{s}$, 
\caches\ assigns the memory content to a fresh $s_{i}$, 
indicating secret information could have potentially
propagated to the value being read. In contrast, when loading from unknown 
memory cells (memory cells of empty value sets) via non-secret addresses, we create
a $p$ to update the memory reader.

While most memory addressing formulas refer to specific locations in the memory,
symbols $p$ and $\top$ represent any program (public) information. To safely
approximate memory read access via $p$ and $\top$, \caches\ assigns $\top$ to
the memory reader. In case a memory storing is via symbol $p$ or $\top$, we
terminate the analysis since this would rewrite the whole memory space.
Additionally, we note that memory loading and storing via $\top$ are considered
to be information leaks as well since $\top$ implies that a variable has certain
residual secrets (see \S~\ref{subsec:design-smt}).

\subsection{Interprocedural Analysis}
\label{subsec:context-sensitivity}
Our interprocedural analysis is context-sensitive. \revise{}{We build a classic
  function summary-based interprocedural analysis framework, where a summary
  $(\langle f', i\rangle, o)$ of a function call towards $f$ maps the calling
  context $\langle f', i\rangle$ ($f'$ is the caller name and $i$ is the input)
  to the function call output $o$.} \caches\ maintains a set of summaries for
each function $f$, and for an upcoming call of $f$, its calling context is first
checked regarding the existing summaries of $f$. In case the context is a subset
of any recorded entries (the partial order of calling context is derived from
the order of program states defined in \S~\ref{subsubsec:program-state}), the
analysis will be skipped and we directly return the corresponding output.

To recover the function inputs, we inquire the employed reverse engineering
platform (details are given in \S~\ref{sec:implementation}) to obtain the number
of parameters the approaching function has. According to the calling convention
of 32-bit x86 platforms, a memory stack is used to store function parameters;
thus, we construct stack memory addresses of function parameters and acquire the
value set of each parameter from the program state lookup table at the call
site. If some memory cells of function parameters are absent, symbol $p$ is used
as an over-approximation.
To compute the output information of a function, we join program states at every
return instruction when the analysis of the target function terminates, which
over-approximates the function return states.

\subsection{Translating Abstract Values into SMT Formulas}
\label{subsec:design-smt}
As noted earlier (\S~\ref{sec:check}), cache-access side channels are summarized
into SMT constraints. Upon the termination of analyzing each function, we
identify secret-dependent memory addresses $a \in \mathbf{AV}_{s}$ and build the
side channel constraints.
SMT solvers are used to solve the constraint and check whether different secrets
can lead to cache line access variants. Nevertheless, while many works to date
leverage symbolic execution to construct SMT formulas, here we reason on program
states within \sas. Therefore, before constraint checking, we first translate
abstract formulas into SMT formulas.

Each abstract formula is maintained as a symbolic ``tree'' in \caches,
where tree leaves are symbols and concrete data while other nodes are operators.
At this step, we translate each leaf on the tree into a bit vector implemented by 
a widely-used SMT solver---Z3~\cite{DeMoura:2008:ZES:1792734.1792766}; a bit vector
would be instantialized with a numeric value if it was derived from a constant.
In addition, we translate abstract operators on the tree into bit vector operations in Z3.
Hence, an abstract formula tree would be reduced bottom-up into an SMT formula.

\noindent \textbf{Translate Secret Symbols into Unique Bit Vectors.}~As noted earlier,
$s_{i}$ symbols are semantically different, each of
which represents different pieces of secrets. For the implementation, we assign a
unique id for each newly-created $s_{i}$ symbol, which further leads
to the creation of unique bit vectors at this step. In contrast, $p$ (and $e$) symbols 
are transformed into identical bit vectors.

\noindent \textbf{Memory Access via $\top$.}~It is easy to see that $\top$
implies that a variable has some residual secrets along with possibly public
information. Hence, in addition to checking the constructed SMT constraints with Z3,
memory accesses are flagged as vulnerable whenever their
corresponding addressing formulas are $\top$.

\section{Implementation}
\label{sec:implementation}
\caches\ is mainly written in Scala (in 6,764 LOC; counted by CLOC~\cite{cloc}).
The tentative implementation (in 7,163 LOC), which models program semantics with
logic formulas (\S~\ref{sec:motivation}), is maintained as a separate ``branch''
of the code base\@.

Starting from an input binary code, the first step is to recover the assembly
program as well as control flow and call graphs. Here we employ a popular
reverse engineering tool, IDA-Pro (version 6.9) for the reverse engineering
task~\cite{IDAPro}. We use the default configurations of IDA-Pro to recover
assembly code and program control structures from the input executables.

\noindent \textbf{Assembly Lifting.}~Many existing binary analysis
infrastructures have provided facilities to lift x86 assembly code into a
high-level intermediate representation. Without reinventing the wheel, here we
employ a well-developed binary analysis platform \textsc{BinNavi}~\cite{binnavi}
to transform x86 assembly code into a \textit{platform-independent} intermediate
language, REIL~\cite{reil}. Our analysis procedures are built on top of the
recovered representations. In addition, for a formal definition of program
concrete semantics in terms of the REIL language, please refer to
Appendix~\ref{sec:formalization-proof}.

The current implementation of \caches\ analyzes ELF binaries on the x86
platform. Nevertheless, since REIL language is designed as
\textit{platform-independent}, there is no fundamental limitation for
\caches\ to analyze binaries of other formats or from other platforms (e.g., PE
binaries on Windows) as long as the assembly instructions can be translated into
REIL statements. As aforementioned, our current prototype focuses on 32-bit ELF
binaries since the state-of-the-art REIL lifter (BinNavi~\cite{binnavi}) does
not have an official support for 64-bit binaries. However, the proposed
technique shall be applicable to 64-bit binaries with no additional technical
hurdles.

\noindent \textbf{Recover x86 Memory Access Instructions from REIL
  Statements.}~As noted in \S~\ref{subsubsec:program-state}, we use memory
access expressions instead of address formulas as the key to simplify the memory
lookup. While the memory access expressions can be acquired by checking assembly
instructions, note that our analysis is launched on REIL IR; one memory access
instruction is extended into multiple IR statements. Hence, we perform def-use
analysis to ``collapse'' IR statements belonging to the same instruction and
recover the corresponding memory access expression.

\noindent \textbf{Critical Functions.}~\caches\ is designed to perform both
inter and intra-procedural analysis on any binary code component. For the
evaluations in this research, instead of starting from the program entry point,
analyses were launched on critical functions of cryptosystems that have become
the target for many previous attacks. Such critical functions are the starting
points of our interprocedural analysis, and we recursively discover all the
reachable functions on the call graph. As reported in our evaluation (see
\T~\ref{tab:accuracy}), these recursively collected functions usually form a
non-trivial subgraph on the program call graph. \revise{}{In addition, taking
  these critical functions as the starting points of \caches\ makes it easier to
  compare our findings with existing work.}

\section{Evaluation}
\label{sec:evaluation}
In contrast to many previous studies in which cache-based side channels are
detected from only simple cases, \caches\ is evaluated on several real-world
cryptosystems.
As reported in \T~\ref{tab:plan}, three cryptosystems are evaluated in this
research. OpenSSL and Libgcrypt are widely used cryptosystems on multi-purpose
computers, while mbedTLS is commonly adopted by embedded devices. Eight critical
functions are selected as the starting point of our analysis, which covers major
security-sensitive components in three crypto algorithm implementations: RSA,
AES, and ElGamal.

To prepare \caches\ inputs, we compile test programs shipped in each
cryptosystem and link with the corresponding libraries. All the crypto libraries
are written in C. We build each library and test program into a 32-bit ELF
binary on Ubuntu 12.04 with gcc compiler (version 4.6.3).

\begin{table}[t]
  \centering
  \caption{Cryptosystems analyzed by \caches.}
  \label{tab:plan}
  \resizebox{\linewidth}{!}{
  \begin{tabular}{c@{~}|@{~}c@{~}c@{~}c}
    \hline
    \multirow{2}{*}{\textbf{Implementation}} & \multirow{2}{*}{\textbf{Versions}} & \textbf{Analysis} & \textbf{Implement} \\
    & & \textbf{Starting Function} & \textbf{Which Algorithm} \\
    \hline
     Libgcrypt~\cite{libgcrypt} & 1.6.1, 1.7.3 & \_gcry\_mpi\_powm  & \multirow{2}{*}{RSA/ElGamal} \\
     OpenSSL~\cite{openssl}   & 1.0.2f, 1.0.2k & BN\_mod\_exp\_mont\_consttime & \\
    \hline
     mbedTLS~\cite{mbedtls}   & 2.5.1 & mbedtls\_mpi\_exp\_mod & RSA \\
    \hline
    OpenSSL~\cite{openssl} & 1.0.2f, 1.0.2k & \_x86\_AES\_decrypt\_compact & \multirow{2}{*}{AES} \\
    mbedTLS~\cite{mbedtls} & 2.5.1 &  mbedtls\_internal\_aes\_decrypt & \\
    \hline
  \end{tabular}
  }
\end{table}

\begin{table*}[t]
  \caption{Evaluation result overview. We compare the identified information
    leakage sites by \caches\ with a recent research
    (\cachediff~\cite{wang2017cached}), and we report \caches\ can identify all
    the leakage sites reported by \cachediff. A summary of all leaks can be
    found in \T~\ref{tab:sec-dep-cache-2} in the appendix.}
  \label{tab:accuracy}
  \centering
  \resizebox{\linewidth}{!}{
\begin{tabular}{c@{~}|@{~}c@{~}|@{~}c@{~}|@{~}c@{~}|@{~}c@{~}|@{~}c@{~}|@{~}c@{~}|@{~}c@{~}|@{~}c||@{~}c|@{~}c|@{~}c}
    \hline
\multirow{2}{*}{\textbf{Algorithm}} & \multirow{2}{*}{\textbf{Implementation}} & \textbf{Information Leakage} & \textbf{\# of Analyzed}  & \textbf{\# of Analyzed} & \textbf{Processing Time} & \textbf{\# of Processed} & \textbf{Peak Memory} & \textbf{Information} & \multicolumn{3}{c}{\textbf{Results Reported in CacheD~\cite{wang2017cached}}}\\\cline{10-12}
   & & \textbf{Sites (known/unknown)} & \textbf{Procedures} & \textbf{Contexts} & \textbf{(CPU Seconds)} & \textbf{REIL Instructions} & \textbf{Usage (MB)} & \textbf{Leakage Units} & \textbf{Leakage Sites} & \textbf{Processing Time} & \textbf{Leakage Units} \\
    \hline
    \textbf{RSA/ElGamal} & Libgcrypt 1.6.1  & 22/18 & 60 & 81  & 228.8  & 50,436  & 7,749 & 11  & 22 &  14293.6 & 5 \\
    \textbf{RSA/ElGamal} & Libgcrypt 1.7.3  & 0/0   & 59 & 59  & 182.2  & 33,386  & 5,823 & 0  & 0 & 11626.0 & 0 \\
    \textbf{RSA/ElGamal} & OpenSSL   1.0.2k & 2/3   & 71 & 81  & 179.2  & 83,183  & 6,134 & 2  & N/A & N/A & N/A \\
    \textbf{RSA/ElGamal} & OpenSSL   1.0.2f & 2/4   & 68 & 72  & 169.5  & 80,096  & 6,113 & 3 & 2 & 165.6 & 2 \\
    \textbf{RSA} & mbedTLS   2.5.1          & 0/29  & 29 & 36  & 775.9  & 35,963  & 9,654 & 2  & N/A & N/A & N/A \\
    \hline
    \textbf{AES}         & OpenSSL   1.0.2k & 32/0  &  1 & 1   & 33.2   & 3,748   & 620  & 1  & N/A & N/A & N/A \\
    \textbf{AES}         & OpenSSL   1.0.2f & 32/0  &  1 & 1   & 35.8   & 3,748   & 578  & 1  & 32 & 48.5 & 1 \\
    \textbf{AES}         & mbedTLS   2.5.1  & 64/0  &  1 & 1   & 32.8   & 4,803   & 619  & 1 & N/A & N/A & N/A \\
    \hline
    \textbf{Total} & & 154/54 & 290 & 332 & 1,637.4 & 295,363 & 37,290 & 21 & 56 & 26,133.7 & 8 \\
    \hline
  \end{tabular}
 }
\end{table*}

\begin{table}[t]
  \centering \caption{Model program semantics in the logic formulas $l \in \mathbf{L}$ and 
  \SAS\ and test OpenSSL 1.0.2k. The second and third rows report the modeling
  results with logic formulas, while the last row reports results in
  \SAS. The comparison of these two program modelings is given in \S~\ref{sec:motivation}.}
  \label{tab:comparison}
  \resizebox{\linewidth}{!}{
  \begin{tabular}{c@{~}|@{~}c@{~}|@{~}c@{~}|@{~}c@{~}|@{~}c@{~}|@{~}c}
    \hline
     \multirow{2}{*}{\textbf{Algorithm}} & \textbf{Execution Time} & \textbf{\# of Processed} & \textbf{\# of Processed} & \textbf{Peak Memory} & \textbf{Detected} \\
     & \textbf{(CPU Second)} & \textbf{Function} & \textbf{Context} & \textbf{Usage (MB)} & \textbf{Leaks} \\
    \hline
    \textbf{RSA/ElGamal} & timeout ($>$ 5 CPU hours) &       15 &      28 & 7,283           & N/A \\
    \textbf{AES}         & timeout ($>$ 5 CPU hours) &        1 &       1 & 47,798          & N/A \\
    \hline
    \textbf{RSA/ElGamal} & timeout ($>$ 5 CPU hours) &       28 &      85 & 53,054         & N/A \\
    \textbf{AES}         & 115.8                   &        1 &       1 & 621            & 32 \\
    \hline
    \textbf{RSA/ElGamal} & 179.2                   &       71 &      81 & 6,134          & 5 \\
    \textbf{AES}         & 33.2                    &        1 &       1 & 620            & 32 \\
    \hline
  \end{tabular}
  }
\end{table}

\subsection{Evaluation Result Overview}
\label{subsec:evaluation-results}
\T~\ref{tab:accuracy} presents the evaluation result overview. In summary, 208
information leak points are reported from the real world cryptosystems evaluated
in this research. We interpret the results as promising; most of the evaluated
cryptosystems contain information leaks due to cache-based side channels, and
\caches\ helps to pinpoint these leaks with program-wide static analysis.

It is commonly acknowledged that the table lookup implementation of the AES
decryption routine is vulnerable to various real-world cache attacks.
\caches\ identifies 32 information leaks from the AES
implementations of OpenSSL (versions 1.0.2f and 1.0.2k), and 64 leaks from
mbedTLS. Indeed, all of these issues are lookup table queries via direct usages
of secrets, which is consistent with findings in existing
research~\cite{7467359,wang2017cached}.

Existing research has pinpointed multiple information leaks in the modular
exponentiation implementation of OpenSSL and
Libgcrypt~\cite{wang2017cached,Liu15}; vulnerable functions are adopted by both
RSA and ElGamal for decryption. \caches\ confirmed these findings (see
\S~\ref{subsec:eval-rsa} for one false positive in OpenSSL).
Furthermore, \caches\ successfully revealed a much larger information leakage
surface than existing trace and static analysis based techniques, because of its
scalable modeling of program semantics. \T~\ref{tab:accuracy} shows that
\caches\ identifies more information leaks from Libgcrypt and OpenSSL in
addition to confirming all issues reported by \cachediff~\cite{wang2017cached}.
Moreover, \caches\ identifies multiple information leakage sites from the
modular exponentiation implementation of mbedTLS, which, to the best of our
knowledge, is unknown to the research community.

While 40 information leakage sites are reported in Libgcrypt (version 1.6.1),
our study shows that they have been fixed in version 1.7.3. Without
secret-dependent memory accesses, the RSA/ElGamal implementation of Libgcrypt
1.7.3 is generally accepted as safe regarding our threat model. Our evaluation
reports consistent findings that no leak is detected regarding our threat model
on secret-dependent cache-line accesses (\revise{}{but we do find
  secret-dependent control flows, see \S~\ref{subsec:secret-control}}).

\noindent \textbf{Computing Resource.}~Our evaluation is launched on a 
machine with 2.90 GHz Intel Xeon(R) E5-2690 CPU and 128 GB memory. 
For each context-sensitive analysis campaign, \T~\ref{tab:accuracy} 
presents the covered functions, contexts, and processed IR instructions. 
We report that \caches\ takes less than 1700 CPU seconds to process all
the test cases, and on average the peak memory usage to evaluate one case is 
less than 5 GB. 
Overall, \caches\ finished all the analysis campaigns with 
reasonable amount of computing resources, and we interpret that the promising results
demonstrate the high scalability of \caches\ in analyzing real-world cryptosystems.

\noindent \textbf{Modeling Program Semantics with Logic Formulas.}~As noted in
\S~\ref{sec:motivation}, we tentatively implement the idea of modeling program
concrete semantics with logic formulas. Note that in addition to the semantics
modeling, all the design and evaluation settings are unchanged.

The first two rows of \T~\ref{tab:comparison} give the evaluation results
for the AES and RSA/ElGamal implementations in OpenSSL 1.0.2k, both
of which report a ``timeout'' after 5 CPU hours. As explained in
\S~\ref{sec:motivation}, we extend the prototype with loop induction
variable detection, and the third row reports the results of the
re-launched tests. Still, the RSA/ElGamal case throws a timeout (a reflection on 
this tentative evaluation is given in
\S~\ref{sec:motivation}). In summary, we interpret that the \SAS\
proposed in this research has largely improved the analysis scalability, which
serves as an indispensable component to pinpoint cache-based timing channels in
real-world cryptosystems.

\noindent \textbf{\revise{}{Comparison with \cachea.}}\footnote{\url{https://github.com/cacheaudit/cacheaudit}}
Besides
\cachediff~\cite{wang2017cached}, we also compare our results with
\cachea~\cite{cacheaudit}.
\cachea\ failed on 
all of our test cases for two
reasons. First, two of
our cases contain some x86 instructions that are not handled by \cachea. Second,
\cachea\ refuses to analyze indirect function calls when constructing the
control flow graph. In addition, we also describe the
key differences between \caches\ and \cachea\ in \S~\ref{sec:related}.

\noindent \textbf{\revise{}{Identifying Information Leakage Units.}}~Considering
some occurrences of information leaks are on adjacent lines of a code component
(a summary of all leaks can be found at 
\T~\ref{tab:sec-dep-cache-2} in the appendix), once a
leak is flagged by \caches, presumably any competent programmer shall spot and
remove all the related defects. Therefore, we group the flagged information
leaks to assess the utility of \caches and also estimate the bug fixing effort.
Though it can be slightly subjective, we propose a metric according to the
source code locations of defects: information leaks will be grouped together as
a ``leakage unit'' if they are within the same or adjacent C statements (e.g.,
within the same loop or adjacent if branches). Also, if a macro is expanded at
different program points (e.g., the macro \texttt{MPN\_COPY} which contains
information leaks in Libgcrypt 1.6.1), we count it only once.

As reported in \T~\ref{tab:accuracy}, \caches\ identified 21 units of
information leaks. We also grouped the findings of \cachediff\ with the same
metric. We have confirmed that \caches\ covered all leakage units reported in
\cachediff, and further revealed new leakage units within statements or
functions not covered by \cachediff (e.g., 6 new leakage units in Libgcrypt
1.6.1). Overall, we interpret the evaluation results as promising; trace-based
analysis, like \cachediff, is incapable of modeling the program collecting
semantics, and therefore underestimates the attack surface.

\noindent \textbf{Confirmation with Library Authors.}~As shown in
\T~\ref{tab:accuracy}, we found unknown information leaks from OpenSSL (versions
1.0.2f and 1.0.2k) and mbedTLS (version 2.5.1). Our findings were reported and
promptly confirmed by the OpenSSL developers~\cite{opensslvul}; the latest
OpenSSL has been patched to eliminate these leaks (the leaks are discussed
shortly in \S~\ref{subsec:eval-rsa}). At the time of writing, we are waiting for
responses from the mbedTLS developers.

\subsection{Exploring the Leaks in mbedTLS}
\label{subsec:verification}
Although mbedTLS developers have not confirmed our findings, we conduct further
study of the 29 flagged information leakage sites from this library to check
whether they can lead to cache-based side channels.

As mentioned above, the constraint solver provides at least one pair of
satisfiable solutions (a pair of secrets $k$ and $k'$) to each leakage site
(\S~\ref{sec:check}). To verify one leak, we instrument the program source code
and modify secrets with $k$ and $k'$. We then compile the instrumented programs
into two binaries and monitor the execution of each binary executable via a
widely-used hardware simulator (\texttt{gem5}~\cite{gem5}). The compiled code is
fed with test cases shipped with the cryptosystems, and we use the full-system
simulation mode of \texttt{gem5} to monitor the execution of the instrumented
program. The full-system simulation mode uses 64-bit Ubuntu 12.04 (this mode
only supports 64-bit OS) to host the application code. We compile the
instrumented source code into 64-bit binaries since executing 32-bit binaries on
the 64-bit OS throws some TLB translation exceptions (this issue is also
reported in \cite{wang2017cached}). The configuration of \texttt{gem5} is
reported in \T~\ref{tab:gem5}. At the leakage point, we intercept the cache
access from CPU to L1 Data Cache; the accessed cache line and corresponding
cache status (hit vs.\ miss) are recorded.

\begin{table}[t]
  \centering
  \caption{\texttt{gem5} configurations.}
  \label{tab:gem5}
\resizebox{0.8\linewidth}{!}{
  \begin{tabular}{c|c}
    \hline
    \textbf{ISA} & x86 \\
    \hline
    \textbf{Processor type} & single core, out-of-order \\
    \hline
    \textbf{L1 Cache} & 4-way, 32KB, 2-cycle latency \\
    \hline
    \textbf{L2 Cache} & 8-way, 1MB, 50-cycle latency \\
    \hline
    \textbf{Cache line size} & 64 Bytes \\
    \hline
    \textbf{Cache replacement policy} & LRU \\
    \hline
  \end{tabular}
  }
\end{table}

\begin{table}[t]
  \centering
  \caption{Hardware simulation results.}
  \label{tab:verification-results}
\resizebox{\linewidth}{!}{
  \begin{tabular}{cccc}
    \hline
    \textbf{\# of \caches\ Detected} & \textbf{\# of Executed} & \textbf{Cache Line} & \textbf{Cache Status} \\
     \textbf{Leakage Sites} & \textbf{Leakage Sites} & \textbf{Access Variants} & \textbf{Variants} \\
   \hline
   29 & 14 & 14 & 6 \\
   \hline
  \end{tabular}
  }
\end{table}

As shown in \T~\ref{tab:verification-results}, among 29 information leakage
sites found in mbedTLS, 14 sites are covered during simulation. We observe that
different cache lines are accessed at these leakage points, when instrumenting
the program with secrets $k$ and $k'$. In other words, by observing the access
of different cache lines, attackers will be able to infer a certain amount of
secret information. In addition, we report that cache status variants (in terms
of cache hit vs.\ miss) are observed in several cases. In summary, we interpret
the verification results as highly promising; we have confirmed that all the
executed information leakages are \textit{true positives} since cache line
access variants are observed.

Although the employed program inputs cannot lead to the full coverage of every
leakage site, we manually checked all the uncovered cases, and we found that
these cases share the same pattern as the covered leaks. For instance, the
covered and uncovered leaks are the same inline assembly sequences residing
within different paths. Overall, we interpret it as convincing to conclude that
all the detected information leaks in mbedTLS are true positives.

\subsection{Case Study of Leaks in mbedTLS}
\label{subsec:evaluation-case}

\begin{figure*}[!ht]
  \centering
  \includegraphics[width=0.95\linewidth]{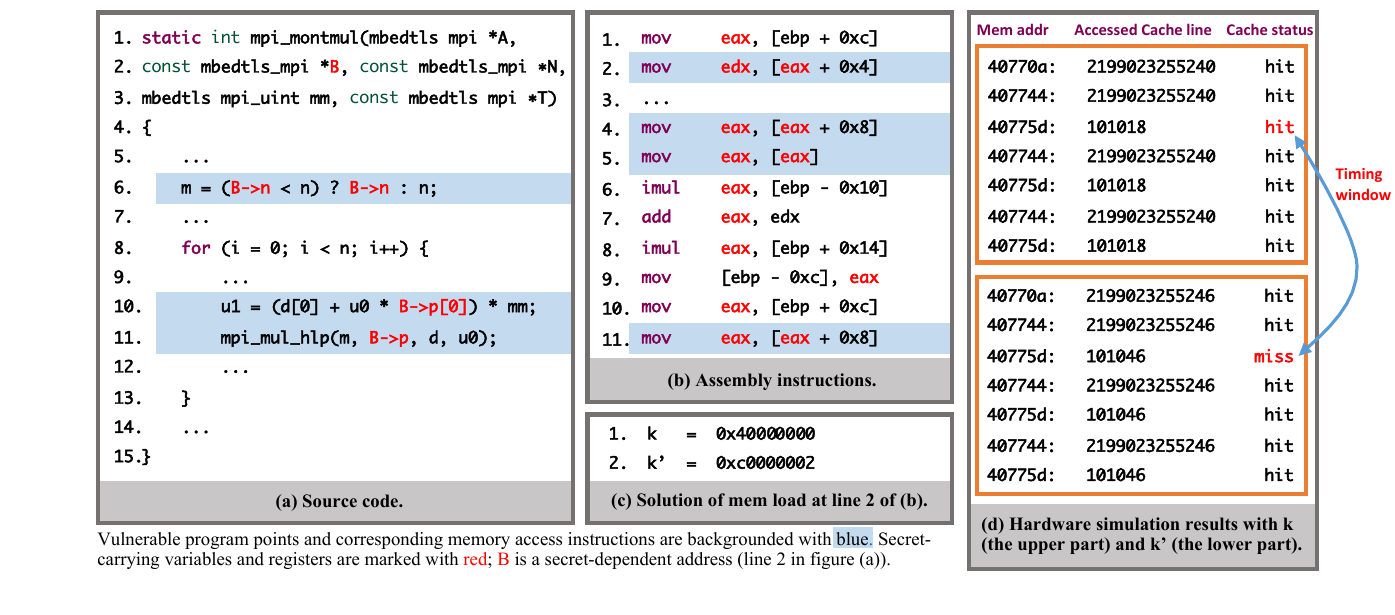}
  \caption{Case study of information leaks in mbedTLS. The constraint solver
    finds a pair of secrets (\texttt{k} and \texttt{k'}) which leads to the
    access of different cache lines at line 2 of (b).}
  \label{fig:rsa-case-study}
\end{figure*}

This section presents a thorough case study of several information leaks
identified by our tool. As presented in \T~\ref{tab:accuracy}, we identified 29
information leakage points in mbedTLS 2.5.1. In particular, the first four leaks
were found in the function \texttt{mpi\_montmul} (source code is given in
\F~\ref{fig:rsa-case-study}(a)), which is a major component of the modular
exponentiation implementation in mbedTLS. The value of function parameter
\texttt{B} is derived from a window size of the secret key (line 2). In
\texttt{mpi\_montmul}, \texttt{B} is used as a pointer to access elements in a C
struct (line 6, line 10, line 11). We envision that different program secrets
would derive into different values of \texttt{B}, which further lead to the
access of different cache lines in secret-dependent memory accesses.

The evaluation shows consistent findings. As shown in \F~\ref{fig:rsa-case-study}(b), 
\caches\ identifies four suspicious memory
accesses in \texttt{mpi\_montmul} (two pointer dereferences at line 6
of \F~\ref{fig:rsa-case-study}(a) are optimized into one memory load 
at line 2 of \F~\ref{fig:rsa-case-study}(b)). By checking the constraint solver,
we find a pair of program secrets that affect the value of \texttt{B} and
further lead to the access of different cache lines at the first memory access
(the solution is given in \F~\ref{fig:rsa-case-study}(c)).

We then instrument the program private key with the solver provided solutions in
\F~\ref{fig:rsa-case-study}(c) and observe the runtime cache access within
\texttt{gem5}. This secret pair is generated by analyzing the first leakage
memory access, but since variants of \texttt{B} may affect the following memory
traffic as well, we report the cache status at all the suspicious memory
accesses in \texttt{mpi\_montmul}. We note that while \caches\ analyzes 32-bit
binaries, at this step we compile the instrumented source code into 64-bit
binaries since the simulated OS throws some exceptions when running 32-bit code.
After compilation, the five leakage points in the source code actually produce
three memory load instructions in the 64-bit assembly code. Cache behaviors,
including the accessed cache line and the corresponding cache status, are
recorded at these points. \F~\ref{fig:rsa-case-study}(d) presents the simulation
results. Due to the limited space, we provide only the first seven records
(59568 records in total). Program counters 0x40770a, 0x407744 and 0x40775d
represent the three identified memory loads of information leaks. It is easy to
see that different cache lines are accessed at each point. Additionally, a
timing window of one cache hit vs.\ miss is found (this memory access represents
a table lookup in the first element of \texttt{B->p}).

\subsection{Information Leaks in the Modular Exponentiation Algorithm}
\label{subsec:eval-rsa}
Both RSA and ElGamal algorithms employ the modular exponentiation algorithm for
decryption. Existing research has reported that such an algorithm is vulnerable
to cache-based timing channel attacks~\cite{wang2017cached,Liu15}. Here, we
evaluate the corresponding implementations in OpenSSL, Libgcrypt, and mbedTLS.
As reported in \T~\ref{tab:accuracy}, \caches\ successfully revealed a much
larger leakage surface, including 80 (54 unknown and 26 known) information leaks,
from our test cases.

\noindent \textbf{Information Leaks in Libgcrypt.}~A large number of leakage
points are reported from the sliding window-based modular exponentiation
implementation in Libgcrypt 1.6.1. Existing research has pointed out the direct
usages of (window-size) secret keys as \textit{exploitable}~\cite{Liu15}, and
\caches\ pinpointed this issue. \revise{}{In addition to the 4 direct usages of
  secrets, we  further
  uncovered 36 leaks due to the propagation of secret information flows,
  as \caches\ keeps track of both variable-level and memory loading
  based information flows (\S~\ref{subsec:design-information}). }

While previous trace-based analysis also keeps track of information flow
propagation (i.e., \cachediff~\cite{wang2017cached}), \caches\ still outperforms
\cachediff\ because of its program-wide analysis. With the help of \cachediff's
authors, we confirmed that \caches\ can detect all 22 leaks reported in
\cachediff~\cite{wang2017cached}, and further reveals 18 additional points.

\noindent \textbf{Information Leaks in mbedTLS.}~\caches\ has also identified
leaks in another commonly used cryptosystem, mbedTLS. 
Appendix~\ref{sec:unknown-mbedtls}
presents several leaks found in the mbedTLS case. In general, function
\texttt{mbedtls\_mpi\_exp\_mod} implements a sliding window-based modular
exponentiation, which leads to secret-dependent memory accesses (precomputed
table lookup). The table lookup statement (line 10) does not generate a leak
point since it only gets a pointer referring to an array element, however,
further memory dereferences on the acquired pointer reveal \revise{}{4 direct
  usages of secrets (discussed in \S~\ref{subsec:evaluation-case}). We also find
  25 leaks due to the propagation of secret information flows
  (\S~\ref{subsec:design-information}).}

We note that mbedTLS uses RSA exponent blinding as a
countermeasure~\cite{Kocher:1996}, which practically introduces noise and
mitigates cache side channels (but is still exploitable with enough collisions
or if the attacker can derive the exponent from a single trace). We leave it as
future work to model the feasibility of exploitations, given the identified
leaks and also taking program randomness (e.g., exponent blinding) into
consideration.

\noindent \textbf{Information Leaks in OpenSSL.}~Information leaks in OpenSSL
are within or derived from functions counting the length of a secret array.
\F~\ref{fig:rsa-openssl} presents a function that contains 4 memory accesses
which engender information leaks: \texttt{BIGNUM} maintains an array of 32-bit
elements, and \texttt{BN\_num\_bits} counts the number of bits within a given
big number. Since the last element of the array may be less than 32 bits, a
lookup table is used to determine the exact bits in the last element of the
secret array (function call at line 7 in \F~\ref{fig:rsa-openssl}). By storing
secrets within a big number structure, table queries in
\texttt{BN\_num\_bits\_word} could lead to secret-dependent memory accesses.

While \cachediff~\cite{wang2017cached} flags only one information leak (line 26
in \F~\ref{fig:rsa-openssl}) covered by its execution trace, \caches\ detects
more leaks. \revise{}{As shown in \F~\ref{fig:rsa-openssl}, four table queries
  are analyzed by \caches, and all of them are flagged as leaks. In addition to
  these four direct usages of secrets, \caches\ also finds one more leak in
  OpenSSL 1.0.2k, and two in OpenSSL 1.0.2f, both are due to the propagation of
  secret information flows.}

\noindent \underline{False Positive.}~In addition to the issues in
Appendix~\ref{sec:unknown-openssl} that have been confirmed and fixed by the OpenSSL
developers, we also find one \textit{false positive} when analyzing OpenSSL
1.0.2f. To defeat side channel attacks against the precomputed table lookup,
OpenSSL forces the cache access at the table lookup point in a constant
order~\cite{Doychev16}. This constant order table lookup is demonstrated with a
sample C code in Appendix~\ref{sec:scatter-gather}. The base address of the lookup
table is aligned to zero the least-significant bits, and scatter and gather
methods are employed to mimic the Fortran-style memory allocation to access the
table in a constant order and remove timing channels.

Ideally, with the base address being aligned, the table access should not
produce an information leak regarding the cache line access model (but
scatter-gather implementation can also be exploited with cache bank
attacks~\cite{yarom2016cachebleed}. As discussed in \S~\ref{sec:check}, our
cache access constraint can be further extended to capture cache bank side
channels). However, since public information (e.g., base address of the table)
is abstracted as a symbol $p$ in \caches, the alignment is not modeled.
Therefore, \caches\ incorrectly flags the table lookup as a leak point, which
leads to a false positive.

\subsection{\revise{}{Flag Secret-Dependent Control Flow}}
\label{subsec:secret-control}
To reduce false negatives and also show the versatility of \caches, we extend
\caches\ to search for secret-dependent branch conditions.
Similar to the detection of secret-dependent memory accesses
(\S~\ref{sec:check}), we check each conditional jump and flag secret-dependent
jump conditions. The conditional jump in REIL IR is \code{jcc}, and the value of
its first operand specifies whether the jump is taken or not. We translate
each secret-dependent condition $c$ into an SMT formula $f$ and solve the
following constraint:

\begin{equation} \label{eq:formula2}
  f \neq f[s'_{i} / s_i]
\end{equation}
\noindent where a satisfiable solution indicates that different secrets lead to
the execution of different branches. In addition, since REIL IR creates
\textit{additional} \code{jcc} statements to model certain x86 instructions
(e.g., the shift arithmetic right \code{sar} and bit scan forward \code{bsf}),
we rule out \code{jcc} statements if their corresponding x86 instructions are
\textit{not} conditional jumps. In this research we do not take \code{jcc} into
consideration if it does not represent an x86 conditional jump, since in general
the silicon implementations of x86 instructions on mainstream CPUs have
\textit{fixed latency}~\cite{intelopt}.

\T~\ref{tab:control-results} presents the evaluation results, including the
information leakage units produced by the same metric used in
\T~\ref{tab:accuracy}. While secret-dependent control flow is absent in all the
AES cases, \caches\ pinpoints multiple instances in every RSA/ElGamal
implementation. We further manually studied each of them, and we report that
besides 4 false positives (explained later in this section), all the other cases
represent secret-dependent branch conditions. In the example given in
Appendix~\ref{sec:branch-openssl}, the value of \texttt{bits}, which is derived from
the private key, is used to construct several conditions. Similar
patterns are also found in other cases.

\noindent \underline{False Negative.}~Bernstein et al.\ exploited the
secret-dependent control flows in Libgcrypt 1.7.6~\cite{bernstein2017sliding},
where the leading and trailing zeros of a window-size secret are used to compute
a branch condition. While the corresponding vulnerable branches also exist in
Libgcrypt 1.7.3, they are not detected by \caches. In general, 32-bit x86 opcode
\code{bsr} and \code{bsf} are used to count the leading and trailing zeros of a
given operand, and both opcodes are lifted into a while loop implemented by a
\code{jcc} statement (for the definition of their semantics, see the \code{bsr}
and \code{bsf} sections of the x86 developer manual~\cite{intel}). Consider a
proof-of-concept pseudo-code below:
\begin{lstlisting}
t = 0;
while (getBit(t, src) == 0) //src could be a secret
{
  t += 1;
}
return t; // the number of trailing zeros in src
\end{lstlisting}
\noindent where the lifted while loop entails implicit information flow, which
is not supported (see \S~\ref{subsec:design-information} for the information
flow policy). In addition, although we disable the checking of \code{jcc}
regarding Constraint~\ref{eq:formula2} if its corresponding x86 instruction is
\textit{not} a conditional jump (like the \code{bsr} and \code{bsf} cases), we
report that once enabling the checking of such \code{jcc} statements,
secret-dependent control flows (e.g., line 3 of the pseudo-code) are detected
for both cases.

 \begin{table}[t]
   \centering
   \caption{Secret-dependent control branches. We found no issue in the AES
     implementations. A summary of all leakage points can be found at
     \T~\ref{tab:sec-dep-control} in the appendix.}
   \label{tab:control-results}
   \resizebox{\linewidth}{!}{
   \begin{tabular}{c@{~}|@{~}c@{~}c@{~}c|@{~}c}
     \hline
     \multirow{2}{*}{\textbf{Implementation}} & \multirow{2}{*}{\textbf{Algorithm}} & \textbf{\# of Secret-dependent} & \textbf{False} & \textbf{Information}  \\
     & & \textbf{conditions} & \textbf{Positive} & \textbf{Leakage Unit} \\
     \hline
      Libgcrypt 1.6.1 & RSA/ElGamal & 21 & 4 & 9 \\
      Libgcrypt 1.7.3 & RSA/ElGamal & 6  & 0 & 4 \\
      mbedTLS 2.5.1   & RSA         & 8  & 0 & 4 \\
      OpenSSL 1.0.2f  & RSA/ElGamal & 12 & 0 & 5 \\
      OpenSSL 1.0.2k  & RSA/ElGamal & 12 & 0 & 5 \\
     \hline
     \texttt{Total} & & 59 & 4 & 27 \\
     \hline
   \end{tabular}
   }
 \end{table}

\noindent \underline{False Positive.}~We find 4 false positives when analyzing
Libgcrypt 1.6.1. This is due to the imprecise modeling of interprocedural call
sites. Consider a sample pseudo-code below:
\begin{lstlisting}
foo(k, p) { <@\comm{// \texttt{k} is $\{\top\}$}@> <@\comm{and \texttt{p} is $\{12\}$}@>
  if (...) {
    r = bar(k); <@\comm{// r is $\{\top\}$}@>
  } else {
    r = bar(p); <@\comm{// r is $\{\top\}$ since $\langle foo, \{12\}\rangle \subseteq \langle foo, \{\top\}\rangle$}@>
    if (r)  <@\comm{// false positive}@>
      ...
}

bar(i){return i;}
\end{lstlisting}
\noindent where \texttt{foo} performs two function calls to \texttt{bar} with
different parameters. The summary of the first call (line 3) is represented as
$(\langle foo, \{\top\}\rangle, \{\top\})$, where $\langle foo, \{\top\}\rangle$
forms the calling context (as explained in \S~\ref{subsec:context-sensitivity},
a calling context includes the caller name and the input), and the second
$\{\top\}$ is the function call output. Then the following function call (line
5) with $\langle foo, \{12\}\rangle$ as the calling context will directly return
$\{\top\}$ and cause a false positive (line 6) according to the recorded
summary, since $\langle foo, \{12\}\rangle \subseteq \langle foo,
\{\top\}\rangle$. Our study shows that such sound albeit imprecise modeling
caused 4 false positives when analyzing Libgcrypt 1.6.1.

\section{Discussion}
\label{sec:discussion}
\noindent \textbf{Soundness.}~Our abstraction is sound (see
Appendix~\ref{sec:formalization-proof}),
but the \caches\ implementation is
soundy~\cite{livshits2015soundness} as it roots the same assumption as previous
techniques that aim to find bugs rather than performing rigorous
verification~\cite{machiry2017checker,xie2005scalable,livshits2003tracking}.

\caches\ adopts a lightweight but unsound memory model implementation; program
state representations are optimized to reduce the memory usage and speed up the
analysis.
There is a line of research aiming to deliver a (nearly) sound memory model when
analyzing x86
assembly~\cite{balakrishnan2004analyzing,reps2008improved,Reynolds79,Bradley2006}.
We leave it to future work to explore practical methods to improve \caches\ with
a sound model without undermining the strength of \caches\ in terms of
scalability and precision. \revise{Furthermore, we present the soundness proof
  of \sas\ in our technical report~\cite{sastr}; a rigorous verification tool
  can be implemented accordingly.}{}

\noindent \textbf{Reduce False Positives.}~Our abstract domain \sas\ models public program
information with free public symbols. To further improve the analysis precision
and eliminate false positives, such as in the case discussed in
\S~\ref{subsec:eval-rsa}, one approach is to perform a
finer-grained modeling of public program information. To this end, so-called
``lazy abstraction'' can be adopted to postpone abstraction until
necessary~\cite{thakur12bilateral}. In contrast to our current approach
where analyses are performed directly over \sas, lazy abstraction provides a
flexible abstraction strategy on demand, where different program points can
exhibit distinct levels of precision. Well-selected program points for lazy
abstraction are critical to achieve scalability. For example, abstraction can be
performed at every loop merge point or whenever abstract formulas become too large and exhaust the memory
resource.
We leave it to future work to explore practical strategies for lazy
abstraction.

\section{Related Work}
\label{sec:related}
\noindent \textbf{Timing Attacks.}~Kocher's seminal
paper~\cite{Kocher1996} identifies timing attacks as a potential threat to
crypto system. Later work finds that timing information reveals the victim
program's usage of data/instruction cache, leading to efficient timing attacks
against real world cryptography software, including
AES~\cite{Gullasch:2011,Osvik+:2006, Tromer10,daniel2005cache,
  bonneau2006cache,aciicmez2007cache}, DES~\cite{tsunoo2003crypt},
RSA~\cite{Aciicmez07, Percival:2005, Yarom14}, ElGamal~\cite{Zhang12}, and ECDSA~\cite{benger2014side}.
Recent work shows that such cache-based timing attacks are possible on emerging
platforms, such as cloud computing, VM
environments, trusted computing environments, and mobile 
platforms~\cite{ristenpart2009hey,xu2011an,Wu12, zhang2011homealone, Liu15,lipp2016mobile,maurice2017hello,gruss2017sgx,bulck2018sgx,gruss2016prefetch}.

\noindent \textbf{Detect Cache-Based Timing Channels.}~\cachea\ leverages static
analysis techniques (i.e., \ai) to reason information leakage due to cache side
channels~\cite{cacheaudit,Doychev16}. \revise{}{\caches\ outperforms
  \cachea\ due to our novel abstract domain. \cachea\ uses relational
  and numerical abstract domains to only infer the information leakage bound,
  while our abstract domain models semantics with symbolic formulas, pinpoints
  information leaks with constraint solving, and enables the generation of
  counter examples to promote debugging. In addition, we propose a principled
  way to improve the scalability by tracking secrets and public information with
  \textit{different granularities}. This enables a context-sensitive
  interprocedural analysis of real-world cryptosystems for which CacheAudit is
  not capable of handling.} Brotzman et al.~\cite{bortzman2018casym} propose a
static symbolic reasoning technique that also covers multiple program paths.
However, their analysis lacks abstraction of public values, and can analyze only
small-size programs.

In contract, dynamic analysis-based approaches, such as taint analysis or
trace-based symbolic execution, are incapable of analyzing the whole
program~\cite{wang2017cached,jan2018microwalk,gorka2017side,weiser2018data,guo2018side}.
\revise{}{\cachediff~\cite{wang2017cached} performs symbolic execution towards a
  single trace to detect side channels. In contrast, abstract interpretation
  framework approximates the program \textit{collecting semantics}, which
  formalizes program abstract semantics at arbitrary program points regarding
  any path and any input. This is fundamentally different and much more
  comprehensive comparing to a path-based tool, like \cachediff.} Wichelmann et
al.~\cite{jan2018microwalk} log execution traces and perform differential
analysis of various granularities to detect side channels. Weiser et
al.~\cite{weiser2018data} detect address-based side-channels by executing test
programs under input variants and further compare traces to detect leakages.

\noindent \textbf{Countermeasure.}~Existing countermeasures against cache
side-channel attacks can be categorized into hardware-based and software-based
approaches.
Hardware-based solutions focus on randomizing the cache accesses with new cache
design~\cite{wang2007new,wang2008novel,kong2009hardware,wang2006covert,liu2014random,liu2016cache},
or enforcing fine-grained isolation 
with respect to cache usage~\cite{shi2011cache,kim2012stealthmem}. Wang et
al. propose locking the cache lines and hiding cache access
patterns~\cite{wang2007new}, which further obfuscates cache accesses by
diversifying the cache mappings~\cite{wang2008novel}. Tiwari et
al.~\cite{tiwari2011crafting} devise a novel micro architecture for
information-flow tracking by design, where noninterference is deployed as the
baseline confidentiality property. Another direction at the hardware level is
based on contracts between software and hardware~\cite{Zhang13,Sapper:asplos,zhang3}, 
where contracts are enforced by formal methods (e.g., type systems) on the hardware
side. Furthermore, some advanced hardware extensions, like hardware transactional memory, 
have also been leveraged to prevent side channels even inside Intel SGX~\cite{gruss2017transactional}.

Analyses are also conducted on the software level to mitigate side channel
attacks~\cite{coppens2009practical,aviram2010determinating,raj2009resource,schwarz2018javascript,schwarz2018keydrown}.
Program transformation techniques are leveraged to remove control-flow timing
leaks by equalizing branches of conditionals with secret
guards~\cite{Agat00}, together with a binary static
checker~\cite{molnar2005program}, and its practicality is
evaluated~\cite{coppens2009practical}. Constant time code defeats timing attacks
by ensuring the control flow, memory accesses, and execution time of individual
instruction is secret
independent~\cite{almeida2013formal,Hedin:Sands:2005,multirun,Barthe:2006,almeida2016verifying,kopf2010approximation}.

\section{Conclusion}
\label{sec:conclusion}
In this paper, we have presented \caches\ for cache-based timing channel
detection. Based on a novel abstract domain \sas, \caches\ does fine-grained
tracking of sensitive information and its dependencies, while performing
scalable analysis with over-approximated public information. We evaluated
\caches\ on multiple real-world cryptosystems. \caches\ confirmed over 154
information leaks reported by previous research and pinpointed 54 leaks not
known previously.

\section{Acknowledgments}
We thank the Usenix Security anonymous reviewers and Gary T.\ Leavens for their
valuable feedback. The work was supported in part by the National Science
Foundation (NSF) under grant CNS-1652790, and the Office of Naval Research (ONR)
under grants N00014-16-12912, N00014-16-1-2265, and N00014-17-1-2894.
\vspace{6pt}

{\normalsize \bibliographystyle{acm}
\bibliography{bib/analysis,bib/sidechannel,bib/logic,bib/symbolic-execution,bib/bmc,bib/timing,bib/new}}

\begin{appendix}
\begin{figure}[t]
\section{Unknown Information Leaks in OpenSSL}
\label{sec:unknown-openssl}
\centering
\begin{lstlisting}
int BN_num_bits(const BIGNUM *a) {
 int i = a<@->@>top - 1;
 bn_check_top(a);

 if (BN_is_zero(a))
   return 0;
 return ((i * BN_BITS2) + BN_num_bits_word(a<@->@>d[i]));
}

int BN_num_bits_word(BN_ULONG @l@) {
  static const char bits[256]={
    0,1,2,2,3,3,3,3,4,4,4,4,4,4,4,4,
    ...
    8,8,8,8,8,8,8,8,8,8,8,8,8,8,8,8,
  };
  if (@l@ & 0xffff0000L) {
    if (l & 0xff000000L)
      <@\textbf{return bits[\textcolor{red}{l} >> 24] + 24;}@>
    else
      <@\textbf{return bits[\textcolor{red}{l} >> 16] + 16;}@>
  }
  else {
    if (@l@ & 0xff00L)
      <@\textbf{return bits[\textcolor{red}{l} >> 8] + 8;}@>
  else
      <@\textbf{return bits[\textcolor{red}{l}];}@>
  }
}
\end{lstlisting}
\caption{RSA information leaks found in OpenSSL (1.0.2f). Program secrets and
  their dependencies are marked as \textcolor{red}{red} and the leakage points
  are \textbf{boldfaced}.}
\label{fig:rsa-openssl}
\end{figure}

\begin{figure}[H]
\section{Scatter \& Gather Methods in OpenSSL}
\label{sec:scatter-gather}
\centering
\begin{lstlisting}
char* align(char* buf) {
  uintptr_t addr = (uintptr_t) buf;
  return (char*)(addr - (addr&(BLOCK_SZ-1)) + BLOCK_SZ);
}

void scatter(char* buf, char p[][16], int k) {
  for (int i = 0; i < N; i++) {
    buf[k+i*spacing] = p[k][i];
  } 
}

void gather(char* r,char* buf,int k) {
  for (int i = 0; i < N; i++) {
    r[i] = buf[k+i*spacing]; 
  }
}

\end{lstlisting}
\caption{Simple C program demonstrating the scatter \& gather methods in OpenSSL
  to remove timing channels. This program should be secure regarding our threat
  model, but it would become insecure by skipping the alignment function.}
\label{fig:scatter-gather}
\end{figure}

\begin{figure}[t]
\section{Unknown Information Leaks in mbedTLS}
\label{sec:unknown-mbedtls}
\centering
\begin{lstlisting}
int mbedtls_mpi_exp_mod(mbedtls_mpi *X, mbedtls_mpi *A,
  mbedtls_mpi *E, mbedtls_mpi *N, mbedtls_mpi *_RR)
{
  ...
  while (1) {
    @ei@ = (<@E->p@>[nblimbs] <@>>@> bufsize) & 1;
    ...
    @wbits@ |= (@ei@ <@<<@> (wsize - nbits));
    ...
    mpi_montmul(X, &W[@wbits@], N, mm, &T);
  }
  ...
}

static int mpi_montmul(mbedtls_mpi *A, mbedtls_mpi *@B@,
  mbedtls_mpi *N, mbedtls_mpi_uint mm, mbedtls_mpi *T)
{
   ...
   <@\textbf{m = (\textcolor{red}{B->n} < n) ? \textcolor{red}{B->n} : n;}@>
   for(i = 0; i < n; i++)
   {
     <@\textbf{u1 = (d[0] + u0 * \textcolor{red}{B->p[0]}) * mm;}@>
     <@\textbf{mpi_mul_hlp(m, \textcolor{red}{B->p}, d, u0);}@>
   }
   ...
}

\end{lstlisting}
\caption{RSA information leaks found in mbedTLS (2.5.1). Program secrets and
  their dependencies are marked as \textcolor{red}{red} and the leakage points
  are \textbf{boldfaced}.}
\label{fig:rsa-mbedtls}
\end{figure}

\begin{figure}[H]
\section{Secret-Dependent Branch Conditions in OpenSSL}
\label{sec:branch-openssl}
\centering
\begin{lstlisting}
int BN_mod_exp_mont_consttime(BIGNUM *rr,
  const BIGNUM *a, const BIGNUM *p,
  const BIGNUM *m, BN_CTX *ctx,
  BN_MONT_CTX *in_mont) {
 ... 
 @bits@ = BN_num_bits(p);
 <@\textbf{if (\textcolor{red}{bits} == 0)}@>
   ...

 window = BN_window_bits_for_exponent_size(@bits@);
 <@\textbf{for (wvalue = 0, \textcolor{red}{i} = \textcolor{red}{bits}\%window; i>=0; i--,\textcolor{red}{bits}--)}@>
 {
   ...
   <@\textbf{while (\textcolor{red}{bits} >= 0)\{}@> 
     ...
   }
 }
  ...
}

#define BN_window_bits_for_exponent_size(@b@) \
       ((<@\textbf{\textcolor{red}b) > 671}@> ? 6 : \
        (<@\textbf{\textcolor{red}b) > 239}@> ? 5 : \
        (<@\textbf{\textcolor{red}b) >  79}@> ? 4 : \
        (<@\textbf{\textcolor{red}b) >  23}@> ? 3 : 1)
\end{lstlisting}
\caption{Several secret-dependent branch conditions found in OpenSSL (1.0.2f).
  Program secrets and their dependencies are marked as \textcolor{red}{red} and
  the information leakage conditions are \textbf{boldfaced}. Note that the
  output of \texttt{BN_num_bits} depends on the private key.}
\label{fig:rsa-branch}
\end{figure}

%


\section{Formalization and Soundness Proof}
\label{sec:formalization-proof}
In this appendix section, we define necessary components to formalize the basis
of our abstract interpretation-based analysis. We also give soundness proof of
our abstraction.

\subsection{Concrete Semantics}
\label{subsec:concrete}
To facilitate a ``down-to-earth'' modeling of program cache access information,
our analysis is directly performed towards program binary executables. For the
ease of analysis, we translate input binary code to an intermediate language
(REIL~\cite{reil}) with low-level program representation (details are disclosed
in the implementation section). Our analysis is conducted on the basis of the
lifted REIL programs. Since we analyze cryptosystems and capture side channels
due to secret-dependent cache access, we assume REIL programs have been
annotated with locations of program secrets before feeding into our analysis
framework. In this section, we define the concrete semantics of REIL programs
upon which our abstraction is derived.

\begin{figure}[H]
$
\begin{array}{l@{\ }l@{\ }l@{\ }@{\ }l}
\mathsf{Prog} & p & ::= & c^\ast \\
\mathsf{Comm} & c &::= & r := e \mid r_1 :=\text{load}(r_2) \mid \text{jmp}(r_1,r_2) \\
&   &~~\mid & r_1 := \text{is\_zero}(r_2) \mid \text{store}(r_1,r_2) \\
\mathsf{Expr} & e &::= & r \mid n \mid e \Diamond e \\
\mathsf{Oper} & \Diamond &::= & + \mid - \mid \mathrm{AND} \mid \mathrm{BSH} \mid \mathrm{DIV} \mid \mathrm{MOD} \\
&   &~~\mid & \mathrm{MUL} \mid \mathrm{OR} \mid \mathrm{SUB} \mid \mathrm{XOR}
\end{array}
$
\caption{The syntax of the REIL language, where $r$ denotes a (temporary)
  register, and $n$ represents a 32-bit unsigned integer.}
\label{fig:reil-syntax}
\end{figure}

\subsubsection{Syntax}
\label{subsubsec:reil-syntax}
REIL is designed as a platform-independent intermediate language. One assembly
instruction is typically mapped into a sequence of REIL instructions to expose 
composed operations, implicit memory accesses, and CPU flag manipulations 
into explicit statements.
Each REIL instruction has a typical three-address representation with at most three 
operands (registers, constant, etc.).
REIL programs contain unlimited number of temporary registers (e.g., \texttt{t0}, \texttt{t1}) and also 
preserve the originally x86 CPU registers like \texttt{eax} and \texttt{esp}. 
Also, REIL uses a flat memory model, which has an infinite amount of storage.
In this paper we treat registers uniformly unless otherwise stated.

\F~\ref{fig:reil-syntax} defines the syntax of REIL programs. A program consists of
a possibly empty sequence of numbered commands. Typical commands include memory
operations, control-flow transfers as well as assignments.
The notation, $\Diamond$, stands for binary operators which compute two operands (i.e.,
$\mathsf{Expr}$) and use the output to update a third one.
For the ease of presentation, we note that immediate addressing, such as \texttt{$r:=load(n)$}, 
is encoded as \texttt{$r_1:=n$ $r:=load(r_1$)}.

\subsubsection{Concrete Semantics}
\label{sec:concrete-semantics}
In this section, we define the concrete semantics of secret-annotated REIL,
which is inspired by the work of Schwartz et al. and has operational
flavors~\cite{Schwartz2010}.
The program assumes the availability of user annotations that specify the location program secrets (e.g., 
a piece of private key or random numbers). In other words, users need to pinpoint some registers which contain
addresses of program secrets, for instance the location of private keys before program starts. 
Note that while the pinpointed content is confidential, memory addresses in those registers
are considered to be non-secret.
Besides those specified registers, we assume the rest stores public values. 
In the semantics definition, we annotate such pointers of secrets by specifying 
their values as elements in set $\Header$. Each $u \in \Header$ represents the 
address of program secrets in memory, for instance the beginning address of a private key array.
Addresses in $\Header$ are disjoint with those addresses storing public values, and it is easy to see that 
memory loading 
from $u \in \Header$ would create one piece of secret.
In addition, we define another set $\IM$  which contains all the literal numbers in the semantics definition.

To model secret information flow, we introduce two labels as $hi$ and $lo$. 
Label $hi$ denotes program secrets, such as elements in a private key array. In contrast, 
label $lo$ denotes program non-secret information. It is easy to see that $lo$ and $hi$ form
a simple lattice of two elements, where $hi$ sits on the upper level.
Therefore any program value stored in memory cells and registers can be represented 
as a pair of an integer and its security tag: $\V = \{\INTEGER\}
\times \{hi, lo\}$. For the ease of the definition, for each data $v \in \V$, we use
$v.n$ to represent the first projection of the pair (i.e., the content) while
$v.l$ gives us the associated security label.

\noindent \textbf{Program State.}~A \emph{program state} $s \in \textit{S}$ is a
tripe of a store, a memory and a program counter. A store, $\sigma$, is a
partial function that maps a register to its value. A memory, $\MEM$, is a
partial function that maps an address to its value.

A program is a sequence of numbered commands that are stored in a command table, $\theta$.
A program table maps a command number to its command, and 
a program counter $pc$ is used to fetch the command from $\theta$.
In \tabref{tb:sem-fun}, we present meta-functions
used in the definition of the semantics. 
\begin{table}
 \centering
\begin{tabular}{|r|l|}
\hline
Notation & Explanation \\
\hline 

$s.\sigma$ & the store $\sigma$, where $s = (\sigma, \MEM, pc)$ \\
\hline
$s.\MEM$ & the memory $\MEM$, where $s=(\sigma, \MEM, pc)$\\
\hline
\multirow{2}{*}{$s.pc$} & the program counter $pc$ in the state $s$, \\
 & where $s = (\sigma, \MEM, pc)$ \\
\hline
$a^\ast$ & the possibly empty sequence $a$ \\
\hline
\end{tabular}
\caption{Meta-functions used by the definition of REIL semantics.}
\label{tb:sem-fun}
\end{table}

\noindent \textbf{Concrete Semantics.}~The concrete semantics rules of REIL are
defined in \F~\ref{fig:reil-semantics}. $\mathcal{N}$ is defined as the standard
meaning function for numeric literals, and function $\MEANINGFUN{E}$ evaluates a
given expression. Function $\MEANINGFUN{MP}$ takes a program table and a program
state, and returns a new state. Function $\MEANINGFUN{MS}$ takes a program state
and a command, and produces a new state. With a bound check on each memory
access, error state may be generated. For simplicity, the bound check is omitted
in the definition.

As shown in \F~\ref{fig:reil-semantics}, $\MEANINGFUN{E}$ evaluates the input expression regarding the current program 
store $\sigma$,
and outputs a value pair $v \in \V$ which includes the computation output and its associated security label. 

The secret labels are propagated following the standard convention~\cite{Denning76}. As shown in
the $\MEANINGFUN{E}$ function, computations among tags of different security degrees naturally lead to $hi$, while 
security tags are preserved for the other cases.
Indeed recall $hi$ and $lo$ forms a lattice of two elements, it is easy to see that 
the least upper bound on the lattice is obtained at this step.

In \F~\ref{fig:reil-semantics}, the command, \texttt{x := e}, writes the value of \texttt{e} to the register \texttt{x}.
The command, \texttt{r$_1$:=load(r$_2$)} fetches the value from a memory location
whose address is in \texttt{r$_2$} and writes to it to \texttt{r$_1$}. 
Again, the value stored in \texttt{r$_2$} will be associated with label $hi$ 
whenever the memory address in \texttt{r$_1$} is within $\Header$. In contrast, 
the original tag (could be $hi$ or $lo$) of the loaded memory content will flow into the memory reader.
Memory store command \texttt{store(r$_1$,$r_2$)}
writes \texttt{r$_1$}'s value into the memory location pointed by 
\texttt{r$_2$}.
Command \texttt{x := is\_zero(y)} sets \texttt{x} to \texttt{(1,lo)} if \texttt{y} is zero and 
sets \texttt{x} to \texttt{(0,lo)} if \texttt{y} is one.
Command \texttt{jmp(r$_1$,r$_2$)} performs a
(conditional) jump towards the address stored in \texttt{r$_2$}, if predicate in \texttt{r$_1$} is true. Note that 
x86 direct jump (i.e., \texttt{jmp}) can be easily modelled by setting the condition as always true.

\begin{figure}[!t]
$
\begin{array}{l}
\MEANINGFUN{E}: \nonterm{Expr} \rightarrow \Store \rightarrow \V \\
\ME{r}{\sigma} = (\sigma(r), lo) \\
\ME{n}{\sigma} = (\mathcal{N}\synbracket{n}, lo) \\
\ME{e_1 \Diamond e_2}{\sigma} = \LETHEAD{v_1}{\ME{e_1}{\sigma}} ~ \LETIN \\
\quad  \LETHEAD{v_2}{\ME{e_2}{\sigma}} ~ \LETIN ~ v_1 ~ \Diamond ~ v_2 \\
\quad \mbox{where } v_1 ~ \Diamond ~ v_2 = \\
\quad \quad \IFHEAD{v_1.l \neq v_2.l} ~ \IFTHEN ~ (v_1.n ~ \MO{\Diamond} ~ v_2.n, hi) \\
\quad \quad \IFELSE ~ (v_1.n ~ \MO{\Diamond} ~ v_2.n, v_1.l) \\
\\
\MEANINGFUN{MP}: \nonterm{Command Table} \rightarrow \nonterm{State} \rightarrow  \nonterm{State} \\
\\
\MP{\theta[pc]}{\sigma, \MEM, pc} =  \\ 
\quad \LETHEAD{(\sigma', \MEM', pc')} { \MS{\theta(pc)}{\sigma,\MEM, pc}} ~ \LETIN \\
\quad \quad (\MS{\theta[pc']}{\sigma', \MEM', pc'} \\
\\
\MEANINGFUN{MS}: \Comm \rightarrow \nonterm{State} \rightarrow  \rightarrow\nonterm{State}  \\
\MS{r := e}{\sigma, \MEM, pc} = (\sigma[r \mapsto \ME{e}{\sigma}], \MEM, pc + 1)  \\
\\
\MS{r_1 := load(r_2)}{\sigma, \MEM, pc} = \LETHEAD{v}{\ME{r_2}{\sigma}} ~ \LETIN \\
\quad \IFHEAD{v.n \in \Header} ~ \IFTHEN \\
\quad\quad (\sigma[r_1 \mapsto (\MEM(v.n).n, hi)], \MEM, pc + 1)  \\
\quad \IFELSE \: (\sigma[r_1 \mapsto \MEM(v.n)], \MEM, pc + 1) \\

\\
\MS{store(r_1, r_2)}{\sigma, \MEM, pc} = \LETHEAD{v_1}{\ME{r_1}{\sigma}} ~ \LETIN ~ \\ \quad\LETHEAD{v_2}{\ME{r_2}{\sigma}} ~ \LETIN ~ (\sigma, \MEM[v_2 \mapsto v_1], pc + 1) \\
\\
\MS{r_1 := is\_zero(r_2)}{\sigma, \MEM, pc} = \LETHEAD{v}{\ME{r_2}{\sigma}} ~ \LETIN ~ \\
\quad \IFHEAD{v.n = 0} ~ \IFTHEN  ~ (s[r_1 \mapsto (1, lo)], \MEM, pc + 1) \\
\quad \IFELSE ~ (\sigma[r_1 \mapsto (0, lo)], \MEM, pc + 1) \\
 \\
\MS{jmp(r_1, r_2)}{\sigma, \MEM, pc} = \LETHEAD{v_1}{\ME{r_1}{\sigma}} ~ \LETIN ~ \\
 \quad \IFHEAD{v_1.n \neq false} ~ \IFTHEN  \\
 \quad \quad \LETHEAD{v_2.n}{\ME{r_2}{\sigma}} ~ \LETIN  ~ (\sigma, \MEM, v_2.n)\\
 \quad \IFELSE (\sigma, \MEM, pc + 1) \\
\\
\end{array}
$
\caption{The concrete semantics rules of REIL, where $\MO{-}$ yields the semantics
  of operators. The \textit{allocate} function takes a memory and the expected allocation size, and returns a new address.}
\label{fig:reil-semantics}
\end{figure}

\subsubsection{Concrete Program Traces}
\label{sec:cpt}
A program trace is a finite non-empty alternating sequence of program states and actions, $s_0, \theta(s_0.pc), s_1, \theta(s_1.pc), \cdots, s_{n}$, such that $s_0$ is an initial state, and $\forall ~ i \in \{0, \cdots, n - 1\}$, $(s_{i + 1}) = \MS{\theta(s_i.pc)}{s_i}$, where $\theta$ is a command table.
A program computation is the powerset of program traces.
\textcolor{black}{Thus, for a given piece of input, the program execution generates a 
concrete computation tree whose path forms an execution trace $t$. 
If the computation is divergent, the execution trace is infinite.}

Let $tr = s_0 \theta(s_0.pc) s_1 \theta(s_1.pc) \cdots s_n$ be a computation trace. 
The function $\ROOT{tr} = s_0$ denotes $tr$'s root. 
The notation $\Indexed{tr}{i}$ is the trace ending with 
the $i$-th state in $t$, i.e., $\Indexed{tr}{i} = s_0, \theta(s_0.pc), s_1, \theta(s_1.pc), \cdots, s_{i-1}$.

The notation $tr \TRANSIT tr'$ means that there is a sequence of actions $a^\ast$, such that $\ROOT{tr'} = \MS{a^\ast}{\ROOT{tr}}$.

\subsubsection{Concrete Path-based Collecting Semantics}

The forward path-based collecting semantics maps a program counter $pc$
to the set of execution paths that end at the program 
point referred by $pc$: $\COLLCP: PC \rightarrow \powerset{\Trace}$, 
which is $\COLLCP(pc) = \Indexed{t}{pc}$, where $t$ is a program computation tree.

\subsection{Abstract Semantics}
\label{sec:abstraction}
This section explains necessary components of our abstraction. To this end, we formalize 
our abstract domain \sas\ and further establish a lattice over \sas. We then
present abstract transfer functions over \sas\ and the abstract collecting semantics, and also discuss the 
termination conditions.

\subsubsection{Abstract Values}
\label{sec:abstract-values}
This section proposes the design of a novel abstract domain named Secret-Augmented Symbolic domain (\sas). 
\sas\ is designed to be acquired by a standard abstract 
interpretation framework to analyze secret-aware software systems.
To do so, 
we start by defining abstract values $f \in \mathbf{AV}$. 
Each $f \in \mathbf{AV}$ composites symbols and constants via operators, and 
the grammar of $\mathbf{AV}$ is defined in \F~\ref{fig:abs-syntax}. 
Shortly 
we will show that \sas\ is defined as the powerset of $\mathbf{AV}$, meaning each element $e \in \sas$ is 
an abstract value set $\{f \mid f \in \mathbf{AV}\}$.

\begin{figure}[!t]
$
\begin{array}{l@{\ }l@{\ }c@{\ }l}
\mathsf{Abstract Value} & av &::=\; &\top \mid P \mid S \mid U \mid N \mid E \\ 
\mathsf{Numbers} & N  & ::= & n \mid n_1 \INFOP n_2 \\
\mathsf{Public} & P & ::= & p \\
\mathsf{Secret} & S & ::= & s \mid s \INFOP n \mid n \INFOP s \mid \epsilon \oplus s \mid h \oplus s \\
     & & & \mid s_1 \INFOP s_2 \\
\mathsf{Header} & U & ::= & h \mid h \INFOP n \\
\mathsf{ESP} & E & ::= &\epsilon \mid \epsilon \oplus n \\
\mathsf{Ops} & \tilde{\Diamond} & ::= & \oplus \mid \otimes \\
\mathsf{Linear Ops} & \oplus &::= & + \mid - \\
\mathsf{Other Ops} & \otimes &::= &\times \mid \div \mid \% \mid \mathrm{AND} \mid \mathrm{OR} \mid \mathrm{XOR} \\
& & & \mid \mathrm{SHIFT} \\
\end{array}
$
\caption{Grammar of $\mathbf{AV}$.}
\label{fig:abs-syntax}
\end{figure}

For the ease of presentation and analysis, 
\F~\ref{fig:abs-syntax} puts abstract values into \textit{disjoint categories}; each category has its own notation. 
We now elaborate on different categories of abstract values $f \in \mathbf{AV}$. 

\begin{itemize}
\item $P$: this category only has one unique identifier $p$, which presents any non-secret program 
information. By keeping track of public information with only one identifier $p$,
$P$ enables a coarse-grained and scalable representation of program semantics.
\item $S$: this category includes secret-carry abstract values, 
where each piece of program secret are represented by one symbol $s_{i}$. 
In other words, $S = \{f \in \mathbf{AV} \mid \exists s\in \{s_i\}~\text{s.t.}~s~\text{is a subterm of}~f\}$
where $\{s_i\}$ represents the set of secrets that our analysis intends to keep track of.
For example, an abstract memory address $s_0 + 4$ means the base address is derived from a piece of program secret.
Note that in this way, different pieces of program secrets can be represented by the variable $s_{i}$ with 
different indexes. Secrets are kept track of in a fine-grained approach, 
and an element in this category is named a $S$-value.
\item $U$: this category represents special memory addresses that point to secrets before program execution.
To be more specific, we create special identifier $h$ to present the base address of a secret memory
  region, e.g., the header of a private key array. Before execution, we assume that pointers referring 
  to the base address of secret memory chunk have been initialized with identifier $h$. Hence 
$U = \{f \in \mathbf{AV} \mid h~\text{is identifical to or a subterm of}~f\}$.
  A value $u \in U$ in this category is called a $U$-value.
\item $E$: this category represents abstract memory addresses of x86 stack.
To this end, we create special identifier $\epsilon$ to present the initial value of stack register 
   \texttt{esp}.
  Before analysis, we assume that the value set of register \texttt{esp}
  has been initialized with $\{\epsilon\}$. Hence 
  $E = \{f \in \mathbf{AV} \mid \epsilon~\text{is identifical to or a subterm of}~f\}$. 
  A value $e \in E$ in this category is named as a $E$-value.
 Also, memory offsets of values in $E$ should be constant (\F~\ref{fig:abs-syntax}), 
 since $E$ is defined as disjoint with other categories.
\end{itemize}

At this step we have introduced several kinds of abstract identifiers and categories, and 
strict partial order $\succ$ can be defined over these values from each category with respect to granularities of 
semantics abstraction.\footnote{It is easy to see that asymmetric (if $C_1 \succ C_2$ holds then $C_2 \succ C_1$ does not),
transitive (if $C_1 \succ C_2$ and $C_2 \succ C_3$ then $C_1 \succ C_3$), and irreflexive ($C_i \succ C_i$ is impossible) 
relations hold. Hence strict partial order over the defined categories is established.}
Nevertheless, it is not possible to define a complete order, since some 
identifiers are not comparable, for instance $u \in U$ and $e \in E$. 

\F~\ref{fig:absvalues} presents 
a lattice of aforementioned categories over $\mathbf{AV}$ with respect to different abstraction levels. 
Obviously, $\top$ 
over-approximates any program information, and hence it is located on the top of the lattice. $p \in P$ represents
all the program public information, and therefore exhibits the least upper bound of any abstract values 
$f \in (N \cup E \cup U)$. On the other hand, $s \in S$ stands for the program secrets and information derived from secrets, 
hence $s \in S$ is not comparable to $p \in P$.
Also, there is no direct comparison between $N$, $U$, and $E$, 
although all of them deliver finer-grained semantics information  comparing to $P$.

\begin{figure}[!t]
\begin{center}
\begin{tikzpicture}[level distance=2.5em, level 1/.style={sibling distance=4em}, level 2/.style={sibling distance=2.7em}]
\node {$\top$}
    child{ node{ $p$}
	  child{ node {$n$} }
	  child{ node {$u$} }
	  child{ node {$e$}}}
    child{ node {$s$} }
; 
\end{tikzpicture}
\caption{Abstraction lattice for abstract values in $\ABSV$, 
where $s \in S$, $u \in U$, $n \in N$, $e \in E$.}
\label{fig:absvalues}
\end{center}
\end{figure}
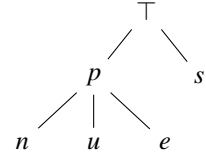

\subsubsection{Domain Definitions}
\label{subsec:sas-domain}
We now elaborate on the design of our proposed abstract domain. The proposed
abstract domain is named as \textit{Secret-Augmented Symbolic} domain (\sas),
since it augments the modeling of secret information by performing fine-grained 
tracking of program secrets and dependencies on the secrets, 
while only delivers coarse-grained overestimations towards public information.

\begin{definition} \label{def:av}
Let $\mathbf{AV}$ be the set of abstract values. Then
\[
    \textbf{\rm\bf SAS} = \powerset{\mathbf{AV}}
\]
forms a domain whose elements are subsets of all valid abstract values.
\end{definition}

We now make $\textbf{SAS}$ as a lattice. To do so, we will specify a top
element $\top \in \textbf{SAS}$ and a join operator $\sqcup$ over
$\textbf{SAS}$.

\noindent \textbf{Set Collapse and Bound.}~Each element in \sas\ is a set of
formulas defined in $\mathbf{AV}$. Considering formulas with different degrees
of abstractions may exist in one set, here we define reasonable rules to
``collapse'' elements in a set. The ``collapse'' function $\text{COL} :
\textbf{SAS} \rightarrow \textbf{SAS}$ is given by:

\[
\scalemath{1.0}{
\text{COL}(X) =
   \begin{cases}
      \{\top\}           & \textrm{if}\; \top \in X \\
      \{\top\}           & \textrm{else if}\; p \in X\land S \cap X \neq \varnothing\\
      \{p\}    & \textrm{else if}\; p \in X\land S \cap X = \varnothing\\
      X & \textrm{otherwise}\\
   \end{cases}
}
\]
 
While the first three rules introduce single symbols as a safe and concise
approximation, the last rule preserve a set in \sas.

In addition, each set in \sas\ is also bounded with a maximum size of $N$
through function $\text{BOU}$ as follows:

\[
\scalemath{1.0}{
\text{BOU}(X) =
   \begin{cases}
      \{\top\}           & \textrm{if}\; \vert X\vert > N \land S \cap X \neq \varnothing\\
      \{p\}    & \textrm{else if}\; \vert X\vert > N \land S \cap X = \varnothing\\
      X  & \textrm{otherwise}\\
   \end{cases}
}
\]

Hence, the abstract value set of any variable is bounded by $N$ during
computations within \sas, which practically speedups the convergence of our
analysis.

With $\text{COL}$ and $\text{BOU}$ defined, we can finally complete
$\textbf{SAS}$ as a lattice.
\begin{claim}
$\text{\rm\bf SAS} = \powerset{\mathbf{AV}}$ forms a lattice with the
top element 
\[\top_{\textbf{\rm\bf SAS}} = \{\top\}\]
bottom element 
\[\bot_{\textbf{\rm\bf SAS}} = \{ \}\]
and the join operator
\[
\sqcup = \text{\rm BOU} \circ \text{\rm COL} \circ \cup
\]
\end{claim}

\begin{theorem} \label{thm:col} The definition of $\sqcup$ is sound.
\end{theorem}

\subsubsection{Abstract Semantics}
\label{subsec:abs-semantics}
This section defines the abstract semantics. In general, abstract semantics defines 
computations over abstract value sets
$\{f \mid f \in \mathbf{AV} \}$. A program's abstract state
is a triple of an abstract store ($\INFSTORE$), an abstract memory ($\INFMEM$),
and a set of program counters. An abstract store is a partial function that maps
a register to a set of abstract values; an abstract memory is a partial function
that maps an abstract value (i.e., abstract memory address) to a set of abstract values.

To interpret abstract memory accesses, the classic challenge is to reason point-to relationships.
It should be noted that aliasing analysis has been extensively studied in literatures to 
date~\cite{Guo2005,kinder2010static,fernandez2002speculative,liu2013binary,debray1998alias},
and therefore in this section we assume to acquire point-to information from a standby module 
along the analysis. In particular,
we assume that each memory pointer maintains a set of pointers which may be
aliased with the pointer. Therefore, memory loading via a pointer indeed joins the value sets of 
its aliased memory locations, while memory storing updates all the aliased locations. 
For the ease of presentation, we omit this part in \F~\ref{fig:infsemantics}.

\begin{figure}[!ht]
\begin{displaymath}
\begin{array}{l}
\MEANINGFUN{\tilde{E}}: \Expr \rightarrow \textit{AbsStore}  \rightarrow \powerset{\mathbf{AV}} \\

\INFME{r}{\INFSTORE} = \INFSTORE(r) \\

\INFME{n}{\INFSTORE} = \{n\} \\

\INFME{e_1 \Diamond e_2}{\INFSTORE} = \{ a \INFOP b \mid a \in \INFME{e_1}{\INFSTORE}, b \in \INFME{e_2}{\INFSTORE}\} \\
\\
\MEANINGFUN{\hat{MP}}: \nonterm{Command Table} \rightarrow \powerset{\nonterm{AbstractState}} \rightarrow \\
\qquad~ \powerset{\nonterm{AbstractState}} \\
\INFMP{\theta}{\hat{S}} = \bigcup_{\hat{s} \in \hat{S}}\{\INFMS{\theta(pc)}{\hat{s}}\}_{pc \in \hat{s}.pc} \\
\\
\MEANINGFUN{\hat{MS}}: \Comm \rightarrow \nonterm{AbstractState} \rightarrow \nonterm{AbstractState} \\
\INFMS{r := e}{\INFSTORE, \INFMEM, \INFPC} = (\INFSTORE[r \mapsto \INFSTORE(r) \sqcup \INFME{e}{\INFSTORE}], \INFMEM, \INFPC + 1)  \\
\\
\INFMS{r_1 := load(r_2)}{\INFSTORE, \INFMEM, \INFPC} = \\
~~ \LET ~ V = \\
~~~~\CASEHEAD{\INFSTORE(r_2)} \\
~~~~~~\CASE{\top}{\top} \\
~~~~~~\CASE{P}{\top} \\
~~~~~~\CASE{S}{\{s_i\} \sqcup \INFSTORE(r_1)~\text{where}~s_i~\text{is fresh secret identifier}} \\
~~~~~~\CASE{U}{\{s_i\} \sqcup \INFSTORE(r_1)~\text{where}~s_i~\text{is fresh secret identifier}} \\
~~~~~~\CASE{E}{\INFMEM{(\INFSTORE(r_2))} \sqcup \INFSTORE(r_1)} \\
~~~~~~\CASE{N}{\INFMEM{(\INFSTORE(r_2))} \sqcup \INFSTORE(r_1)} \\
~~~~\CASEEND \\ 
~~\LETIN ~ (\INFSTORE[r_1 \mapsto V], \INFMEM, \INFPC + 1) \\
\\
\INFMS{store(r_1, r_2)}{\INFSTORE, \INFMEM, \INFPC} = (\INFSTORE, \INFMEM', \INFPC + 1) \\
~~\mbox{where } \INFMEM' = \forall ~ v \in \INFSTORE(r_1). \INFMEM[v \mapsto (\INFMEM(v) \sqcup \INFSTORE(r_2))] \\
\\
\INFMS{r_1 := is\_zero(r_2)}{\INFSTORE, \INFMEM, \INFPC} = \\
~~\IFHEAD{\INFSTORE(r_2) \subseteq N} \IFTHEN ~ \\
~~~~~~\CASEHEAD{\INFSTORE(r_2)} \\
~~~~~~~~\CASE{\{0\}}{(\INFSTORE[r_1 \mapsto \{0\}], \INFMEM, \INFPC + 1)} \\
~~~~~~~~\CASE{0 \not\in V}{(\INFSTORE[r_1 \mapsto\{1\}], \INFMEM, \INFPC + 1)} \\
~~~~~~~~\CASEELSE ~ (\INFSTORE[r_1 \mapsto \{0, 1\}], \INFMEM, \INFPC + 1) \\
~~~~~~\CASEEND \\
~~\IFELSE ~ (\INFSTORE[r_1 \mapsto \{0, 1\}], \INFMEM, \INFPC + 1) \\
\\
\INFMS{jmp(r_1, r_2)}{\INFSTORE, \INFMEM, \INFPC} = (\INFSTORE, \INFMEM, 
(\INFPC + 1) \cup \INFSTORE(r_2)) \\
\\
\end{array}
\end{displaymath}
\caption{The abstract semantics of REIL. The notation $\INFPC + 1$ is a short hand for $\{pc + 1 \mid pc \in \INFPC\}$.}
\label{fig:infsemantics}
\end{figure}

\F~\ref{fig:infsemantics} presents the definition of abstract semantics.
Function $\MEANINGFUN{\tilde{E}}$ evaluates a given abstract expression, and 
function $\MEANINGFUN{\tilde{MP}}$ takes a program 
table and a set of abstract states, and returns a new set of abstract states.
Function $\MEANINGFUN{\tilde{MS}}$ takes a program abstract state and 
a command, and computes a new abstract state.
As noted earlier, we assume pointers referring to the base memory address of program secrets 
have been pinpointed 
and initialized with special identifier $h \in U$.
Also, we follow the standard convention to 
perform ``weak update'' when updating the value sets of variables.

Note that $\MEANINGFUN{\tilde{E}}$ defines the evaluation rules of abstract expressions, and 
given two abstract values $f_1, f_2 \in \mathbf{AV}$, we define abstract 
operator $\INFOP$ in \F~\ref{fig:absop}.
The first six cases reasonbly over-approximate computations among 
$f \in \mathbf{AV}$ with different degrees of abstraction. The seventh case would be applied 
to compute two pieces of constants, while the last case is applied if no
other case can be matched. Indeed the last case concatenates two abstract values in
$\mathbf{AV}$ with an abstract operator. 

\begin{figure}[!th]
\begin{displaymath}
\begin{array}{l}
 v_1 ~ \INFOP ~ v_2 = \CASEHEAD{(v_1, v_2)} \\
 ~~~~\CASE{(v, \top)}{\top} \\
 ~~~~\CASE{(\top, v)}{\top} \\
 ~~~~\CASE{(p, s)}{\top} \mbox{ where } s \in S \\
 ~~~~\CASE{(s, p)}{\top} \mbox{ where } s \in S \\
 ~~~~\CASE{(p, v)}{p} \mbox{ where } v \not\in S\\
 ~~~~\CASE{(v, p)}{p} \mbox{ where } v \not\in S\\
 ~~~~\CASE{(n_1, n_2)}{n_1 \INFOP n_2} \\  
 ~~~~\CASE{(av_1, av_2)}{av_1 \INFOP av_2} \\  
 ~~\CASEEND
 \end{array}
 \end{displaymath}
 \caption{Definition of abstract binary operations where $v_1 \in \mathbf{AV}$
   and $v_2 \in \mathbf{AV}$.}
 \label{fig:absop}
 \end{figure}

$\MEANINGFUN{\tilde{MS}}$ defines memory operations within our abstract domain.

Memory loading yields a $\top$ as a safe over-estimation if the memory address
$addr \in P$ or $addr \in \top$. For memory address $addr \in S \cup U$, we create a fresh piece of secret
identifier $s_i$ to update the register. For the other cases, we 
use the recorded content
to perform a weak update. 

The abstract semantics of the store command ($store(r_1, r_2)$) is straightforward.
If $r_1$ stores the value $\top$, then the whole memory is updated to $\top$.
In this case, the program will terminate immediately. 
Otherwise, values of $r_2$ will be added to the memory location which is pointed to by $r_1$.

\subsubsection{Termination}
\label{subsec:termination}
One key reason to use the abstract interpretation framework is that it
guarantees the termination of the analysis within finite computation steps, even
though a precise modeling of program semantics may not be achievable. To do so,
the abstract domain is required to form a lattice of finite height, and abstract
transfer functions defined over the abstract domain needs to be monotonic.

It is easy to see that all the abstract transfer functions defined in
$\mathbf{SAS}$ are monotonic (see \S~\ref{subsec:abs-semantics}), and we have
also shown that \sas\ forms a lattice of finite height by defining the
$\text{BOU}$ function to bound the size of each element (i.e., $\{f \mid f \in
\mathbf{AV}\}$) in \sas. Hence, the termination guarantee of \sas\ is naturally
established.

\subsubsection{Abstract Program Traces}
A program abstract trace is a finite non-empty alternating sequence of 
abstract program states and actions,
 $\hat{s_0}, \theta(\hat{s_0}.pc) \hat{s_1} \theta(\hat{s_1}.pc) \cdots \hat{s_n}$, 
 such that $\forall i \in \{0, \cdots, n - 1\}$, $\hat{s_{i + 1}} = \INFMS{a}{\hat{s_i}}$.
Compared with the concrete program traces defined in \secref{sec:cpt}, 
the abstract data represents properties of the concrete one. 
For given abstract input data, the program execution generates an abstract computation tree.

Due to the abstraction, a terminating concrete computation may terminate, but its abstract one may not.
Following Schmidt's work~\cite{schmidt1998trace}, we 
build a \emph{regular} computation tree~\cite{cousineau1979rational}, so that 
trace building is terminated at the nodes that have been seen earlier in the path~\cite{schmidt1998trace}.

\subsubsection{Abstract Path-based Collecting Semantics}
The abstract forward path-based collecting semantics maps an abstract 
program counter to the set of abstract execution paths that 
end with the program 
counter : $\COLLAP: \textit{AbstractPC} \rightarrow \powerset{\textit{AbstractTraces}}$, which is defined as $
\COLLAP(\hat{pc}) = \Indexed{\hat{t}}{\INFPC}$, where $\hat{t}$ is a computation tree.

\subsection{Soundness}
\label{sec:soundness}
In this section we give the soundness proof of our abstraction. We start by
defining the abstraction and concretization functions. We then present the
proofs accordingly.

\subsubsection{Abstraction and Concretization Functions}

\begin{figure}
\begin{displaymath}
\begin{array}{l}
\INFABS(v : \V) = \\
~~\CASEHEAD{v} \\
~~~~ \CASE{(n, hi)}{\{s\}} \\
~~~~ \CASE{(n, lo)}{\{n\}} \mbox{ where } n \in \IM\\
~~~~ \CASE{(n, lo)}{\{u\}} \mbox{ where } n \in \Header\\
~~~~ \CASE{(n, lo)}{\{e\}} \mbox{ where } n \in \Stack \\
~~~~ \CASE{(n, lo)}{\{p\}} \mbox{ where } n \in \mathbb{Z} \land n \not\in \Header \cup \IM \cup \Stack\\
~~\CASEEND \\
\\
\INFCONCRETE^{\circ}(\tilde{v} : \mathbf{AV}) = 
~~\CASEHEAD{\tilde{v}} \\
~~~~ \CASE{n}{\{(n, lo) \mid  n \in \IM \}} \\
~~~~ \CASE{p}{\{(n, lo) \mid  n \in \mathbb{Z} \}} \\
~~~~ \CASE{u}{\{(n, lo) \mid  n \in \Header\}} \\
~~~~ \CASE{e}{\{(n, lo) \mid  n \in \Stack \}} \\
~~~~ \CASE{s}{\{(n, hi) \mid n \in \mathbb{Z}\}} \\
~~~~ \CASE{\top}{\{(n, t) \mid n \in \mathbb{Z}, t \in \{lo, hi\}\}} \\
~~\CASEEND \\
\\
\gamma(S) := \bigcup_{f \in S}\gamma^{\circ}(f)
\\
\end{array}
\end{displaymath}
\caption{The definition of abstraction function $\INFABS$ and concretization function $\INFCONCRETE$. 
$\Stack$ defines the set of stack memory addresses, and 
$s \in S$, $u \in U$, $n \in N$, $e \in E$.}
\label{fig:abs-conc}
\end{figure}

\F~\ref{fig:abs-conc} defines the abstraction function $\INFABS$ and the
concretization $\INFCONCRETE$ function. The abstraction function $\INFABS$ is
defined over concrete values $\alpha: \V \rightarrow \textbf{SAS}$. As shown in
\F~\ref{fig:abs-conc}, each piece of secret is translated into a singleton set
of $\{s_i\}$, and concrete data is preserved. Memory addresses towards x86
memory stack is modeled as $\{e\}$ where $e \in E$, and pointers referring to
the locations of program secrets are mapped into $\{u\}$. As for the other cases
where a non-secret value does not belong to any of the aforementioned category,
it would be lifted into $\{p\}$.

We also define the concretization function $\gamma: \textbf{SAS} \rightarrow
\powerset{Val}$ which casts a set of abstract values into a set of concrete
values. Here we start by defining a gadget concretization function
$\gamma^{\circ} : \mathbf{AV} \rightarrow \powerset{\V}$ which operates on
individual abstract value.
As shown in \F~\ref{fig:abs-conc}, function $\gamma^{\circ}$ preserves the
constant data back and forth, and $e$ is translated into the stack memory
addresses in $\Stack$. Public symbol $p$ over-approximates all the non-secret
values in $\powerset{\V}$. $s_{i}$ is converted into its corresponding secret
value in \V, while $\top$ maps to the whole set of values in $\powerset{\V}$.
Note that for most cases the output of $\gamma$ is a singleton set.

\subsubsection{Soundness Of Abstract Values}
The following lemma states the abstraction and concretization functions on values are correct.
As defined earlier, given a concrete value $v$, $\INFABS(v)$ 
computes $v$'s abstract value, 
and $\INFCONCRETE(\INFABS(v))$ maps the abstract value back to concrete values. A correct approximation means that $v \in \INFCONCRETE(\INFABS(v))$.
\begin{lemma} \label{lem:gamma-alpha} 
    Let $\V$ be the set of concrete values. 
    Let $\INFABS$ and $\INFCONCRETE$ be the abstract function and the 
    concretization function respectively. Then $\forall v \in \V, v \in \INFCONCRETE(\INFABS(v))$ holds.
\end{lemma}
\begin{proof} Let $s$ be a an arbitrary program state. Let $v \in \V$ be an arbitrary value in $s$. There are five cases.
\begin{enumerate}
\item $(n, lo)$, where $n \in \IM$. In this case, we need to show that $(n, lo) \in \INFCONCRETE(\INFABS(n, lo))$. By the definition in \figref{fig:abs-conc},
$\INFCONCRETE(\INFABS((n, lo))) = \INFCONCRETE(n)$ and $\INFCONCRETE(n) = \{(n, lo) \mid n \in \IM\}$.

\item $(n, lo)$, where $n \in \Header$. In this case, we need to show that $(n, lo) \in \INFCONCRETE(\INFABS(n, lo))$. By the definition in \figref{fig:abs-conc}, $\INFCONCRETE(\INFABS(n, lo)) = \INFCONCRETE(u)$ 
for some $u \in U$, and $\INFCONCRETE(u) = \{(n, lo) \mid n \in \Header \}$.

\item $(n, lo)$, where $n \in \mathbb{Z} \mbox{ and } n \not\in \Header \cup \IM \cup \Stack$. In this case, we need to show that $(n, lo) \in \INFCONCRETE(\INFABS(n, lo))$. By the definition in \figref{fig:abs-conc}, $\INFCONCRETE(\INFABS(n, lo)) = \INFCONCRETE(p)$ 
for $p \in P$, and $\INFCONCRETE(p) = \{(n, lo) \mid n \in \mathbb{Z} \}$.

\item $(n, lo)$, where $n \in \Stack$. In this case, we need to show 
that $(n, lo) \in \INFCONCRETE(\INFABS(n, lo))$. By the definition 
in \figref{fig:abs-conc}, $\INFCONCRETE(\INFABS(n, lo)) = \INFCONCRETE(e)$ for 
some $e \in E$, and $\INFCONCRETE(\epsilon) = \{(n, lo) \mid n \in \Stack\}$.

\item $(n, hi)$. In this case, we need to show 
that $(n, hi) \in \INFCONCRETE(\INFABS(n, hi))$. By the definition in \figref{fig:abs-conc},
$\INFCONCRETE(\INFABS((n, hi))) = \INFCONCRETE(s)$ for 
some $s \in S$ and $\INFCONCRETE(s) = \{(n, hi) \mid n \in \mathbb{Z}\}$.

\end{enumerate}
\end{proof}

The following lemma shows that the semantics of the abstract operator defined in \figref{fig:absop} is sound.
\begin{lemma} \label{lem:op} Let $\V$ be the set of concrete values. 
    Let $\INFABS$ and $\INFCONCRETE$ be the abstract function and the concretization function, respectively. 
    Let $\Diamond$ be a concrete binary operator. 
    Then $\forall v_1 \in \V, v_2 \in \V, 
    (v_1 \Diamond v_2) \in \INFCONCRETE(\INFABS(v_1) \INFOP \INFABS(v_2))$ holds.
\end{lemma}
\begin{proof} Let $v_1 \in \V$ and $v_2 \in \V$ be arbitrary. And Let $\INFOP$ be $\INFABS(\Diamond)$. 
\begin{enumerate}
\item ($(n_1, lo), (n_2, lo)$), where $n_1 \in \IM$ and $n_2 \in \IM$. 
In this case, we need to show that $(n_1, lo) \MO{\Diamond} (n_2, lo) \in \INFCONCRETE(\INFABS(n_1, lo) \INFOP \INFABS(n_2, lo))$. By the definition of the concrete semantics (\figref{fig:reil-semantics}),  we know that $(n_1, lo)\Diamond (n_2, lo) = (n_1 \MO{\Diamond} n_2, lo)$.
By the definition in \figref{fig:abs-conc} twice and by the definition in \figref{fig:absop}, we know $\INFABS(n_1, lo) \INFOP \INFABS(n_2, lo) = n_1 \INFOP n_2$.  
Thus, by the definition in \figref{fig:abs-conc}, we get $(n_1 \Diamond n_2, lo) \in \INFCONCRETE(n_1 \INFOP n_2)$.
\item ($(n_1, lo), (n_2, hi)$), where $n_1 \in \IM$ and $n_2 \in \mathbb{Z}$. 
In this case, we need to show that $(n_1, lo) \MO{\Diamond} (n_2, hi) \in \INFCONCRETE(\INFABS(n_1, lo) \INFOP \INFABS(n_2, hi))$.
By the definition of the concrete semantics (\figref{fig:reil-semantics}), we know that $(n_1, lo)\Diamond (n_2, hi) = (n_1 \MO{\Diamond} n_2, hi)$.
By the definition in \figref{fig:abs-conc} twice, we know $\INFABS(n_1, lo) \INFOP \INFABS(n_2, hi) = n_1 \INFOP s$, for some $s \in S$. 
By the definition of the abstract value in \figref{fig:abs-syntax}, we know $n_1 \INFOP s \in S$.
By the definition in \figref{fig:abs-conc}, we have $\INFCONCRETE(n_1 \INFOP s) = \{(n, hi) \mid n \in \mathbb{Z}\}$.
Thus, we get $(n_1 \MO{\Diamond} n_2, hi) \in \INFCONCRETE(n_1 \INFOP s)$.

\item ($(n_1, lo), (n_2, lo)$), where $n_1 \in \IM$ and $n_2 \in \mathbb{U}$. 
In this case, we need to show that $(n_1, lo) \MO{\Diamond} (n_2, hi) \in \INFCONCRETE(\INFABS(n_1, lo) \INFOP (n_2, lo))$.
By the definition of the concrete semantics (\figref{fig:reil-semantics}), we know $(n_1, lo) \Diamond (n_2, lo) = (n_1 \MO{\Diamond} n_2, lo)$.
By the definition in \figref{fig:abs-conc} twice, we know $\INFABS(n_1, lo) \INFOP \INFABS(u, lo) = n_1 \INFOP u$, for some $u \in U$.
By the definition of the abstract value in \figref{fig:abs-syntax}, we know $n_1 \INFOP u \in U$.
By the definition in \figref{fig:abs-conc}, we have $\INFCONCRETE(n_1 \INFOP u) = \{(n, lo) \mid n \in \mathbb{U}\}$. Thus, by the definition in \figref{fig:abs-conc}, we get $(n_1 \MO{\Diamond} n_2, lo) \in \INFCONCRETE(n_1 \INFOP u)$.

\item ($(n_1, lo), (n_2, lo)$), where $n_1 \in \IM$, and $n_2 \in \Stack$.
In this case, we need to show that $(n_1, lo) \MO{\Diamond} (n_2, lo) \in \INFCONCRETE(\INFABS(n_1, lo) \INFOP (n_2, lo))$.
By the definition of the concrete semantics (\figref{fig:reil-semantics}), we know $(n_1, lo) \Diamond (n_2, lo) = (n_1 \MO{\Diamond} n_2, lo)$.
By the definition in \figref{fig:abs-conc} twice, we know $\INFABS(n_1, lo) \INFOP \INFABS(\epsilon, lo) = n_1 \INFOP \epsilon$, 
for some $e \in E$.
By the definition of the abstract value in \figref{fig:abs-syntax}, we know that $n_1 \INFOP e \in E$.
By the definition in \figref{fig:abs-conc}, we have $\INFCONCRETE(n_1 \INFOP e) = \{(n, lo) \mid n \in E\}$. Thus, by the definition in \figref{fig:abs-conc}, we get $(n_1 \MO{\Diamond} n_2, lo) \in \INFCONCRETE(n_1 \INFOP e)$.

\item ($(n_1, lo), (n_2, lo)$), where $n_1 \in \IM$ and $n_2 \not\in (\mathbb{U} \cup \IM \cup \Stack)$. 
In this case, we need to show that $(n_1, lo) \MO{\Diamond} (n_2, lo) \in \INFCONCRETE(\INFABS(n_1, lo) \INFOP \INFABS(n_2, lo))$.
By the definition of the concrete semantics (\figref{fig:reil-semantics}), we know $(n_1, lo) \Diamond (n_2, lo) = (n_1 \MO{\Diamond} n_2, lo)$.
By the definition in \figref{fig:abs-conc} twice, we know $\INFABS(n_1, lo) \INFOP \INFABS(n_2, lo) = n_1 \INFOP p$, for $p \in P$.
By the definition in \figref{fig:absop}, we know $n_1 \INFOP p \in P$.
Thus, by the definition in \figref{fig:abs-conc}, we get $(n_1 \MO{\Diamond} n_2, lo) \in \INFCONCRETE(p)$, for $p \in P$.

\item ($(n_1, lo), (n_2, lo)$), where $n_1 \in \mathbb{U}$ and $n_2 \not\in (\mathbb{U} \cup \IM \cup \Stack)$. 
In this case, we need to show that $(n_1, lo) \Diamond (n_2, lo) \in \INFCONCRETE(\INFABS(n_1, lo) \INFOP \INFABS(n_2, lo))$.
By the definition of the concrete semantics (\figref{fig:reil-semantics}), we know $(n_1, lo) \Diamond (n_2, lo) = (n_1 \MO{\Diamond} n_2, lo)$.
By the definition in \figref{fig:abs-conc} twice, we know $\INFABS(n_1, lo) \INFOP \INFABS(n_2, lo) = u \INFOP p$, for some $u \in U$ and $p \in P$.
By the definition in \figref{fig:absop}, we know $u \INFOP p \in P$.
Thus, by the definition in \figref{fig:abs-conc}, we get $(n_1 \MO{\Diamond} n_2, lo) \in \INFCONCRETE(p)$, for $p \in P$.

\item ($(n_1, lo), (n_2, lo)$), where $n_1 \in \{esp\}$ and $n_2 \not\in (\mathbb{U} \cup \IM \cup \Stack)$. 
In this case, we need to show that $(n_1, lo) \Diamond (n_2, lo) \in \INFCONCRETE(\INFABS(n_1, lo) \INFOP \INFABS(n_2, lo))$.
By the definition of the concrete semantics (\figref{fig:reil-semantics}), we know $(n_1, lo) \Diamond (n_2, lo) = (n_1 \MO{\Diamond} n_2, lo)$.
By the definition in \figref{fig:abs-conc} twice, we know $\INFABS(n_1, lo) \INFOP \INFABS(n_2, lo) = \epsilon \INFOP p$, for some $\epsilon \in \mathcal{E}$ and $p \in P$.
By the definition in \figref{fig:absop}, we know $e \INFOP p \in P$.
Thus, by the definition in \figref{fig:abs-conc}, we get $(n_1 \MO{\Diamond} n_2, lo) \in \INFCONCRETE(p)$, for $p \in P$.

\item (($n_1, lo), (n_2, hi)$), where $n_1 \in U$. 
In this case, we need to show that $(n_1, lo) \Diamond (n_2, hi) \in \INFCONCRETE(\INFABS(n_1, lo) \INFOP (n_2, hi))$.
By the definition of the concrete semantics (\figref{fig:reil-semantics}), $(n_1, lo) \MO{\Diamond} (n_2, hi) = (n_1 \Diamond n_2, hi)$.
By the definition in \figref{fig:abs-conc} twice, $\INFABS(n_1, lo) \INFOP \INFABS(n_2, hi) = u \INFOP s$, for some $u \in U$ and $s \in S$.
By the definition of the abstract value (\figref{fig:abs-syntax}), we know $u \INFOP s \in S$.
By the definition in \figref{fig:abs-conc}, we have $\INFCONCRETE(u \INFOP s) = \{(n, hi) \mid n \in \mathbb{Z}\}$. 
Thus, by the definition in \figref{fig:abs-conc}, we get $(n_1 \MO{\Diamond} n_2, hi) \in \INFCONCRETE(u \INFOP s)$.

\item (($n_1, lo), (n_2, hi)$), where $n_1 \in \Stack$. 
In this case, we need to show that $(n_1, lo) \Diamond (n_2, hi) \in \INFCONCRETE(\INFABS(n_1, lo) \INFOP (n_2, hi))$.
By the definition of the concrete semantics (\figref{fig:reil-semantics}), $(n_1, lo) \MO{\Diamond} (n_2, hi) = (n_1 \Diamond n_2, hi)$.
By the definition in \figref{fig:abs-conc} twice, $\INFABS(n_1, lo) \INFOP \INFABS(n_2, hi) = e \INFOP s$, for some $e \in \mathcal{E}$ and $s \in S$.
By the definition of the abstract value (\figref{fig:abs-syntax}), we know $e \INFOP s \in S$.
By the definition in \figref{fig:abs-conc}, we have $\INFCONCRETE(e \INFOP s) = \{(n, hi) \mid n \in \mathbb{Z}\}$. 
Thus, by the definition in \figref{fig:abs-conc}, we get $(n_1 \MO{\Diamond} n_2, hi) \in \INFCONCRETE(e \INFOP s)$.

\item (($n_1, lo), (n_2, hi)$), where $n_1 \not\in (\mathbb{U} \cup \IM \cup \Stack)$. 
In this case, we need to show that $(n_1, lo) \Diamond (n_2, hi) \in \INFCONCRETE(\INFABS(n_1, lo) \INFOP (n_2, hi))$.
By the definition of the concrete semantics (\figref{fig:reil-semantics}), $(n_1, lo) \MO{\Diamond} (n_2, hi) = (n_1 \MO{\Diamond} n_2, hi)$.
By the definition in \figref{fig:abs-conc} twice, $\INFABS(n_1, lo) \INFOP \INFABS(n_2, hi) = p \INFOP s$, for some $p \in P$ and $s \in S$.
By the definition of the abstract operator (\figref{fig:absop}), we know $p \INFOP s \in \top$.
By the definition in \figref{fig:abs-conc}, we have $\INFCONCRETE(p \INFOP s) = \{(n, t) \mid n \in \mathbb{Z}, t \in \{lo, hi\}\}$. 
Thus, by the definition in \figref{fig:abs-conc}, we get $(n_1 \MO{\Diamond} n_2, hi) \in \INFCONCRETE(p \INFOP s)$.

\item (($n_1, hi), (n_2, hi)$). 
In this case, we need to show that $(n_1, hi) \Diamond (n_2, hi) \in \INFCONCRETE(\INFABS(n_1, hi) \INFOP (n_2, hi))$.
By the definition of the concrete semantics (\figref{fig:reil-semantics}), $(n_1, hi) \MO{\Diamond} (n_2, hi) = (n_1 \MO{\Diamond} n_2, hi)$.
By the definition in \figref{fig:abs-conc} twice, $\INFABS(n_1, hi) \INFOP \INFABS(n_2, hi) = s_1 \INFOP s_2$, for some $s_1 \in S$ and $s_2 \in S$.
By the definition of the abstract value (\figref{fig:abs-syntax}), we know $ s_1 \INFOP s_2 \in S$.
By the definition in \figref{fig:abs-conc}, we have $\INFCONCRETE(s_1 \INFOP s_2) = \{(n, hi) \mid n \in \mathbb{Z}\}$. 
Thus, by the definition in \figref{fig:abs-conc}, we get $(n_1 \MO{\Diamond} n_2, hi) \in \INFCONCRETE(s_1 \INFOP s_2)$.

\item (($n_1, lo), (n_2, lo)$), where $n_1 \mbox{ or } n_2 \not\in (\mathbb{U} \cup \IM \cup \Stack)$. 
In this case, we need to show that $(n_1, lo) \Diamond (n_2, lo) \in \INFCONCRETE(\INFABS(n_1, lo) \INFOP (n_2, lo))$.
By the definition of the concrete semantics (\figref{fig:reil-semantics}), $(n_1, lo) \MO{\Diamond} (n_2, lo) = (n_1 \MO{\Diamond} n_2, lo)$.
By the definition in \figref{fig:abs-conc} twice, $\INFABS(n_1, lo) \INFOP \INFABS(n_2, lo) = p \INFOP p$, where $p \in P$.
By the definition of the abstract operator (\figref{fig:absop}), we know $p \INFOP p \in P$.
By the definition in \figref{fig:abs-conc}, we have $\INFCONCRETE(p) = \{(n, lo) \mid n \in \mathbb{Z}, n \not\in (\mathbb{U} \cup \IM \cup \Stack)\}$. 
Thus, by the definition in \figref{fig:abs-conc}, we get $(n_1 \MO{\Diamond} n_2, lo) \in \INFCONCRETE(p)$.

\item (($n_1, lo), (n_2, lo)$), where $n_1 \mbox{ and } n_2 \in U$. 
In this case, we need to show that $(n_1, lo) \Diamond (n_2, lo) \in \INFCONCRETE(\INFABS(n_1, lo) \INFOP (n_2, lo))$.
By the definition of the concrete semantics (\figref{fig:reil-semantics}), $(n_1, lo) \MO{\Diamond} (n_2, lo) = (n_1 \MO{\Diamond} n_2, lo)$.
By the definition in \figref{fig:abs-conc} twice, $\INFABS(n_1, lo) \INFOP \INFABS(n_2, lo) = u_1 \INFOP u_2$, for some $u_1 \in U$ and $u_2 \in U$.
By the definition of the abstract value (\figref{fig:abs-syntax}), we know $u_1 \INFOP u_2 \in U$.
By the definition in \figref{fig:abs-conc}, we have $\INFCONCRETE(u) = \{(n, lo) \mid n \in \mathbb{U}\}$. 
Thus, by the definition in \figref{fig:abs-conc}, we get $(n_1 \MO{\Diamond} n_2, lo) \in \INFCONCRETE(u_1 \INFOP u_2)$.

\end{enumerate}

\end{proof}

\subsubsection{Soundness of Abstract Semantics}
In order to avoid introducing further meta-functions to our notation, we will
reuse the notation for the abstract and concretization functions.

Before the definitions, we introduce some common notations used in this section. Let $X$ and $Y$ be two sets. Let $f$ be a function that maps from $X$ to $Y$. 
Then the set $X$ is the \emph{domain} of $f$, written $\DOM(f) = X$; the set $Y$ is the \emph{codomain} of $f$, written $\CODOM(f) = Y$.

The following defines the abstraction functions for the store, the memory and
the program counter.

\begin{definition} \label{def:alpha-state} Let $\sigma$ be a store. Then $\INFABS(\sigma)$ generates the abstract store, such that
\[\INFABS(\sigma) = \{(x, v) \mid  x \in \DOM(\sigma), v = \{\INFABS(\sigma(x))\}\}. \]
Let $\MEM$ be a memory. Then $\INFABS(\MEM)$ generates the abstract memory, such that
\[\INFABS(\MEM) = \{(l, v) \mid l' \in \DOM(\MEM), l = \{\INFABS(l')\}, v =  \{\INFABS(\MEM(l'))\} \}).\]
Let $pc$ be a program counter. Then $\INFABS(pc)$ generates the abstract program counter, such that $\INFABS(pc) = \{pc\}$.

Let $(\sigma, \MEM, pc)$ be a program state. Then $\INFABS(\sigma, \MEM, pc)$
generates the abstract state, such that $\INFABS(\sigma,
\MEM, pc) = (\INFABS(\sigma), \INFABS(\MEM), \INFABS(pc))$.
\end{definition}

The following defines the concretization functions for the store, the memory and the program counter.
\begin{definition} \label{def:gamma-state}
Let $\INFSTORE$ be an abstract store. Then $\INFCONCRETE(\INFSTORE)$ generates the set of possible concrete stores, such that 
\[\INFCONCRETE(\INFSTORE) = \{\sigma \mid  x \in \DOM(\INFSTORE), \sigma(x) \in  \INFCONCRETE(\INFSTORE(x)) \}).\]
Let $\INFMEM$ be an abstract memory. Then $\INFCONCRETE(\INFMEM)$ generates the set of possible concrete memories, such that
\[\INFCONCRETE(\INFMEM) = \{\MEM \mid l \in \DOM(\INFMEM), \MEM(l) \in  \INFCONCRETE(\MEM(l)) \}).\]
Let $\INFPC$ be an abstract program counter. Then $\INFCONCRETE(\INFPC)$ generates the set of possible program counters, such that $\INFCONCRETE(\INFPC) = \{pc \mid pc \in \INFPC\}$.

Let $(\INFSTORE, \INFMEM, \INFPC)$ be an abstract state. Then $\INFCONCRETE(\INFSTORE, \INFMEM, \INFPC)$ generates the set of possible concrete states, such that $\INFCONCRETE(\INFSTORE, \INFMEM, \INFPC) = \{(\sigma, \MEM, pc) \mid \sigma \in \INFCONCRETE(\INFSTORE), \MEM \in \INFCONCRETE(\INFMEM), pc \in \INFPC\}$.
\end{definition}

A concrete state $s$ is safely approximated by an abstract state $\hat{s}$, written $s \APPROX \hat{s}$, if $s \in \gamma(\hat{s})$, where $\gamma$ is a concretization function.

\begin{lemma} \label{lem:stateappro} Let $\alpha$ and $\gamma$ be the abstract and concretization functions defined in \defref{def:alpha-state} and \defref{def:gamma-state}. Let $s$ be a concrete state, then $s \APPROX \alpha(s)$. 
\end{lemma}

The following lemma shows that definition of the abstract evaluation on expressions is correct. 
\begin{lemma} \label{lem:exprappox} Let $e$ be an expression, and $\INFSTORE$ be an abstract store. Let $\INFCONCRETE$ be the concretization function for $\INFSTORE$ defined in \defref{def:gamma-state}. 
Let $\sigma$ be a program store, such that $\sigma \in \INFCONCRETE(\INFSTORE)$. Then $\ME{e}{\sigma} \in \INFCONCRETE(\INFME{e}{\INFSTORE})$
\end{lemma}
\begin{proof} Let $e$, $\INFSTORE$, $\INFCONCRETE$ and $\sigma$ be given. The proof is by the induction on the expression's structure. The proof is done in calculational style, starting from the concrete evaluation. There are two base cases.
\begin{enumerate}
\item ($x$) In this case, $e$ has the form $x$. 
\begin{calculation}
\FORMULA{\ME{x}{\sigma}}
\REASON{=}{the concrete semantics in \figref{fig:reil-semantics}}
\FORMULA{\sigma(x)}
\REASON{\in}{\lemref{lem:gamma-alpha}}
\FORMULA{\INFCONCRETE(\INFSTORE(x))}
\REASON{=}{the abstract semantics in \figref{fig:infsemantics}}
\FORMULA{\INFCONCRETE(\INFME{x}{\INFSTORE})}
\end{calculation}

\item ($n$) In this case, $e$ has the form $n$. 
\begin{calculation}
\FORMULA{\ME{n}{\sigma}}
\REASON{=}{the semantics in \figref{fig:reil-semantics}}
\FORMULA{(\mathcal{N}\synbracket{n}, lo)}
\REASON{\in}{\lemref{lem:gamma-alpha}}
\FORMULA{\INFCONCRETE(n)}
\REASON{=}{the semantics in \figref{fig:infsemantics}}
\FORMULA{\INFCONCRETE(\INFME{n}{\INFSTORE})}
\end{calculation}
\end{enumerate}

The inductive hypothesis is that for all subexpressions $e_i$, $\ME{e_i}{\sigma} \in \INFCONCRETE(\INFME{e_i}{\INFSTORE})$.
There is one inductive case, where $e$ is ($e_1 \Diamond e_2$).
\begin{calculation}
\FORMULA{\ME{e_1 \Diamond e_2}{\sigma}}
\REASON{=}{the semantics in \figref{fig:reil-semantics}}
\FORMULA{\ME{e_1}{\sigma} \Diamond \ME{e_2}{\sigma}}
\REASON{\Rightarrow}{inductive hypothesis}
\FORMULA{\ME{e_1}{\sigma} \in \INFCONCRETE(\INFME{e_1}{\INFSTORE}) \mbox{ and } \ME{e_2}{\sigma} \in \INFCONCRETE(\INFME{e_2}{\INFSTORE})}
\REASON{\Rightarrow}{assumption}
\FORMULA{\INFME{e_1}{\INFSTORE} = \INFABS(\ME{e_1}{\sigma}) \mbox{ and } \INFME{e_2}{\INFSTORE} = \INFABS(\ME{e_2}{\sigma})}
\REASON{\Rightarrow}{\lemref{lem:op}}
\FORMULA{\ME{e_1}{\sigma} \Diamond \ME{e_2}{\sigma} \in \INFCONCRETE(\INFME{e_1}{\INFSTORE} \INFOP \INFME{e_2}{\INFSTORE})}
\REASON{=}{the semantics in \figref{fig:infsemantics}}
\FORMULA{\ME{e_1}{\sigma} \Diamond \ME{e_2}{\sigma} \in \INFCONCRETE(\INFME{e_1 \INFOP e_2}{\INFSTORE})}
\end{calculation}
\end{proof}

The rest of the section shows that the abstract path-based collecting semantics is a safe modeling of the concrete semantics. 
The following definition con-inductively defines that the abstract computation trace safely approximates the concrete computation trace.
\begin{definition} \label{def:traceappro} A concrete computation trace, $tr_c$ is \emph{safely approximated} by an abstract computation trace, $tr_a$, written $tr_c \APPROX tr_a$, if and only if 
\begin{enumerate}
\item $\ROOT{tr_c} \APPROX \ROOT{tr_a}$
\item If $\{s_i = \MS{a_i}{\ROOT{tr_c}}\}_{i \in I}$ is the set of all possible transitions from $\ROOT{tr_c}$, and for each $i$, $a_i$, if $s_i = \MS{a_i}{\ROOT{tr_c}}$, then there exists a transition $\hat{s_j} = \INFMS{a_i}{\ROOT{tr_a}}$, such that $s_i \APPROX \hat{s_j}$.
\end{enumerate}
\end{definition}

\begin{theorem} Let $tr_c$ be a concrete computation trace, and let $tr_a$ be its abstract computation trace. Then $tr_c$ is safely approximated by $tr_a$.
\end{theorem}
\begin{proof} Let $tr_c$ be a concrete program trace, and let $s_0 = \ROOT{tr_c}$. 
Let $tr_a$ be the abstract computation trace of $tr_c$, and $\alpha(s_0) = \ROOT{tr_a}$. According to the definition of the safe approximation (\defref{def:traceappro}), there are two proof obligations.

Obligation 1: $\ROOT{t_c} \APPROX \ROOT{t_a}$. This is discharged by \lemref{lem:stateappro}, where $s_0 \APPROX \hat{s_0}$.

Obligation 2: Assume $s_0 \APPROX \hat{s_0}$. Let $s_0 = (\sigma_0, \MEM_0, pc_0)$, and let $\hat{s_0} = (\INFSTORE_0, \INFMEM_0, \INFPC_0)$. We enumerate all the possible transitions $a$.
\begin{enumerate}
\item ($r := e$) 
\begin{calculation}
\FORMULA{\MS{r := e}{\sigma_0, \MEM_0, pc_0}}
\REASON{=}{the concrete semantics defined in \figref{fig:reil-semantics}}
\FORMULA{(\sigma_0[r \mapsto \ME{e}{\sigma_0}], \MEM_0, pc_0 + 1)}
\end{calculation}
\begin{calculation}
\FORMULA{\INFMS{r := e}{\INFSTORE_0, \INFMEM_0, \INFPC_0}}
\REASON{=}{the abstract semantics defined in \figref{fig:infsemantics}}
\FORMULA{(\INFSTORE_0[r \mapsto \INFME{e}{\INFSTORE_0}], \INFMEM_0, \INFPC + 1)}
\end{calculation}

Because  $\sigma_0 \APPROX \INFSTORE_0$ by the assumption, by \lemref{lem:exprappox}, $\sigma_0[r \mapsto \ME{e}{\sigma_0}] \APPROX \INFSTORE[r \mapsto \INFME{e}{\INFSTORE_0}]$.
As $pc_0 \APPROX \INFPC_0$, $pc_0 \in \INFPC_0$. By the definition of $\INFPC_0 + 1$, we know that $(pc_0 + 1) \APPROX (\INFPC_0 + 1)$.
Thus, $(\sigma_0[r \mapsto \ME{e}{\sigma_0}], \MEM_0, pc_0 + 1) \APPROX (\INFSTORE_0[r \mapsto \INFME{e}{\INFSTORE_0}], \INFMEM_0, \INFPC + 1)$.

\item ($r_1 := load(r_2)$) Let $v = \sigma_0(r_2)$.
There are three cases.
\begin{enumerate}
\item $v.n \in \Header$. 
\begin{calculation}
\FORMULA{\MS{r_1 := load(r_2)}{\sigma_0, \MEM_0, pc_0}}
\REASON{=}{the concrete semantics defined \figref{fig:reil-semantics}}
\FORMULA{(\sigma_0[r_1 \mapsto (\MEM_0(v.n).n, hi)], \MEM_0, pc_0 + 1)}
\end{calculation}

By the assumption $s_0 \APPROX \hat{s_0}$ and definition of the abstraction function in \figref{fig:abs-conc}, we know that $\alpha(v) = \{u\}$ for some $u \in U$, where $v.n \in \Header$.
\begin{calculation}
\FORMULA{\INFMS{r_1 := load(r_2)}{\INFSTORE_0, \MEM_0, \INFPC_0}}
\REASON{=}{the abstract semantics defined in \figref{fig:infsemantics}}
\FORMULA{(\INFSTORE_0[r_1 \mapsto (\{s\} \sqcup \INFSTORE_0(r_1))], \INFMEM_0, \INFPC_0 + 1)}
\end{calculation}
where $s = \INFMEM_0(u)$. We need to prove that for all possible value of $\INFSTORE_0(r_1)$, $\alpha(\MEM_0(v.n).n, hi) \subseteq (\{s\} \sqcup \INFSTORE_0(r_1))$, which is true by the soundness of $\sqcup$ (\thmref{thm:col}).

\item $v.l = hi$.
\begin{calculation}
\FORMULA{\MS{r_1 := load(r_2)}{\sigma_0, \MEM_0, pc_0}}
\REASON{=}{the concrete semantics defined \figref{fig:reil-semantics}}
\FORMULA{(\sigma_0[r_1 \mapsto \MEM_0(v.n)], \MEM_0, pc_0 + 1)}
\end{calculation}

By the assumption $s_0 \APPROX \hat{s_0}$ and definition of the abstraction function in \figref{fig:abs-conc}, we know that $\alpha(v) = \{s'\}$, for some $s'$. Because loading from a secret address gets a secret value, let $\INFMEM_0(s') = \{s\}$.
\begin{calculation}
\FORMULA{\INFMS{r_1 := load(r_2)}{\INFSTORE_0, \MEM_0, \INFPC_0}}
\REASON{=}{the abstract semantics defined in \figref{fig:infsemantics}}
\FORMULA{(\INFSTORE_0[r_1 \mapsto (\{s\} \sqcup \INFSTORE_0(r_1))], \INFMEM_0, \INFPC_0 + 1)}
\end{calculation}
We need to prove that for all possible value of $\INFSTORE_0(r_1)$, $\alpha(\MEM_0(v.n)) \subseteq (\INFMEM_0(s') \sqcup \INFSTORE_0(r_1))$, where $\alpha(v) = \{s'\}$ and $\INFMEM_0(s') = \{s\}$.
By the assumption $s_0 \APPROX \hat{s_0}$, $\alpha(\MEM_0) \subseteq \INFMEM_0$. Thus, $\alpha(\MEM_0(v.n)) \subseteq \INFMEM_0(\alpha(v))$.
By the soundness of $\sqcup$ (\thmref{thm:col}), $\{\alpha(\MEM_0(v.n))\} \subseteq (\INFMEM_0(s') \sqcup \INFSTORE_0(r_1))$ is true.

\item $v.n \not\in \Header$ and $v.l = lo$.
\begin{calculation}
\FORMULA{\MS{r_1 := load(r_2)}{\sigma_0, \MEM_0, pc_0}}
\REASON{=}{the concrete semantics defined \figref{fig:reil-semantics}}
\FORMULA{(\sigma_0[r_1 \mapsto \MEM_0(v.n)], \MEM_0, pc_0 + 1)}
\end{calculation}
There are two case:
\begin{enumerate}
\item $\alpha(v) = P$. By the abstract semantics defined in \figref{fig:infsemantics}, $r_1$ is updated to $\top$. As $\top$ sits on the top of the lattice, the result is sound.
\item $\alpha(v) = \epsilon$ or $\alpha(v) = n$.
 \begin{calculation}
\FORMULA{\INFMS{r_1 := load(r_2)}{\INFSTORE_0, \INFMEM_0, \INFPC_0}}
\REASON{=}{the abstract semantics defined in \figref{fig:infsemantics}}
\FORMULA{(\INFSTORE_0[r_1 \mapsto (\INFMEM(\INFSTORE(r_2)) \sqcup \INFSTORE_0(r_1))], \INFMEM_0, \INFPC_0 + 1)}
\end{calculation}
\end{enumerate}
We need to prove that for all possible value of $\INFSTORE_0(r_1)$, $\alpha(\MEM_0(\sigma_0(r_2).n)) \subseteq (\INFMEM_0(\INFSTORE_0(r_2)) \sqcup \INFSTORE_0(r_1))$.
By the assumption $s_0 \APPROX \hat{s_0}$ and definition of the abstraction function in \figref{fig:abs-conc}, $\alpha(\MEM_0) \subseteq \INFMEM_0$ and $\alpha(\sigma_0) \subseteq \INFSTORE_0$. Therefore, $\alpha(\sigma_0(r_2)) \subseteq \INFSTORE_0(r_2)$ and $\alpha(\MEM_0(\sigma_0(r_2).n) \subseteq \INFMEM_0(\INFSTORE_0(r_2))$.
By the soundness of $\sqcup$ (\thmref{thm:col}), $\alpha(\MEM_0(\sigma_0(r_2).n)) \subseteq (\INFMEM_0(\INFSTORE_0(r_2)) \sqcup \INFSTORE_0(r_1))$ is true.
\end{enumerate}

\item ($r_1 := is\_zero(r_2)$). Let $v = \sigma_0(r_2)$. There are four cases.
\begin{enumerate}
\item $v = (0, lo)$, where $0 \in \IM$.
\begin{calculation}
\FORMULA{\MS{r_1 := is\_zero(r_2)}{\sigma_0, \MEM_0, pc_0}}
\REASON{=}{the concrete semantics defined \figref{fig:reil-semantics}}
\FORMULA{(\sigma_0[r_1 \mapsto (1, lo)], \MEM_0, pc_0 + 1)}
\end{calculation}
By the assumption $s_0 \APPROX \hat{s_0}$ and definition of the abstraction function in \figref{fig:abs-conc}, we know that $\alpha(v) = \{0\}$.
\begin{calculation}
\FORMULA{\INFMS{r_1 := load(r_2)}{\INFSTORE_0, \INFMEM_0, \INFPC_0}}
\REASON{=}{the abstract semantics defined in \figref{fig:infsemantics}}
\FORMULA{(\INFSTORE_0[r_1 \mapsto \{1\}], \INFMEM_0, \INFPC_0 + 1)}
\end{calculation}
We need to prove that $\alpha((1, lo)) \subseteq \{1\}$, which is true by the definition of $\alpha$ in \figref{fig:abs-conc}.

\item $v = (n, lo)$, where $n \neq 0$ and $n \in \IM$.
\begin{calculation}
\FORMULA{\MS{r_1 := is\_zero(r_2)}{\sigma_0, \MEM_0, pc_0}}
\REASON{=}{the concrete semantics defined \figref{fig:reil-semantics}}
\FORMULA{(\sigma_0[r_1 \mapsto (0, lo)], \MEM_0, pc_0 + 1)}
\end{calculation}
By the assumption $s_0 \APPROX \hat{s_0}$ and definition of the abstraction function in \figref{fig:abs-conc}, we know that $\alpha(v) = \{n\}$, where $n \in IM$.
\begin{calculation}
\FORMULA{\INFMS{r_1 := load(r_2)}{\INFSTORE_0, \INFMEM_0, \INFPC_0}}
\REASON{=}{the abstract semantics defined in \figref{fig:infsemantics}}
\FORMULA{(\INFSTORE_0[r_1 \mapsto \{0\}], \INFMEM_0, \INFPC_0 + 1)}
\end{calculation}
We need to prove that $\alpha((0, lo)) \subseteq \{0, 1\}$, which is true by the definition of $\alpha$ in \figref{fig:abs-conc}.

\item $v = (0, l)$, where $l \in \{lo, hi\}$.
\begin{calculation}
\FORMULA{\MS{r_1 := is\_zero(r_2)}{\sigma_0, \MEM_0, pc_0}}
\REASON{=}{the concrete semantics defined \figref{fig:reil-semantics}}
\FORMULA{(\sigma_0[r_1 \mapsto (1, l)], \MEM_0, pc_0 + 1)}
\end{calculation}
By the assumption $s_0 \APPROX \hat{s_0}$ and definition of the abstraction function in \figref{fig:abs-conc}, we know that $\alpha((0, l))$ could be any value in $p$, $\Header$, $\mathcal{E}$ or $S$.
\begin{calculation}
\FORMULA{\INFMS{r_1 := load(r_2)}{\INFSTORE_0, \INFMEM_0, \INFPC_0}}
\REASON{=}{the abstract semantics defined in \figref{fig:infsemantics}}
\FORMULA{(\INFSTORE_0[r_1 \mapsto \{0, 1\}], \INFMEM_0, \INFPC_0 + 1)}
\end{calculation}
We need to prove that $\alpha((0, l)) \subseteq \{0, 1\}$, which is true by the definition of $\alpha$ in \figref{fig:abs-conc}.

\item $v = (n, l)$, where $n \neq 0$ and $l \in \{lo, hi\}$.
\begin{calculation}
\FORMULA{\MS{r_1 := is\_zero(r_2)}{\sigma_0, \MEM_0, pc_0}}
\REASON{=}{the concrete semantics defined \figref{fig:reil-semantics}}
\FORMULA{(\sigma_0[r_1 \mapsto (0, l)], \MEM_0, pc_0 + 1)}
\end{calculation}
By the assumption $s_0 \APPROX \hat{s_0}$ and definition of the abstraction function in \figref{fig:abs-conc}, we know that $\alpha((n, l))$, could be any value in $p$, $\Header$, $\mathcal{E}$ or $S$.
\begin{calculation}
\FORMULA{\INFMS{r_1 := load(r_2)}{\INFSTORE_0, \INFMEM_0, \INFPC_0}}
\REASON{=}{the abstract semantics defined in \figref{fig:infsemantics}}
\FORMULA{(\INFSTORE_0[r_1 \mapsto \{0, 1\}], \INFMEM_0, \INFPC_0 + 1)}
\end{calculation}
 We need to prove that $\alpha((n, l)) \subseteq \{0, 1\}$, which is true by the definition of $\alpha$ in \figref{fig:abs-conc}.
\end{enumerate}

\item ($store(r_1, r_2)$). Recall the concrete semantics as follows: 
\begin{calculation}
\FORMULA{\MS{store(r_1, r_2)}{\sigma_0, \MEM_0, pc_0}}
\REASON{=}{the concrete semantics defined in \figref{fig:reil-semantics}}
\FORMULA{(\sigma_0, \MEM_0[\sigma_0(r_1) \mapsto \sigma_0(r_2)], pc_0 + 1)}
\end{calculation}
And recall its abstract semantics as follows:
\begin{calculation}
\FORMULA{\INFMS{store(r_1, r_2)}{\INFSTORE_0, \INFMEM_0, \INFPC_0}}
\REASON{=}{the abstract semantics defined in \figref{fig:infsemantics}}
\FORMULA{(\INFSTORE, \INFMEM, \INFPC_0 + 1)}
\end{calculation}
where $\INFMEM = \forall ~ v \in \INFSTORE_0(r_1). \INFMEM_0[v \mapsto (\INFSTORE_0(r_1)\sqcup \INFSTORE_0(r_2))]$.

By the assumption $(\sigma_0, \MEM_0, pc_0) \APPROX (\INFSTORE_0, \INFMEM_0, \INFPC_0)$,
we need to show that (1) $\alpha(\sigma_0(r_1)) \subseteq \INFSTORE_0(r_1)$ and (2) $\alpha(\sigma_0(r_2)) \subseteq (\INFSTORE_0(r_1) \sqcup \INFSTORE_0(r_2))$.
The obligation (1) is discharged by the assumption. The obligation (2) is discharged by  \thmref{thm:col}.

\item ($jmp(r_1, r_2)$). There are two cases. 
\begin{enumerate}
\item $\sigma_0(r_1).n \neq 0$. Recall the concrete semantics as follows:
\begin{calculation}
\FORMULA{\MS{jmp(r_1, r_2)}{\sigma_0, \MEM_0, pc_0}}
\REASON{=}{the concrete semantics defined in \figref{fig:reil-semantics}}
\FORMULA{(\sigma_0, \MEM_0, \sigma_0(r_2).n)}
\end{calculation}
Recall the abstract semantics as follows:
\begin{calculation}
\FORMULA{\INFMS{jmp(r_1, r_2)}{\INFSTORE_0, \INFMEM_0, \INFPC_0}}
\REASON{=}{the abstract semantics defed in \figref{fig:infsemantics}}
\FORMULA{(\INFSTORE_0, \INFMEM_0, (\INFPC_0 + 1) \cup \INFSTORE_0(r_2))}
\end{calculation}
We need to show that $\alpha(\sigma_0(r_2)) \subseteq (\INFPC_0 + 1) \cup \INFSTORE_0(r_2))$, which is true by the assumption $(\sigma_0, \MEM_0, pc_0) \APPROX (\INFSTORE_0, \INFMEM_0, \INFPC_0)$.
\item $\sigma_0(r_1).n = 0$ Recall the concrete semantics as follows:
\begin{calculation}
\FORMULA{\MS{jmp(r_1, r_2)}{\sigma_0, \MEM_0, pc_0}}
\REASON{=}{the concrete semantics defined in \figref{fig:reil-semantics}}
\FORMULA{(\sigma_0, \MEM_0, pc_0 + 1)}
\end{calculation}
Recall the abstract semantics as follows:
\begin{calculation}
\FORMULA{\INFMS{jmp(r_1, r_2)}{\INFSTORE_0, \INFMEM_0, \INFPC_0}}
\REASON{=}{the abstract semantics defed in \figref{fig:infsemantics}}
\FORMULA{(\INFSTORE_0, \INFMEM_0, (\INFPC_0 + 1) \cup \INFSTORE_0(r_2))}
\end{calculation}
We need to show that $\alpha(pc_0 + 1) \subseteq (\INFPC_0 + 1) \cup \INFSTORE_0(r_2))$, which is true by the assumption $(\sigma_0, \MEM_0, pc_0) \APPROX (\INFSTORE_0, \INFMEM_0, \INFPC_0)$.
\end{enumerate}
\end{enumerate}
\end{proof}

\section{\revise{}{Evaluating Different Configurations of the $\text{BOU}$ Function}}
\label{sec:bou-n}
The definition of the $\text{BOU}$ function includes a parameter $N$ as the
maximum size of each abstract value set. \T~\ref{tab:n-bou} reports the
evaluation results of \caches\ with respect to different $N$. As expected, with
the increase of the allowed size, analyses took more time before reaching the
fixed point. Also, when the allowed size is small (i.e., $N$ is 1 or 10), the
value set of certain registers is lifted into $\{p\}$ rapidly
and terminates the analysis due to memory write accesses through $p$ (see
\S~\ref{subsec:design-information}; we terminate the analysis for memory access
of $p$ since it rewrites the whole memory). The full evaluation data in terms of
different configurations is available in 
\T~\ref{tab:size-1}--\ref{tab:size-100}.
%

\begin{table}[t]
  \centering
  \caption{Evaluating different configurations of $\texttt{BOU}$. When $N$ is
    set as 1 and 10, several analyses terminated before reaching the fixed point
    due to memory write accesses through the public symbol $p$. The full
    evaluation data in terms of each configuration can be found at
    \T~\ref{tab:size-1}--\ref{tab:size-100}.}
  \label{tab:n-bou}
\resizebox{0.85\linewidth}{!}{
  \begin{tabular}{c|c|c|c}
    \hline
    \textbf{Value of $N$} & \textbf{True Positive} & \textbf{False Positive} & \textbf{Processing Time (CPU Seconds)} \\
    \hline
    \textbf{1} & N/A & N/A & N/A \\
    \hline
    \textbf{10} & 167 & 1 & 584.5 \\
    \hline
    \textbf{25} & 207 & 1 & 1,446.8 \\
    \hline
    \textbf{50 (the default config)} & 207 & 1 & 1,637.4 \\
    \hline
    \textbf{100} & 207 & 1 & 3,563.46 \\
    \hline
  \end{tabular}
  }
\end{table}

\begin{table}[t]
\centering
\caption{Table: Information leakage sites due to secret-dependent control branches. False positives are marked as \textcolor{red}{red}.}
\label{tab:sec-dep-control}
\resizebox{\linewidth}{!}{
\begin{tabular}{c@{~}|@{~}|c@{~}|c@{~}|c@{~}|c@{~}|c@{~}}
\hline
\textbf{Library} & \textbf{Algorithm} & \textbf{File} & \textbf{line number} & \textbf{Function} & \textbf{Leakage Units} \\
\hline
mbedTLS 2.5.1 & RSA & bignum.c & 1736 & mbedtls_mpi_exp_mod & 1 \\
\hline
mbedTLS 2.5.1 & RSA & bignum.c & 1739 & mbedtls_mpi_exp_mod & 1 \\
\hline
mbedTLS 2.5.1 & RSA & bignum.c & 1784 & mbedtls_mpi_exp_mod & 2 \\
\hline
mbedTLS 2.5.1 & RSA & bignum.c & 1793 & mbedtls_mpi_exp_mod & 3 \\
\hline
mbedTLS 2.5.1 & RSA & bignum.c & 1128 & mpi_mul_hlp & 4  \\
\hline
mbedTLS 2.5.1 & RSA & bignum.c & 1143 & mpi_mul_hlp & 4 \\
\hline
mbedTLS 2.5.1 & RSA & bignum.c & 1154 & mpi_mul_hlp & 4 \\
\hline
mbedTLS 2.5.1 & RSA & bignum.c & 1167 & mpi_mul_hlp & 4 \\
\hline
\hline
OpenSSL 1.0.2f & RSA/ElGmal & bn_lib.c & 199 & BN_num_bits_word & 1 \\
\hline
OpenSSL 1.0.2f & RSA/ElGmal & bn_lib.c & 200 & BN_num_bits_word & 1 \\
\hline
OpenSSL 1.0.2f & RSA/ElGmal & bn_lib.c & 208 & BN_num_bits_word & 1 \\
\hline
OpenSSL 1.0.2f & RSA/ElGmal & bn_lib.c & 771 & BN_is_bit_set & 2 \\
\hline
OpenSSL 1.0.2f & RSA/ElGmal & bn_lib.c & 775 & BN_is_bit_set & 2 \\
\hline
OpenSSL 1.0.2f & RSA/ElGmal & bn_exp.c & 684 & BN_mod_exp_mont_consttime & 3 \\
\hline
OpenSSL 1.0.2f & RSA/ElGmal & bn_exp.c & 1096 & BN_mod_exp_mont_consttime & 4 \\
\hline
OpenSSL 1.0.2f & RSA/ElGmal & bn_exp.c & 1106 & BN_mod_exp_mont_consttime & 4 \\
\hline
OpenSSL 1.0.2f & RSA/ElGmal & bn_lcl.h & 148 & BN_window_bits_for_exponent_size & 5 \\
\hline
OpenSSL 1.0.2f & RSA/ElGmal & bn_lcl.h & 149 & BN_window_bits_for_exponent_size & 5 \\
\hline                                                                          
OpenSSL 1.0.2f & RSA/ElGmal & bn_lcl.h & 150 & BN_window_bits_for_exponent_size & 5 \\
\hline                                                                         
OpenSSL 1.0.2f & RSA/ElGmal & bn_lcl.h & 151 & BN_window_bits_for_exponent_size & 5 \\
\hline
\hline
OpenSSL 1.0.2k & RSA/ElGmal & bn_lib.c & 199 & BN_num_bits_word & 1 \\
\hline
OpenSSL 1.0.2k & RSA/ElGmal & bn_lib.c & 200 & BN_num_bits_word & 1 \\
\hline
OpenSSL 1.0.2k & RSA/ElGmal & bn_lib.c & 208 & BN_num_bits_word & 1 \\
\hline
OpenSSL 1.0.2k & RSA/ElGmal & bn_lib.c & 771 & BN_is_bit_set & 2 \\
\hline
OpenSSL 1.0.2k & RSA/ElGmal & bn_lib.c & 775 & BN_is_bit_set & 2 \\
\hline
OpenSSL 1.0.2k & RSA/ElGmal & bn_exp.c & 724 & BN_mod_exp_mont_consttime & 3 \\
\hline
OpenSSL 1.0.2k & RSA/ElGmal & bn_exp.c & 1136 & BN_mod_exp_mont_consttime & 4 \\
\hline
OpenSSL 1.0.2k & RSA/ElGmal & bn_exp.c & 1146 & BN_mod_exp_mont_consttime & 4 \\
\hline
OpenSSL 1.0.2k & RSA/ElGmal & bn_lcl.h & 148 & BN_window_bits_for_exponent_size & 5 \\
\hline
OpenSSL 1.0.2k & RSA/ElGmal & bn_lcl.h & 149 & BN_window_bits_for_exponent_size & 5 \\
\hline                                                                         
OpenSSL 1.0.2k & RSA/ElGmal & bn_lcl.h & 150 & BN_window_bits_for_exponent_size & 5 \\
\hline                                                                        
OpenSSL 1.0.2k & RSA/ElGmal & bn_lcl.h & 151 & BN_window_bits_for_exponent_size & 5 \\
\hline
\hline
Libgcrypt 1.6.1 & RSA/ElGmal & secmem.c & 116(1) & ptr_into_pool_p & 1 \\
\hline
Libgcrypt 1.6.1 & RSA/ElGmal & secmem.c & 116(2) & ptr_into_pool_p & 1 \\
\hline
Libgcrypt 1.6.1 & RSA/ElGmal & mpi-pow.c & 615 & _gcry_mpi_powm & 2 \\
\hline
Libgcrypt 1.6.1 & RSA/ElGmal & mpi-pow.c & 670 & _gcry_mpi_powm & 2 \\
\hline
Libgcrypt 1.6.1 & RSA/ElGmal & mpi-pow.c & 704 & _gcry_mpi_powm & 2 \\
\hline
Libgcrypt 1.6.1 & RSA/ElGmal & mpi-pow.c & 706 & _gcry_mpi_powm & 2 \\
\hline
Libgcrypt 1.6.1 & RSA/ElGmal & mpi-pow.c & 769 & _gcry_mpi_powm & 3 \\
\hline
Libgcrypt 1.6.1 & RSA/ElGmal & mpih-mul.c & 493 & _gcry_mpih_mul & 4 \\
\hline                                                           
Libgcrypt 1.6.1 & RSA/ElGmal & mpih-mul.c & 494 & _gcry_mpih_mul & 4 \\
\hline                                                          
Libgcrypt 1.6.1 & RSA/ElGmal & mpih-mul.c & 510 & _gcry_mpih_mul & 4 \\
\hline                                                         
Libgcrypt 1.6.1 & RSA/ElGmal & mpih-mul.c & 512 & _gcry_mpih_mul & 4 \\
\hline
Libgcrypt 1.6.1 & RSA/ElGmal & mpih-mul.c & 83 & mul_n_basecase & 5 \\
\hline
Libgcrypt 1.6.1 & RSA/ElGmal & mpih-mul.c & 84 & mul_n_basecase & 5 \\
\hline
Libgcrypt 1.6.1 & RSA/ElGmal & mpih-mul.c & 100 & mul_n_basecase & 5 \\
\hline
Libgcrypt 1.6.1 & RSA/ElGmal & mpih-mul.c & 102 & mul_n_basecase & 5 \\
\hline
\textcolor{red}{Libgcrypt 1.6.1} & \textcolor{red}{RSA/ElGmal} & \textcolor{red}{mpih-mul.c} & \textcolor{red}{214} & \textcolor{red}{mul_n} & 6 \\
\hline
\textcolor{red}{Libgcrypt 1.6.1} & \textcolor{red}{RSA/ElGmal} & \textcolor{red}{mpih-mul.c} & \textcolor{red}{219} & \textcolor{red}{mul_n} & 6 \\
\hline
Libgcrypt 1.6.1 & RSA/ElGmal & mpi-inline.h & 148 & _gcry_mpih_cmp & 7 \\
\hline
Libgcrypt 1.6.1 & RSA/ElGmal & mpi-inline.h & 157 & _gcry_mpih_cmp & 7 \\
\hline
\textcolor{red}{Libgcrypt 1.6.1} & \textcolor{red}{RSA/ElGmal} & \textcolor{red}{mpi-inline.h} & \textcolor{red}{51} & \textcolor{red}{_gcry_mpih_add_1} & 8 \\
\hline
\textcolor{red}{Libgcrypt 1.6.1} & \textcolor{red}{RSA/ElGmal} & \textcolor{red}{mpi-inline.h} & \textcolor{red}{97} & \textcolor{red}{_gcry_mpih_sub_1} & 9 \\
\hline
\hline
Libgcrypt 1.7.3 & RSA/ElGmal & mpi-internal.h & 116 & MPN_NORMALIZE & 1 \\
\hline
Libgcrypt 1.7.3 & RSA/ElGmal & mpi-pow.c & 609 & _gcry_mpi_powm & 2 \\
\hline
Libgcrypt 1.7.3 & RSA/ElGmal & mpi-pow.c & 680 & _gcry_mpi_powm & 2 \\
\hline
Libgcrypt 1.7.3 & RSA/ElGmal & mpi-pow.c & 706 & _gcry_mpi_powm & 3 \\
\hline
Libgcrypt 1.7.3 & RSA/ElGmal & mpi-pow.c & 724 & _gcry_mpi_powm & 3 \\
\hline
Libgcrypt 1.7.3 & RSA/ElGmal & mpi-pow.c & 780 & _gcry_mpi_powm & 4 \\
\hline
\end{tabular}
}
\end{table}

\begin{table}[t]
\centering
\caption{Table: Information leakage sites due to secret-dependent cache
  accesses. False positives are marked as \textcolor{red}{red}.}
\label{tab:sec-dep-cache-2}
\resizebox{\linewidth}{!}{
\begin{tabular}{c@{~}|@{~}|c@{~}|c@{~}|c@{~}|c@{~}|c@{~}}
\hline
\textbf{Library} & \textbf{Algorithm} & \textbf{File} & \textbf{line number} & \textbf{Function} & \textbf{Leakage Units} \\
\hline
Libgcrypt 1.6.1 & RSA/ElGmal & mpi-pow.c & 677 & _gcry_mpi_powm & 1 \\
\hline                                                          
Libgcrypt 1.6.1 & RSA/ElGmal & mpi-pow.c & 678 & _gcry_mpi_powm & 1 \\
\hline                                                         
Libgcrypt 1.6.1 & RSA/ElGmal & mpi-pow.c & 713 & _gcry_mpi_powm & 2 \\
\hline                                                        
Libgcrypt 1.6.1 & RSA/ElGmal & mpi-pow.c & 714 & _gcry_mpi_powm & 2 \\
\hline
Libgcrypt 1.6.1 & RSA/ElGmal & mpih-mul.c & 82 & mul_n_basecase & 3 \\
\hline
Libgcrypt 1.6.1 & RSA/ElGmal & mpih-mul.c & 99 & mul_n_basecase & 3 \\
\hline
Libgcrypt 1.6.1 & RSA/ElGmal & mpi-internal.h & 88 & MPN_COPY & 4 \\
\hline
Libgcrypt 1.6.1 & RSA/ElGmal & mpi-inline.h & 146 & _gcry_mpih_cmp & 5 \\
\hline
Libgcrypt 1.6.1 & RSA/ElGmal & mpi-inline.h & 147 & _gcry_mpih_cmp & 5 \\
\hline
Libgcrypt 1.6.1 & RSA/ElGmal & mpih-mul.c & 135 & mul_n & 6 \\
\hline
Libgcrypt 1.6.1 & RSA/ElGmal & mpih-mul.c & 137 & mul_n & 6 \\
\hline
Libgcrypt 1.6.1 & RSA/ElGmal & mpih-sub1.S & 81 & _gcry_mpih_sub_n & 7 \\
\hline
Libgcrypt 1.6.1 & RSA/ElGmal & mpih-sub1.S & 82 & _gcry_mpih_sub_n & 7 \\
\hline
Libgcrypt 1.6.1 & RSA/ElGmal & mpih-sub1.S & 84 & _gcry_mpih_sub_n & 7 \\
\hline                                                             
Libgcrypt 1.6.1 & RSA/ElGmal & mpih-sub1.S & 85 & _gcry_mpih_sub_n & 7 \\
\hline                                                            
Libgcrypt 1.6.1 & RSA/ElGmal & mpih-sub1.S & 87 & _gcry_mpih_sub_n & 7 \\
\hline                                                           
Libgcrypt 1.6.1 & RSA/ElGmal & mpih-sub1.S & 88 & _gcry_mpih_sub_n & 7 \\
\hline                                                          
Libgcrypt 1.6.1 & RSA/ElGmal & mpih-sub1.S & 90 & _gcry_mpih_sub_n & 7 \\
\hline
Libgcrypt 1.6.1 & RSA/ElGmal & mpih-sub1.S & 91 & _gcry_mpih_sub_n & 7 \\
\hline
Libgcrypt 1.6.1 & RSA/ElGmal & mpih-sub1.S & 93 & _gcry_mpih_sub_n  & 7 \\
\hline                                                             
Libgcrypt 1.6.1 & RSA/ElGmal & mpih-sub1.S & 94 & _gcry_mpih_sub_n  & 7 \\
\hline                                                            
Libgcrypt 1.6.1 & RSA/ElGmal & mpih-sub1.S & 96 & _gcry_mpih_sub_n  & 7 \\
\hline                                                           
Libgcrypt 1.6.1 & RSA/ElGmal & mpih-sub1.S & 97 & _gcry_mpih_sub_n  & 7 \\
\hline                                                          
Libgcrypt 1.6.1 & RSA/ElGmal & mpih-sub1.S & 99 & _gcry_mpih_sub_n  & 7 \\
\hline                                                         
Libgcrypt 1.6.1 & RSA/ElGmal & mpih-sub1.S & 100 & _gcry_mpih_sub_n & 7  \\
\hline
Libgcrypt 1.6.1 & RSA/ElGmal & mpih-sub1.S & 102 & _gcry_mpih_sub_n & 7 \\
\hline
Libgcrypt 1.6.1 & RSA/ElGmal & mpih-sub1.S & 103 & _gcry_mpih_sub_n & 7 \\
\hline
Libgcrypt 1.6.1 & RSA/ElGmal & mpih-mul1.S & 68 & _gcry_mpih_mul_1 & 8 \\
\hline
Libgcrypt 1.6.1 & RSA/ElGmal & mpih-mul2.S & 69 & _gcry_mpih_addmul_1 & 9 \\
\hline
Libgcrypt 1.6.1 & RSA/ElGmal & mpih-add1.S & 80 & _gcry_mpih_add_n & 10  \\
\hline
Libgcrypt 1.6.1 & RSA/ElGmal & mpih-add1.S & 83 & _gcry_mpih_add_n & 10 \\
\hline
Libgcrypt 1.6.1 & RSA/ElGmal & mpih-add1.S & 86 & _gcry_mpih_add_n & 10 \\
\hline                                                            
Libgcrypt 1.6.1 & RSA/ElGmal & mpih-add1.S & 89 & _gcry_mpih_add_n & 10 \\
\hline                                                           
Libgcrypt 1.6.1 & RSA/ElGmal & mpih-add1.S & 92 & _gcry_mpih_add_n & 10 \\
\hline                                                             
Libgcrypt 1.6.1 & RSA/ElGmal & mpih-add1.S & 95 & _gcry_mpih_add_n & 10 \\
\hline                                                            
Libgcrypt 1.6.1 & RSA/ElGmal & mpih-add1.S & 98 & _gcry_mpih_add_n & 10 \\
\hline                                                           
Libgcrypt 1.6.1 & RSA/ElGmal & mpih-add1.S & 101 & _gcry_mpih_add_n & 10  \\
\hline
Libgcrypt 1.6.1 & RSA/ElGmal & mpih-mul.c & 492 & _gcry_mpih_mul & 11 \\
\hline
Libgcrypt 1.6.1 & RSA/ElGmal & mpih-mul.c & 495 & _gcry_mpih_mul & 11 \\
\hline
Libgcrypt 1.6.1 & RSA/ElGmal & mpih-mul.c & 509 & _gcry_mpih_mul & 11 \\
\hline
\hline
mbedTLS 2.5.1 & RSA & bignum.c & 1563(1)--(2) & mpi_montmul & 1 \\
\hline
mbedTLS 2.5.1 & RSA & bignum.c & 1571 & mpi_montmul & 1 \\
\hline
mbedTLS 2.5.1 & RSA & bignum.c & 1573 & mpi_montmul & 1 \\
\hline
mbedTLS 2.5.1 & RSA & bignum.c & 1131(1)--(2) & mpi_mul_hlp & 2 \\
\hline
mbedTLS 2.5.1 & RSA & bignum.c & 1132(1)--(2) & mpi_mul_hlp & 2 \\
\hline                                        
mbedTLS 2.5.1 & RSA & bignum.c & 1133(1)--(2) & mpi_mul_hlp & 2 \\
\hline                                       
mbedTLS 2.5.1 & RSA & bignum.c & 1134(1)--(2) & mpi_mul_hlp & 2 \\
\hline                                        
mbedTLS 2.5.1 & RSA & bignum.c & 1136(1)--(2) & mpi_mul_hlp & 2 \\
\hline                                       
mbedTLS 2.5.1 & RSA & bignum.c & 1137(1)--(2) & mpi_mul_hlp & 2 \\
\hline                                        
mbedTLS 2.5.1 & RSA & bignum.c & 1138(1)--(2) & mpi_mul_hlp & 2 \\
\hline                                       
mbedTLS 2.5.1 & RSA & bignum.c & 1139(1)--(2) & mpi_mul_hlp & 2 \\
\hline                                        
mbedTLS 2.5.1 & RSA & bignum.c & 1146(1)--(2) & mpi_mul_hlp & 2 \\
\hline                                       
mbedTLS 2.5.1 & RSA & bignum.c & 1147(1)--(2) & mpi_mul_hlp & 2 \\
\hline                                        
mbedTLS 2.5.1 & RSA & bignum.c & 1149(1)--(2) & mpi_mul_hlp & 2 \\
\hline                                        
mbedTLS 2.5.1 & RSA & bignum.c & 1150(1)--(2) & mpi_mul_hlp & 2 \\
\hline
mbedTLS 2.5.1 & RSA & bignum.c & 1157 & mpi_mul_hlp & 2 \\
\hline
\hline
OpenSSL 1.0.2f & RSA/ElGmal & bn_lib.c & 201 & BN_num_bits_word & 1 \\
\hline                                                         
OpenSSL 1.0.2f & RSA/ElGmal & bn_lib.c & 203 & BN_num_bits_word & 1 \\
\hline                                                         
OpenSSL 1.0.2f & RSA/ElGmal & bn_lib.c & 209 & BN_num_bits_word & 1 \\
\hline                                                         
OpenSSL 1.0.2f & RSA/ElGmal & bn_lib.c & 212 & BN_num_bits_word & 1 \\
\hline
OpenSSL 1.0.2f & RSA/ElGmal & bn_lib.c & 777 & BN_is_bit_set & 2 \\
\hline
\textcolor{red}{OpenSSL 1.0.2f} & \textcolor{red}{RSA/ElGmal} & \textcolor{red}{bn_exp.c} & \textcolor{red}{633} & \textcolor{red}{MOD_EXP_CTIME_COPY_FROM_PREBUF} & 3 \\
\hline
\hline
OpenSSL 1.0.2k & RSA/ElGmal & bn_lib.c & 201 & BN_num_bits_word & 1 \\
\hline                                                         
OpenSSL 1.0.2k & RSA/ElGmal & bn_lib.c & 203 & BN_num_bits_word & 1 \\
\hline                                                        
OpenSSL 1.0.2k & RSA/ElGmal & bn_lib.c & 209 & BN_num_bits_word & 1 \\
\hline                                                         
OpenSSL 1.0.2k & RSA/ElGmal & bn_lib.c & 212 & BN_num_bits_word & 1 \\
\hline
OpenSSL 1.0.2k & RSA/ElGmal & bn_exp.c & 777 & BN_is_bit_set & 2 \\
\hline
\hline
mbedTLS 2.5.1 & AES & aes.c & 788(1)--(16) & mbedtls_internal_aes_decrypt & 1 \\
\hline                                                          
mbedTLS 2.5.1 & AES & aes.c & 789(1)--(16) & mbedtls_internal_aes_decrypt & 1 \\
\hline                                                         
mbedTLS 2.5.1 & AES & aes.c & 792(1)--(16) & mbedtls_internal_aes_decrypt & 1 \\
\hline                                                        
mbedTLS 2.5.1 & AES & aes.c & 795 & mbedtls_internal_aes_decrypt & 1 \\
\hline                                                       
mbedTLS 2.5.1 & AES & aes.c & 796 & mbedtls_internal_aes_decrypt & 1 \\
\hline                                                      
mbedTLS 2.5.1 & AES & aes.c & 797 & mbedtls_internal_aes_decrypt & 1 \\
\hline                                                     
mbedTLS 2.5.1 & AES & aes.c & 798 & mbedtls_internal_aes_decrypt & 1 \\
\hline                                                    
mbedTLS 2.5.1 & AES & aes.c & 801 & mbedtls_internal_aes_decrypt & 1 \\
\hline                                                   
mbedTLS 2.5.1 & AES & aes.c & 802 & mbedtls_internal_aes_decrypt & 1 \\
\hline                                                  
mbedTLS 2.5.1 & AES & aes.c & 803 & mbedtls_internal_aes_decrypt & 1 \\
\hline                                                 
mbedTLS 2.5.1 & AES & aes.c & 804 & mbedtls_internal_aes_decrypt & 1 \\
\hline                                                
mbedTLS 2.5.1 & AES & aes.c & 807 & mbedtls_internal_aes_decrypt & 1 \\
\hline                                               
mbedTLS 2.5.1 & AES & aes.c & 808 & mbedtls_internal_aes_decrypt & 1 \\
\hline                                              
mbedTLS 2.5.1 & AES & aes.c & 809 & mbedtls_internal_aes_decrypt & 1 \\
\hline                                             
mbedTLS 2.5.1 & AES & aes.c & 810 & mbedtls_internal_aes_decrypt & 1 \\
\hline                                            
mbedTLS 2.5.1 & AES & aes.c & 813 & mbedtls_internal_aes_decrypt & 1 \\
\hline                                           
mbedTLS 2.5.1 & AES & aes.c & 814 & mbedtls_internal_aes_decrypt & 1 \\
\hline                                          
mbedTLS 2.5.1 & AES & aes.c & 815 & mbedtls_internal_aes_decrypt & 1 \\
\hline                                         
mbedTLS 2.5.1 & AES & aes.c & 816 & mbedtls_internal_aes_decrypt & 1 \\
\hline
\hline
OpenSSL 1.0.2f & AES & aes-586.pl & 1357(1)--(4) & _x86_AES_decrypt_compact & 1 \\
\hline                                                            
OpenSSL 1.0.2f & AES & aes-586.pl & 1358(1)--(4) & _x86_AES_decrypt_compact & 1 \\
\hline                                                           
OpenSSL 1.0.2f & AES & aes-586.pl & 1359(1)--(4) & _x86_AES_decrypt_compact & 1 \\
\hline                                                          
OpenSSL 1.0.2f & AES & aes-586.pl & 1360(1)--(4) & _x86_AES_decrypt_compact & 1 \\
\hline                                                         
OpenSSL 1.0.2f & AES & aes-586.pl & 1377(1)--(4) & _x86_AES_decrypt_compact & 1 \\
\hline                                                        
OpenSSL 1.0.2f & AES & aes-586.pl & 1378(1)--(4) & _x86_AES_decrypt_compact & 1 \\
\hline                                                       
OpenSSL 1.0.2f & AES & aes-586.pl & 1379(1)--(4) & _x86_AES_decrypt_compact & 1 \\
\hline                                                      
OpenSSL 1.0.2f & AES & aes-586.pl & 1380(1)--(4) & _x86_AES_decrypt_compact & 1 \\
\hline
\hline
OpenSSL 1.0.2k & AES & aes-586.pl & 1357(1)--(4) & _x86_AES_decrypt_compact & 1 \\
\hline                                                                     
OpenSSL 1.0.2k & AES & aes-586.pl & 1358(1)--(4) & _x86_AES_decrypt_compact & 1 \\
\hline                                                                    
OpenSSL 1.0.2k & AES & aes-586.pl & 1359(1)--(4) & _x86_AES_decrypt_compact & 1 \\
\hline                                                                   
OpenSSL 1.0.2k & AES & aes-586.pl & 1360(1)--(4) & _x86_AES_decrypt_compact & 1 \\
\hline                                                                  
OpenSSL 1.0.2k & AES & aes-586.pl & 1377(1)--(4) & _x86_AES_decrypt_compact & 1 \\
\hline                                                                 
OpenSSL 1.0.2k & AES & aes-586.pl & 1378(1)--(4) & _x86_AES_decrypt_compact & 1 \\
\hline                                                                
OpenSSL 1.0.2k & AES & aes-586.pl & 1379(1)--(4) & _x86_AES_decrypt_compact & 1 \\
\hline                                                               
OpenSSL 1.0.2k & AES & aes-586.pl & 1380(1)--(4) & _x86_AES_decrypt_compact & 1 \\
\hline
\end{tabular}
}
\end{table}

\vspace{-5pt}
\begin{table*}[t]
\captionsetup{font=small}
  \caption{Evaluation result overview when the parameter $\text{N}$ of the
    $\text{BOU}$ function is 1. All the analysis campaigns are terminated before
    reaching the fixed point due to memory write accesses through the public
    symbol $p$. See the last paragraph of Section 6.2 ``Information Flow'' for
    discussion of such cases.}
  \label{tab:size-1}
  \centering
  \resizebox{0.9\linewidth}{!}{
 \begin{tabular}{c@{~}|@{~}c@{~}|@{~}c@{~}|@{~}c@{~}|@{~}c@{~}|@{~}c@{~}|@{~}c@{~}@{~}c@{~}|@{~}c@{~}}
    \hline
    \multirow{2}{*}{\textbf{Algorithm}} & \multirow{2}{*}{\textbf{Implementation}} & \textbf{Processing Time} & \textbf{\# of Analyzed} & \textbf{\# of Analyzed} & \textbf{\# of Analyzed} & \textbf{Peak Memory}  && \textbf{Information Leakage Sites} \\
    & &  \textbf{Procedures} & \textbf{Contexts} & \textbf{(CPU Seconds)}  & \textbf{REIL Instructions} & \textbf{Usage (MB)} & & \textbf{(known/unknown)} \\
    \hline
    \textbf{RSA/Elgamal} & libgcrypt 1.6.1 & \multicolumn{5}{c}{The analysis is terminated before reaching the fixed point due to a memory write access through $p$.} && NA \\
    \hline
    \textbf{RSA/Elgamal} & libgcrypt 1.7.3 & \multicolumn{5}{c}{The analysis is terminated before reaching the fixed point due to a memory write access through $p$.} && NA \\
    \hline
    \textbf{RSA/Elgamal} & OpenSSL 1.0.2k & \multicolumn{5}{c}{The analysis is terminated before reaching the fixed point due to a memory write access through $p$.} && NA \\
    \hline
    \textbf{RSA/Elgamal} & OpenSSL 1.0.2f & \multicolumn{5}{c}{The analysis is terminated before reaching the fixed point due to a memory write access through $p$.}  && NA \\
    \hline
    \textbf{RSA} & mbedTLS 2.5.1 & \multicolumn{5}{c}{The analysis is terminated before reaching the fixed point due to a memory write access through $p$.}  && NA \\ 
    \hline
    \textbf{AES} & OpenSSL 1.0.2k & \multicolumn{5}{c}{The analysis is terminated before reaching the fixed point due to a memory write access through $p$.}  && NA \\
    \hline
    \textbf{AES} & OpenSSL 1.0.2f  & \multicolumn{5}{c}{The analysis is terminated before reaching the fixed point due to a memory write access through $p$.} && NA \\
    \hline
    \textbf{AES} & mbedTLS 2.5.1 &  \multicolumn{5}{c}{The analysis is terminated before reaching the fixed point due to a memory write access through $p$.}  && NA \\
    \hline
  \end{tabular}
 }
\end{table*}

\vspace{-5pt}
\begin{table*}[t]
\captionsetup{font=small}
  \caption{Evaluation result overview when the parameter $\text{N}$ of the
    $\text{BOU}$ function is 10. Two analysis campaigns are terminated before
    reaching the fixed point due to memory write accesses through the public
    symbol $p$. See the last paragraph of Section 6.2 ``Information Flow'' for
    discussion of such cases. Variances in terms of analyzed contexts are
    highlighted.}
  \label{tab:size-10}
  \centering
  \resizebox{0.9\linewidth}{!}{
 \begin{tabular}{c@{~}|@{~}c@{~}|@{~}c@{~}|c|@{~}c@{~}|@{~}c@{~}|@{~}c@{~}@{~}c@{~}|@{~}c@{~}}
    \hline
    \multirow{2}{*}{\textbf{Algorithm}} & \multirow{2}{*}{\textbf{Implementation}}  & \textbf{\# of Analyzed} & \textbf{\# of Analyzed} & \textbf{Processing Time} & \textbf{\# of Analyzed} & \textbf{Peak Memory}  && \textbf{Information Leakage Sites} \\
    & &  \textbf{Procedures} & \textbf{Contexts} & \textbf{(CPU Seconds)}  & \textbf{REIL Instructions} & \textbf{Usage (MB)} & & \textbf{(known/unknown)} \\
    \hline
    \textbf{RSA/Elgamal} & libgcrypt 1.6.1 & \multicolumn{5}{l}{The analysis is terminated before reaching the fixed point due to a memory write access through $p$.} && NA  \\ 
    \hline
    \textbf{RSA/Elgamal} & libgcrypt 1.7.3 & \multicolumn{5}{c}{The analysis is terminated before reaching the fixed point due to a memory write access through $p$.} && NA  \\
    \hline
    \textbf{RSA/Elgamal} & OpenSSL 1.0.2k & 71 & 81 & 168.67  & 83,183  & 4,237 && 2/3 \\
    \textbf{RSA/Elgamal} & OpenSSL 1.0.2f & 68  & 72 & 150.3  & 80,096  & 4,151  && 2/4 \\
    \textbf{RSA} & mbedTLS 2.5.1 & 29  & \cellcolor{blue!25}34 & 125.34  & 34,137  & 3,501  && 0/29 \\
    \hline
    \textbf{AES} & OpenSSL 1.0.2k  & 1  & 1 & 46.76  & 3,748  & 566 && 32/0 \\
    \textbf{AES} & OpenSSL 1.0.2f  & 1  & 1 & 47.91 & 3,748  & 574 && 32/0 \\
    \textbf{AES} & mbedTLS 2.5.1   & 1  & 1 & 45.52 & 4,803  & 569 && 64/0 \\
    \hline
    \textbf{Total} & & 191 & 190 & 584.5 & 209,715 & 13,598 && 132/36 \\
    \hline
  \end{tabular}
 }
\end{table*}

\vspace{-5pt}
\begin{table*}[t]
\captionsetup{font=small}
  \caption{Evaluation result overview when the parameter $\text{N}$ of the
    $\text{BOU}$ function is 25. Variances in terms of analyzed contexts
    are highlighted.}
  \label{tab:size-25}
  \centering
  \resizebox{0.9\linewidth}{!}{
 \begin{tabular}{c@{~}|@{~}c@{~}|@{~}c@{~}|c|@{~}c@{~}|@{~}c@{~}|@{~}c@{~}|@{~}c@{~}}
    \hline
    \multirow{2}{*}{\textbf{Algorithm}} & \multirow{2}{*}{\textbf{Implementation}} & \textbf{\# of Analyzed} & \textbf{\# of Analyzed} & \textbf{Processing Time} & \textbf{\# of Analyzed} & \textbf{Peak Memory} & \textbf{Information Leakage Sites} \\
    & &  \textbf{Procedures} & \textbf{Contexts} & \textbf{(CPU Seconds)}  & \textbf{REIL Instructions} & \textbf{Usage (MB)} & \textbf{(known/unknown)} \\
    \hline
    \textbf{RSA/Elgamal} & libgcrypt 1.6.1 & 60  & 81 & 351.69  & 50,436   & 6,956 & 22/18 \\
    \textbf{RSA/Elgamal} & libgcrypt 1.7.3 & 59  & 59 & 211.72 & 33,386  & 6,385 & 0/0 \\
    \textbf{RSA/Elgamal} & OpenSSL 1.0.2k  & 71   & 81 & 199.41  & 83,183   & 5,100 &  2/3 \\
    \textbf{RSA/Elgamal} & OpenSSL 1.0.2f  & 68  & 72 & 165.35 & 80,096   & 5,345 &  2/4 \\
    \textbf{RSA} & mbedTLS 2.5.1 & 29  & \cellcolor{blue!25}35 & 378.98 & 35,050  & 6,347 & 0/29  \\
     \hline
    \textbf{AES} & OpenSSL 1.0.2k  & 1  & 1 & 47.29 & 3,748  & 557 & 32/0 \\
    \textbf{AES} & OpenSSL 1.0.2f  & 1  & 1 & 46.46 & 3,748  & 628 & 32/0 \\
    \textbf{AES} & mbedTLS 2.5.1  & 1  & 1 & 45.94  & 4,803  & 633  & 64/0\\
    \hline
    \textbf{Total} & & 290 & 331 & 1,446.84 & 294,450 & 31,951 & 154/54 \\
    \hline
  \end{tabular}
 }
\end{table*}

\vspace{-5pt}
\begin{table*}[t]
\captionsetup{font=small}
  \caption{Evaluation result overview when the parameter $\text{N}$ of the
    $\text{BOU}$ function is 50 (configuration used in our evaluation).
    Variances in terms of analyzed contexts are highlighted.}
  \label{tab:size-50}
  \centering
  \resizebox{0.9\linewidth}{!}{
 \begin{tabular}{c@{~}|@{~}c@{~}|@{~}c@{~}|c|@{~}c@{~}|@{~}c@{~}|@{~}c@{~}|@{~}c@{~}}
    \hline
    \multirow{2}{*}{\textbf{Algorithm}} & \multirow{2}{*}{\textbf{Implementation}} & \textbf{\# of Analyzed} & \textbf{\# of Analyzed} & \textbf{Processing Time} & \textbf{\# of Analyzed} & \textbf{Peak Memory} & \textbf{Information Leakage Sites} \\
    & &  \textbf{Procedures} & \textbf{Contexts} & \textbf{(CPU Seconds)}  & \textbf{REIL Instructions} & \textbf{Usage (MB)} & \textbf{(known/unknown)} \\
    \hline
     \textbf{RSA/ElGamal} & libgcrypt 1.6.1  & 60  & 81 & 228.8  & 50,436   & 7,749 &  22/18 \\
    
    \textbf{RSA/ElGamal} & libgcrypt 1.7.3  & 59  & 59 & 182.2 & 33,386  & 5,823 & 0/0  \\

    \textbf{RSA/ElGamal} & OpenSSL 1.0.2k   & 71   & 81 & 179.2   & 83,183   & 6,134 &  2/3 \\

    \textbf{RSA/ElGamal} & OpenSSL 1.0.2f  & 68   & 72 & 169.5  & 80,096   & 6,113  &  2/4  \\

    \textbf{RSA} & mbedTLS 2.5.1  & 29  & \cellcolor{blue!25}36 & 775.9 & 35,963  & 9,654 & 0/29 \\
    \hline
    \textbf{AES} & OpenSSL 1.0.2k  & 1  & 1 & 33.2 & 3,748  & 620 & 32/0 \\
    \textbf{AES} & OpenSSL 1.0.2f  & 1  & 1 & 35.8 & 3,748  & 578 & 32/0 \\
    \textbf{AES} & mbedTLS 2.5.1  & 1  & 1 & 32.8 & 4,803  & 619 & 64/0 \\
    \hline
    \textbf{Total} & & 290 & 332 & 1,637.4 & 295,363 & 37,290 & 154/54 \\
    \hline
  \end{tabular}
 }
\end{table*}

\vspace{-5pt}
\begin{table*}[t]
\captionsetup{font=small}
  \caption{Evaluation result overview when the parameter $\text{N}$ of the
    $\text{BOU}$ function is 100. Variances in terms of analyzed contexts
    are highlighted.}
  \label{tab:size-100}
  \centering
  \resizebox{0.9\linewidth}{!}{
 \begin{tabular}{c@{~}|@{~}c@{~}|@{~}c@{~}|c|@{~}c@{~}|@{~}c@{~}|@{~}c@{~}|@{~}c@{~}}
    \hline
    \multirow{2}{*}{\textbf{Algorithm}} & \multirow{2}{*}{\textbf{Implementation}} & \textbf{\# of Analyzed} & \textbf{\# of Analyzed} & \textbf{Processing Time}  & \textbf{\# of Analyzed} & \textbf{Peak Memory}  & \textbf{Information Leakage Sites} \\
    & &  \textbf{Procedures} & \textbf{Contexts} & \textbf{(CPU Seconds)}  & \textbf{REIL Instructions} & \textbf{Usage (MB)} & \textbf{(known/unknown)} \\
    \hline
    \textbf{RSA/Elgamal} & libgcrypt 1.6.1 & 60  & 81 & 335.08  & 50,436   & 6,750 &  22/18 \\
    \textbf{RSA/Elgamal} & libgcrypt 1.7.3  & 59 & 59 & 226.16  & 33,386  & 7,339 & 0/0 \\
    \textbf{RSA/Elgamal} & OpenSSL 1.0.2k  & 71  & 81 & 358.02  & 83,183   & 8,431 &  2/3 \\
    \textbf{RSA/Elgamal} & OpenSSL 1.0.2f  & 68  & 72 & 314.53  & 80,096   & 8,549 &  2/4 \\
    \textbf{RSA} & mbedTLS 2.5.1   & 29 & \cellcolor{blue!25}36 & 2,189.69   & 35,963  & 22,140 & 0/29 \\
    \hline
    \textbf{AES} & OpenSSL 1.0.2k   & 1 & 1 & 47.02 & 3,748  & 567 & 32/0 \\
    \textbf{AES} & OpenSSL 1.0.2f   & 1 & 1 & 45.57 & 3,748  & 590 & 32/0 \\
    \textbf{AES} & mbedTLS 2.5.1   & 1 & 1 & 47.39  & 4,803  & 631 & 64/0 \\
    \hline
    \textbf{Total} & & 290 & 332 & 3,563.46 & 259,363 & 54,997 & 154/54 \\
    \hline
  \end{tabular}
 }
\end{table*}

\end{appendix}

\end{document}